\documentclass[journal, onecolumn]{IEEEtran}
\usepackage{hyperref}
\usepackage{epsf,pgf,xcolor}
\usepackage{epstopdf,epsfig}
\usepackage[utf8]{inputenc}
\usepackage[english]{babel}
\usepackage{amsthm}

\usepackage{graphicx}% Include figure files
\graphicspath{ {C:/Users/alireza/Desktop/} }
\usepackage{dcolumn}% Align table columns on decimal point
\usepackage{bm}% bold math
\usepackage[export]{adjustbox}
\usepackage{array}
\usepackage{fullpage}
\usepackage{algorithm} 
\usepackage{algpseudocode} 
\usepackage{times}
\usepackage{hyperref}
\usepackage{fancyhdr,graphicx,amsmath,amssymb}
\usepackage{caption}
\usepackage[utf8]{inputenc} 
\usepackage[T1]{fontenc}
\usepackage{url}
\usepackage{ifthen}
\usepackage[english]{babel}
\usepackage{amsthm}
\usepackage{caption}
\usepackage{cite}
\usepackage[english]{babel}
\usepackage{amssymb}
\usepackage[utf8]{inputenc}
\newtheorem{theorem}{Theorem}
\newtheorem{lemma}{Lemma}

\newtheorem{example}{Example}
\newtheorem{corollary}{Corollary}

\graphicspath{ {./main/} } 
\usepackage{epstopdf}
\usepackage{verbatim}
\providecommand{\keywords}[1]
{
  \textbf{\textit{Index Terms---}} #1
}
\begin{document}
\title{Age of Information for Multiple-Source Multiple-Server Networks} 

% %%% Single author, or several authors with same affiliation:
% \author{%
%   \IEEEauthorblockN{Stefan M.~Moser}
%   \IEEEauthorblockA{ETH Zürich\\
%                     ISI (D-ITET)\\
%                     CH-8092 Zürich, Switzerland\\
%                     Email: moser@isi.ee.ethz.ch}
% }

%%% Several authors with up to three affiliations:
%\author{%
%  \IEEEauthorblockN{Alireza Javani}
%  \IEEEauthorblockA{Center for Pervasive Communications and Computing\\
% 					University of California, Irvine\\
%                  Email: ajavani@uci.edu}
%  \and
% \IEEEauthorblockN{Zhiying Wang }
% \IEEEauthorblockA{Center for Pervasive Communications and Computing\\
%  					University of California, Irvine\\             
%                    Email: zhiying@uci.edu}
%}

%%% Many authors with many affiliations:
 \author{%
   \IEEEauthorblockN{Alireza Javani, Student Member, IEEE, Marwen Zorgui, Student Member, IEEE, Zhiying Wang, Member, IEEE}
  
 }

\maketitle

%%%%%%
%% Abstract: 
%% If your paper is eligible for the student paper award, please add
%% the comment "THIS PAPER IS ELIGIBLE FOR THE STUDENT PAPER
%% AWARD." as a first line in the abstract. 
%% For the final version of the accepted paper, please do not forget
%% to remove this comment!
%%

\begin{abstract}

Having timely and fresh knowledge about the current state of information sources is critical in a variety of applications. In particular, a status update may arrive at the destination later than its generation time due to processing and communication delays. The freshness of the status update at the destination is captured by the notion of age of information. In this study, we analyze a multiple sensing network with multiple  sources, multiple servers, and a monitor (destination). Each source corresponds to an independent piece of information and its age is measured individually. Given a particular source, the servers independently sense the source of information and send the status update to the monitor. We assume that updates arrive at the servers according to  Poisson random processes. Each server sends its updates to the monitor through a direct link, which is modeled as a queue. The service time to transmit an update is considered to be an exponential random variable. 
We examine both homogeneous and heterogeneous service and arrival rates for the single-source case, and only homogeneous arrival and service rates for the multiple-source case. 
We derive a closed-form expression for the average age of information under a last-come-first-serve (LCFS) queue for a single source and arbitrary number of homogeneous servers. Using a recursive method, we derive the explicit average age of information for any number of sources and homogeneous servers. We also investigate heterogeneous servers and a single source, and present algorithms for finding the average age of information. %We solve this general case using the introduced algorithms and prove that the algorithms find the average AoI. Also, for $n=2$ servers, we find the optimal arrival rates given fixed sum arrival rates and service rates.
\end{abstract}

%% The paper must be self-contained. However, if you are referring to
%% a full version for checking certain proofs, please provide the
%% publically accessible location below.  If the paper is completely
%% self-contained, you can remove the following line from your
%% submission.
%%\textit{A full version of this paper is accessible at:}
%%\url{http://isit2019.fr/} 
\section{Introduction}
\let\thefootnote\relax\footnotetext{This
paper was presented in part at the 2019 IEEE Global Communications Conference (Globecom).
The authors are with the Center for Pervasive Communications and Computing, University of California Irvine (e-mail:
ajavani@uci.edu,
mzorgui@uci.edu, zhiying@uci.edu).}

Widespread sensor network applications such as health monitoring using wireless sensors \cite{8908653} and the Internet of things (IoT)\cite{10.1145/3368089.3409682}, as well as applications like stock market trading and vehicular networks \cite{du2015effective}, require sending several status updates to their designated recipients (called monitors). Outdated information in the monitoring facility may lead to undesired situations. As a result, having the data at the monitor as fresh as possible is crucial. 
%In order to have a sense of the freshness of the received status update, the age of information
In order to quantify the freshness of the received status update, the age of information (AoI) metric was introduced in~\cite{kaul2012real}. For an update received by the monitor, AoI is defined as the time elapsed since the generation of the update. AoI captures the timeliness of status updates, which is different from other standard communication metrics like delay and throughput. It is affected by the inter-arrival time of updates and the delay that is caused by queuing during update processing and transmission. \begin{comment}
Another age-related metric of peak AoI was later introduced in \cite{costa2016age}, which corresponds to the age of information at the monitor right before the receipt of the next update. 
The average peak AoI minimization in IoT networks and wireless systems was considered in \cite{abd2018average, he2016optimal}.
\end{comment}

Instead of sensing the source by one server, we consider the \emph{multiple sensing} problem, where  updates arrive at sensors and are sent to the receiver (monitoring facility) through multiple servers. In this work we study the settings of  homogeneous and heterogeneous arrival and service rates and extend our previous results in \cite{9013935}.
We study the average age of information defined as in~\cite{kaul2012real}. 
We consider AoI in a multiple sensing network and assume that a number of shared sources are sensed and then the data is transmitted to the monitor by $n$ independent servers. For example, the sources of information could be some shared environmental parameters, and independently operated sensors in the surrounding area obtain such information.
As another example, the source of information can be the prices of several stocks which is transmitted to the user by multiple independent service providers. Throughout this paper, a sensor or a service provider is called a server, since it is responsible to serve the updates to the monitor. In this paper we aim to answer the question \emph{how much  gain in terms of AoI we can get using multiple servers}.

We assume that status updates arrive at the servers independently according to Poisson random processes, and the server is modeled as a queue whose service time for an update is exponentially distributed. %We first consider the single source case, and then generalize to the model of multiple sources. 
We assume information sources are independent and are sensed by $n$ independent servers. We mainly consider the Last-Come-First-Serve with preemption in service  (in short, LCFS) queue model, namely, upon the arrival of a new update, the server immediately starts to serve it and drops any old update being served.

In summary, this paper makes the following main contributions:
\begin{itemize}
    \item We propose the multiple-sensing network for updating information of multiple sources. Depending on the information arrival rates and the service rates, the network is categorized as homogeneous or heterogeneous. The stochastic hybrid system (SHS) is established for various cases  to derive the average AoI similar to \cite{yates2018status , yates2018age}. 
    \item A closed-form expression of the average AoI for a single-source multiple-server network under LCFS policy is derived.
    \item We develop a recursive algorithm that calculates the average AoI for LCFS with multiple sources and  multiple servers in a homogeneous network. Moreover, closed-form AoI expressions are derived for an arbitrary number of sources and $n=2,3$ servers. 
    \item The heterogeneous network with a single source is considered. For the cases of $n=2,3$, the expressions for the average AoI are developed. For the general case, an algorithm is developed for computing the average AoI.
    \item Simulations are carried out for different queue models and network setups.
\end{itemize}
 %Moreover, LCFS, LCFS with preemption in waiting, and FCFS queue models are compared in terms of the average AoI through simulation.
%We also provide optimal arrival rate for a given service rate in the FCFS case numerically. 
% For multiple source, recursive algorithms are shown to calculate AoI.

{\bf Related work.}
In \cite{kaul2012real}, the authors considered the single-source single-server and first-come-first-serve (FCFS) queue model and determined the arrival rate that minimizes AoI. A series of works afterwards investigated average AoI minimization under various system models with multiple sources, servers and different queue models. Different cases of multiple-source single-server under FCFS and last-come-first-serve (LCFS) were considered in \cite{yates2018age}, \cite{9099557} and the region of feasible age was derived. In \cite{yates2018status, 8406966}, the system is modeled as a source that submits status updates to a network of parallel and serial servers, respectively, for delivery to a monitor and AoI is evaluated. The parallel-server network is also studied in \cite{kam2016effect} when the number of servers is 2 or infinite, and the average AoI for the FCFS queue model was derived. The authors in \cite{doncel2020age} also 
considered a system with multiple sources, where packets are sent to the parallel queues. They compute the average AoI of a system with only two parallel servers and compare the average AoI with the case of a single queue. In \cite{8433704}, the authors considered a model with multiple sources, a single queue and multiple destinations. A real time monitoring system where IoT devices have to transmit status updates to a common destination is considered in \cite{zhou2020age}. The authors considered correlated status updates at the devices and showed that the optimal policy is threshold-based with respect to AoI at the destination.

The AoI has also been applied to different network models as a  performance metric for various communication systems that deal with time-sensitive information,
e.g., cellular wireless networks \cite{9174054,9148771,mankar2020throughput,8761106}, source nodes powered by energy harvesting \cite{8849636,8437547,9154225,9086254,9086214,gindullina2020ageofinformation,9217302}, wireless erasure networks and coding \cite{9149052,9013381,9144090,8636088,9174086}, scheduling in networks \cite{9374448,9097584,9348053,9055353,9184001,feng2020pecoding}, unmanned aerial vehicle (UAV)-assisted communication
systems \cite{8570843,9285214,9374461}, and multi-hop networks \cite{9048736,8732378,9149009,8262777}. In particular, the goal of this line of research  is to identify the characteristics of the optimal policies that minimize the average AoI.
Another age-related metric of peak AoI was also introduced in \cite{costa2016age}, which corresponds to the age of information at the monitor right before the receipt of the next update. 
The average peak AoI minimization in IoT networks and wireless systems was considered in \cite{abd2018average, 9360520,dogan2020multisource,9061059,9322606}.

This paper is organized as follows. Section \ref{sec:preliminaries} formally introduces the system model of interest, and provides preliminaries on SHS. Section \ref{sec:homo} studies the average AoI for homogeneous servers. In Section \ref{LCFS}, we derive the average age of information formula by applying the SHS method to our model for a homogeneous network with a single information source. In Section \ref{multiple}  we derive AoI for an arbitrary number of information sources and for any $n$. We also obtain the optimal arrival rates when $n=2$ that minimizes the weighted sum of average AoI. In Section \ref{hetro-sec}, we investigate the heterogeneous network and prove that the average AoI can be computed using our proposed algorithms. When $n=2$, we find the optimal arrival rate at each server given the service rates. In the end, we discuss our findings, future directions, and conclusion in Section \ref{conc}.

{\bf Notation.} In this paper, we use boldface for vectors, and normal font with a subscript for its elements. For example, for a vector $\mathbf{x}$, the $j$-th element is denoted by $x_j$. For non-negative integers $a$ and $b$, $b \geq a$, we define $[a:b] \triangleq  \{a,  \ldots, b\}$, and $[a]\triangleq  [1:a]$. If $a > b$, $[a:b] \triangleq \emptyset$.

\section{System Model and Preliminaries}
\label{sec:preliminaries}
In this section, we first present our network model, and then briefly review the stochastic hybrid system analysis from \cite{yates2018age}.
The network consists of $m$ information sources that are sensed by $n$ independent servers as illustrated in Figure \ref{fig2}. Updates from the information sources are aggregated at the monitor  after going through separate links. Server $j$ collects updates from source $i$ following a Poisson random process with rate $\lambda_{j}^{(i)}$, $j \in [n], i \in [m]$. For Server $j$, the service time is an exponential random variable with average $\frac{1}{\mu_{j}}$, independent of all other servers. We focus on the queuing model of \emph{last-come-first-come with preemption in service}, or in short, \emph{LCFS}. In this model, a server starts to transmit the new update right upon its arrival, thus dropping the previous update being served regardless of its source, if any.

A network is called \emph{homogeneous} if $\lambda_{j}^{(i)}=\lambda^{({i})}, \mu_j=\mu$, for all $j \in [n], i \in [m]$; otherwise, it is \emph{heterogeneous}. In the case of a single source in a homogeneous network, we denote $\lambda^{({1})}$ simply by $\lambda$.

Consider one particular source. Suppose the freshest update at the monitor at time $t$ is generated at time $u(t)$, the \emph{age of information} at the monitor (in short, AoI) is defined as $\Delta(t) =t-u(t)$, which is the time elapsed since the generation of the last received update \cite{kaul2012real}. 
From the definition, it is clear that AoI linearly increases at a unit rate with respect to $t$, except some reset jumps to a lower value at points when the monitor receives a fresher update from the source. The age of information of our network is shown in Figure \ref{fig1}.
For a particular source, let $t_{1}, t_{2},\dots, t_N$ be the generation times of all transmitted updates at all servers in increasing order. 
The black dashed lines show the age of every update.
Let $T_{1}, T_{2},\dots, T_N$ be the receipt time of all updates. Note that due to the contention among different updates at the same server, some updates may be dropped and not delivered at all. The red solid lines show AoI.

We note a key difference between the model in this work and most previous models. Updates come from different servers, therefore they might be out of order at the monitor and thus a new arrived update might not  have any effect on AoI because a fresher update has  already been delivered. As an example, from the $6$ updates shown in Figure \ref{fig1}, \emph{useful} updates that change AoI are updates $1,3,4$ and $6$, while the rest are disregarded as their information is obsolete when arriving at the monitor. 
Thus among all the received updates, we only need to consider the \emph{useful} ones that lead to a change in AoI. 
\begin{figure}
\centering
\includegraphics[width=0.42\textwidth]{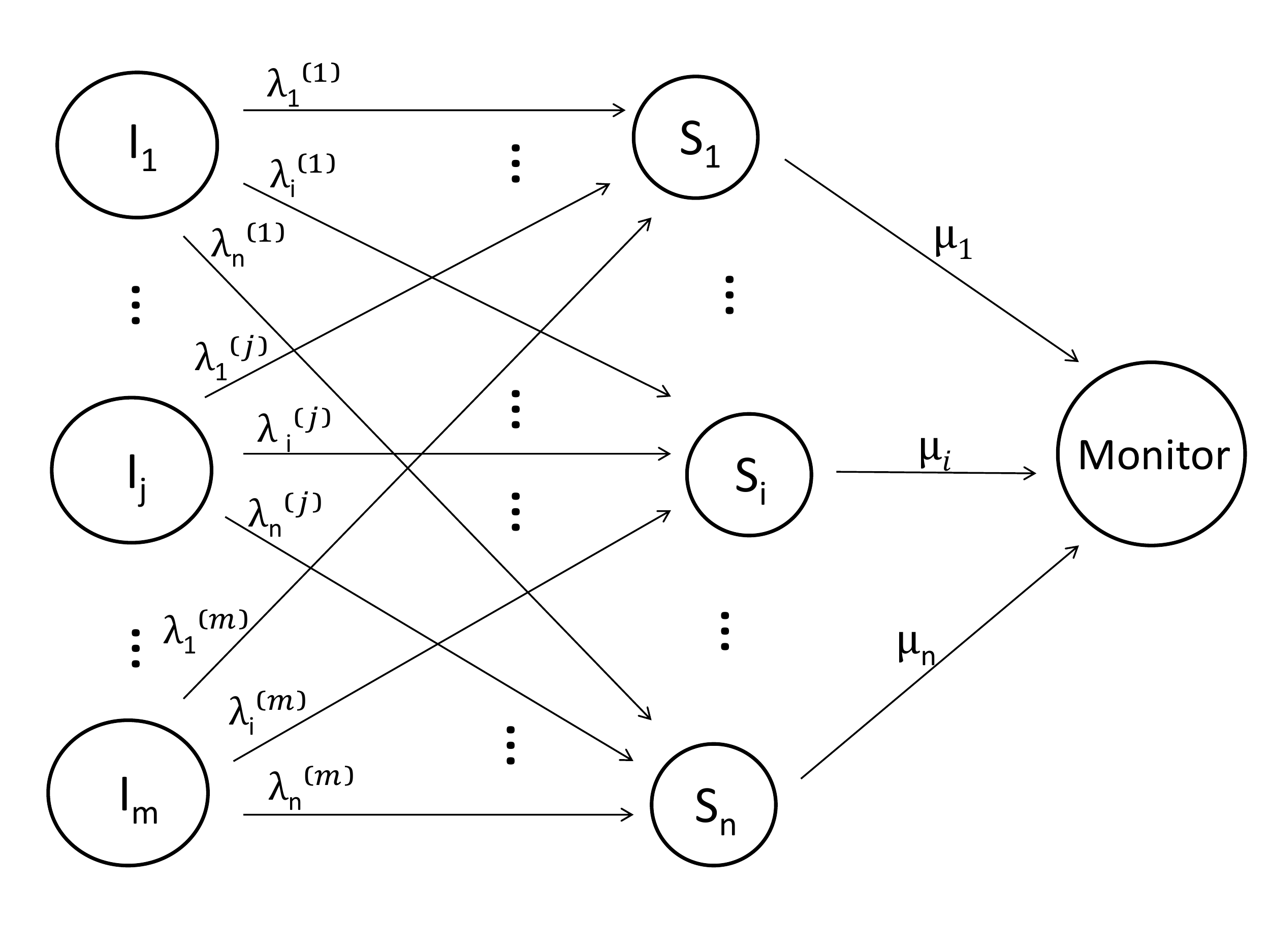} 
\caption{The $n$-server monitoring network with $S_{1},S_{2},...,S_{n}$ being the servers and $I_{1},I_{2},...,I_m$ being the independent information sources, sending the updates from sources to the monitor.
}

\label{fig2}
\end{figure}

\begin{figure}
\centering
\includegraphics[width=0.4\textwidth]{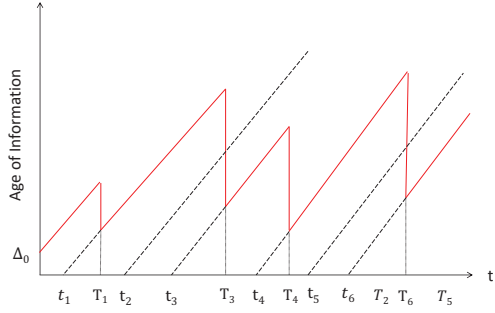} 
\caption{AoI for a particular source in a network with $n$ servers.}
\label{fig1}
\end{figure}

 The interest of this paper is the average AoI for each source at the monitor. The \emph{average AoI} \cite{kaul2012real} is the limit of the average age over time: $\Delta \triangleq \lim_{T \to \infty}  \frac{1}{T} \int_0^T \Delta(t)  \text{d}t$, and for a stationary ergodic system, it is also the limit of the average age over the ensemble: $\Delta= \lim_{t \to \infty} \mathbb{E}[\Delta(t)]$.
\begin{comment}
Based on \cite[Thm. 1]{yates2018age} for a stationary ergodic status updating system we have:
\begin{equation}\label{eq1}
\Delta= \frac{\mathbb{E}{[YT]}+ \mathbb{E}{[Y^2]}/2}{\mathbb{E}{[Y]}},
\end{equation}
where $Y$ is the inter-arrival time of useful updates at the monitor and $T$ is the time each useful update spends in the system. Calculating closed-form expressions of \eqref{eq1} is not tractable in general.
\end{comment}

In the paper, we view our system as a stochastic hybrid system (SHS) and apply a method first introduced in \cite{yates2018age} in order to calculate AoI.

In the SHS, the state is composed of a discrete state and a continuous state. 
The discrete state $q(t) \in \mathcal{Q}$, for a discrete set $\mathcal{Q}$, is a continuous-time discrete Markov chain, and the continuous-time continuous state $\mathbf{x}(t) = (x_0(t),x_1(t),\dots,x_n(t)) \in \mathbb{R}^{n+1}$ is a continuous-time stochastic process. For example,  the discrete state can represent which server has the freshest update in the network. For another example, we can use $x_0(t)$ to represent the age at the monitor, and $x_j(t)$ for the age at the $j$-th server, $j=1,2,\dots,n$.  

Graphically, we represent each State $q \in \mathcal{Q}$ by a node. 
For the discrete Markov chain $q(t)$, transitions happen from one state to another through directed transition edge $l$, and the time spent before the transition occurs is exponentially distributed with rate $\lambda{(l)}$. Note that it is possible to transit from the one state to itself. The transition occurs when an update arrives at a server, or an update is received at the monitor. Thus the transition rate is the update arrival rate or the service rate,  $\lambda{(l)} \in \{\lambda_{1}^{(1)},...,\lambda_{n}^{(m)}, \mu_{1},...,\mu_{n}\}$. 
Denote by $L'_{{q}}$ and $L_{{q}}$ the sets of incoming and outgoing transitions of State $q$, respectively.
When transition $l$ occurs, we write that the discrete state transits from $q_l$ to $q_l^\prime$. 
For a transition, we denote that
the continuous state changes from $\mathbf{x}$ to $\mathbf{x}^\prime$.
In our problem, this transition is linear in the vector space of $\mathbb{R}^{n+1}$, i.e., $\mathbf{x}^\prime=\mathbf{x}A_{l}$, for some real matrix $A_{l}$ of size $(n+1) \times (n+1)$.  
Note that when we have no transition, the age grows at a unit rate for the monitor and relevant servers, and is kept unchanged for irrelevant servers.
Hence, within the discrete State $q$, $\mathbf{x}(t)$ evolves as a piece-wise linear function in time, namely, $\frac{\partial{\mathbf{x}(t)}}{\partial{t}} = \mathbf{b}_{q}$,  for some $ \mathbf{b}_q \in \{0,1\}^{n+1}$. %In other words,  the age grows at a unit rate for the monitor and relevant servers; and the age is kept unchanged for irrelevant servers.
%To illustrate the concepts, we use a running example of having $1$ source and $2$ heterogeneous servers below.
\begin{example} 
Consider the case of 2 heterogeneous servers and 1 source. At each time, we keep track of the age of information in the continuous state $\mathbf{x}=(x_0,x_1,x_2)$. Here $x_0$ is the age at the monitor, $x_1$ is the age for the first server, and $x_2$ is the age for the second server. In this example, the discrete states are $\mathcal{Q}=\{1,2\}$. In State $1$, Server $1$ contains the freshest information, i.e., $x_1 \le x_2$; and in State $2$, Server $2$ has the freshest information, namely, $x_2 \le x_1$. Obviously, our system changes its state when servers receive new information. For instance, there is a transition $l$ from State $1$ to State $2$ with rate of $\lambda_2$, when a new update arrives at Server 2 and the freshest information was at Server $1$ before that. Hence, $q_{l}=1$ and $q_{l}^{\prime} = 2$  which shows that State $2$ is an outgoing transition for State $1$ and State $1$ is an incoming transition for State $2$. Moreover, the continuous state changes from $\mathbf{x}=(x_0,x_1,x_2)$ to $\mathbf{x}^\prime = (x_0,x_1,0)$. When there are no transitions, all entries of $\mathbf{x}$ grow linearly in time.
%Also, we can have some self loops in the state transitions where the state of system does not change but the vector $\textbf{X}$ changes with every transition with a specific rate ($\lambda_1, \lambda_2, \mu_1, \mu_2$).
\end{example}

For our purpose, we consider the discrete state probability
 \begin{equation}
 \pi_{\hat{q}} (t) \triangleq \mathop{\mathbb{E}}[\delta_{\hat{q},q(t)}] =P[q(t)=\hat{q}],
\end{equation}
and the correlation between the continuous state $\mathbf{x}(t)$ and the discrete state $q(t)$:
 \begin{equation}
 \mathbf{v}_{\hat{q}} = (v_{\hat{q}_0} (t),\dots, v_{\hat{q}_n} (t)) \triangleq \mathop{\mathbb{E}}[\mathbf{x}(t) \delta_{\hat{q},q(t)}].
\end{equation}
Here $\delta_{\hat{q},q(t)}$ denotes the Kronecker delta function, i.e., it equals $1$ if $q(t)=\hat{q}$, and it equals $0$ otherwise. When the discrete state $q(t)$ is ergodic, ${\pi}_q(t)$ converges uniquely to the stationary probability ${{\pi}}_q$, for all $q \in \mathcal{Q}$. We can find these stationary probabilities from the following set of equations knowing that $\sum_{q \in \mathcal{Q}}^{}\pi_{q} = 1$,
\begin{align*}
{{\pi}}_{{q}} \sum_{l \in L_{{q}}}^{} \lambda{(l)}=
 \sum_{l \in L^\prime_{{q}}}^{}
\lambda{(l)} {{\pi}}_{q_l}, \quad {q} \in \mathcal{Q}.
\end{align*}

A key lemma we use to develop AoI for our LCFS queue model is the following from \cite{yates2018age}, which was derived from the general SHS results in \cite{hespanha2006modelling}.

\begin{lemma}[\cite{yates2018age}] \label{lem:yates}
 If the discrete-state Markov chain $q(t)$ is ergodic with stationary distribution ${\pi}$ and we can find a non-negative solution of $ \{{\mathbf{v}}_{{q}}, {q} \in \mathcal{Q} \}$ such that 
\begin{equation}\label{eq:yateslemma}
{\mathbf{v}}_{{q}} \sum_{l \in L_{{q}}}^{} \lambda{(l)}=
\mathbf{b}_{{q}} {\pi}_{{q}} + \sum_{l \in L^\prime_{{q}}}^{}
\lambda{(l)} {\mathbf{v}}_{q_l} A_{l}, \quad {q} \in \mathcal{Q},
\end{equation}
then the average age of information is given by 
\begin{align}
\Delta= \sum_{{q} \in \mathcal{Q}}^{} {v}_{{q} {0}}.
\end{align} 
\end{lemma}

\section{AoI in Homogeneous Networks} \label{sec:homo}
\subsection{Single Source Multiple Servers}
\label{LCFS}
In this section, we present AoI calculation with the LCFS queue for the single-source $n$-server homogeneous network using SHS techniques. %In this network, upon arrival of a new update, each server immediately drops any previous update in service and starts to serve the new update.
Note that to compute the average AoI, Lemma \ref{lem:yates} requires solving $|\mathcal{Q}|(n+1)$ linear equations of $ \{{\mathbf{v}}_{{q}}, {q} \in \mathcal{Q} \}$. To obtain explicit solutions for these equations, the complexity grows with the number of discrete states. Since the discrete state typically represents the number of idle servers in the system for homogeneous servers, $|\mathcal{Q}|$ should be $n+1$. In the following, we introduce a method inspired by \cite{yates2018status} to reduce the number of discrete states and efficiently describe the transitions.

We define our continuous state $\mathbf{x}$ at  time $t$ as follows:
the $0$-th element $x_{0}$ is AoI at the monitor, the first element $x_1$ corresponds to the freshest update among all updates in the servers, the second element $x_2$ corresponds to the second freshest update in the servers, etc. With this definition we always have $x_{1} \leq x_{2} \leq .... \leq x_{n}$, for any time $t$. Note that the index $i$ of $x_i$ does not represent a physical server index, but the $i$-th smallest age of information among the $n$ servers. The physical server index for $x_i$ changes with each transition. We say that the server corresponding to $x_i$ is the $i$-th \emph{virtual} server. 

A transition indexed by $l$ is triggered by (i) the arrival of an update at a server, or (ii) the delivery of an update to the monitor. Recall that we use $\mathbf{x}$ and $\mathbf{x}'$ to denote the continuous state of AoI right before and after the transition $l$.

When one update arrives at the monitor and the server delivering the update becomes idle, we introduce a \emph{fake update} to the server using the method introduced in \cite{yates2018status}. 
Thus we can reduce the calculation complexities and only have one discrete state indicating that all servers are virtually busy. We denote this state by $q=0$.
In particular, we put the current update that is in the monitor to an idle server until the next update reaches this server. This assumption does not affect our final calculation for AoI, because even if the fake update is delivered to the monitor, AoI at the monitor does not change. Moreover, serving the fake update does not affect the service of future actual updates because of preemption in service.

When an update is delivered to the monitor from the $k$-th virtual server, the server becomes idle and as previously stated, receives the fake update. The age at the monitor becomes $x'_0=x_k$, and the age at the $k$-th server becomes $x'_k = x'_0=x_k$.
In this scenario, consider the update at the $j$-th virtual server, for $j>k$.  Its delivery to the monitor does not affect AoI since  it is older than the current update of the monitor, i.e., $x_j \ge x_k = x'_0$. 
Hence, we can adopt a \emph{fake preemption} where the update for the $j$-th virtual server, for all $k \le j \le n$, is preempted and replaced with the fake current update at the monitor. Therefore, we set $x'_j = x'_0 = x_k$, $k \le j \le n$.
Physically, these updates are not preempted and as a beneficial result, the servers do not need to cooperate and can work in a distributed manner. 

\begin{figure}
\centering
\includegraphics[width=0.2\textwidth]{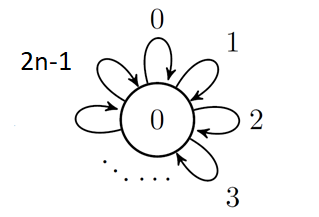} 
    \caption{SHS for our model with $n$ servers.}
   
  \label{fig3}  
\end{figure}

\begin{table}
%\label{table1}
\centering
\begin{tabular}{ cccccccccc }
 
 $l$ & $\lambda{(l)}$ & $\mathbf{x}^\prime$ =$\mathbf{x}A_{l}$ \\          \hline
 $0$ & $\lambda$ & $[x_{0},0,x_{2},x_{3},x_{4},...,x_{n}]$ \\ \hline
 $1$ & $\lambda$ & $[x_{0},0,x_{1},x_{3},x_{4},...,x_{n}]$ \\ \hline
 $2$ & $\lambda$ & $[x_{0},0,x_{1},x_{2},x_{4},...,x_{n}]$ \\ \hline
 
  & $\vdots$ & $\vdots$ \\  \hline
 $n-1$ & $\lambda$ & $[x_{0},0,x_{1},x_{2},x_{3},..,x_{n-1}]$ \\ \hline
 $n$ & $\mu$ & $[x_{1},x_{1},x_{1},x_{1},...,x_{1}]$ \\  \hline
 $n+1$ & $\mu$ & $[x_{2},x_{1},x_{2},x_{2},...,x_{2}]$ \\  \hline
 $n+2$ & $\mu$ & $[x_{3},x_{1},x_{2},x_{3},...,x_{3}]$ \\  \hline
  & $\vdots$ & $\vdots$ \\  \hline
 $2n-1$ & $\mu$ & $[x_{n},x_{1},x_{2},x_{3},...,x_{n}]$ \\  
 
\end{tabular}
\caption{Table of transformation for the Markov chain in Figure~\ref{fig3}.}
\label{table:table 0}
\end{table}
%By utilizing the fake update and fake preemption, i.e, putting the current age of the monitor to the idle servers and preempting all the updates older than that of the monitor, 
By utilizing virtual servers, fake updates, and fake preemptions, we reduce SHS to a single discrete state with linear transition described by matrix $A_l$, $l \in [0:2n-1]$.
We illustrate our SHS with discrete state space of $Q=\{0\}$ in Figure \ref{fig3}. The stationary distribution ${\pi}_{0}$ is trivial and ${\pi}_{0}=1$. We set $\mathbf{b}_{q}=(1,...,1)$ which indicates that the age at the monitor and the age of each update in the system grows at a unit rate. The transitions are labeled $l  \in \{0,1,...,2n-1\}$ and for each transition $l$ we list the transition rate and the transition mapping in Table \ref{table:table 0}. For simplicity, we drop the index $q=0$ in the vector $\mathbf{v}_0$, and write it as $\mathbf{v}=(v_0,v_1,\dots,v_n)$. Because we have one state, $\mathbf{x}A_{l}$ and $\mathbf{v}A_{l}$ are in correspondence. Next, we describe the transitions in Table \ref{table:table 0}. 

{\bf Case I.} $l=0,1,..,n-1:$ When a fresh update arrives at virtual server $l+1$, the age at the monitor remains the same and $x_{l+1}$ becomes zero. This server has the smallest age, so we take this zero and reassign it to the first virtual server, namely, $x^\prime_{1}=0$. In fact virtual Servers $1,2,\dots,l+1$ all get reassigned virtual server numbers. Specifically, after transition $l$, virtual server $l+1$ becomes virtual Server $1$, virtual Server $1$ becomes virtual Server $2$,..., and virtual Server $l$ becomes virtual Server $l+1$. The transition rate is the arrival rate of the update, $\lambda$. The matrix $A_l$ is

\vspace{-0.35cm}
\begin{footnotesize}
\begin{align}
    \bordermatrix{
~	&	0	&	1	&	2	&	\dots	&	l+1	&	l+2	&	\dots	&	n	\cr
0	&	1	&		&		&		&		&		&		&		\cr
1	&		&	0	&	1	&		&		&		&		&		\cr
\vdots	&		&		&		&	\ddots	&		&		&		&		\cr
l	&		&		&		&		&	1	&		&		&		\cr
l+1	&		&		&		&		&		&	0	&		&		\cr
l+2	&		&		&		&		&		&	1	&		&		\cr
\vdots	&		&		&		&		&		&		&	\ddots	&		\cr
n	&		&		&		&		&		&		&		&	1	\cr
    }.
\end{align}
\end{footnotesize}

\vspace{-0.35cm}
\noindent{\bf Case II.} $l=n,n+1,..,2n-1:$ When an update is received at the monitor from virtual Server $l+1-n$, the age at the monitor changes to $x_{l+1-n}$ and this server becomes idle. Using fake updates and fake preemption we assign $x'_{j}=x_{l+1-n}$, for all $l+1-n \le j \le n$. The transition rate is the service rate of a server, $\mu$. The matrix $A_l$ is

\vspace{-0.25cm}
\begin{footnotesize}
\begin{align}
\bordermatrix{
~	&	0	&	1	&			\dots	&	l-n	&	l+1-n	&	\dots	&	n	\cr
0	&	0	&		&				&		&		&		&		\cr
1	&		&	1	&				&		&		&		&		\cr
\vdots	&		&		&			\ddots	&		&		&		&		\cr
l-n	&		&		&				&	1	&		&		&		\cr
l+1-n	&	1	&	0	&			\dots	&	0	&	1	&	\dots	&	1	\cr
l+2-n	&	0	&		&			\dots	&	\dots	&	\dots	&	\dots	&	0	\cr
\vdots	&	\vdots	&		&				&		&		&		&	\vdots	\cr
n	&	0	&		&			\dots	&	\dots	&	\dots	&	\dots	&	0	\cr
}.
\end{align} 
\end{footnotesize}

\vspace{-0.25cm}
\noindent	Below we state our main theorem on the average AoI for the single-source $n$-server network.

\begin{theorem} \label{theory}
Define $\rho = \frac{\lambda}{\mu}$. The average age of information at the monitor for a homogeneous single-source $n$-server network where each server has a LCFS queue is:

\begin{align}
\label{AoI_single_source}
&\Delta
%\sum_{j=2}^{n} w_{j} +\frac{1}{n\lambda} +\frac{\lambda}{n\mu} w_{n}= \\
%&\sum\limits_{j = 2}^n \frac{1}{n \lambda} \prod\limits_{i=1}^{j-1} \frac{\rho (n-i+1)}{i    + (n-i) \rho} + 
% \frac{1}{n \lambda}
%+
% \frac{1}{n \mu}
%\frac{1}{n  } \prod\limits_{i=1}^{n-1} \frac{\rho (n-i+1)}{i    + (n-i) \rho} \\
 =
  \frac{1}{\mu} \left[ 
\frac{1}{n \rho} \sum\limits_{j = 1}^{n-1}  \prod\limits_{i=1}^{j } \frac{\rho (n-i+1)}{i    + (n-i) \rho}+ 
 \frac{1}{n \rho}
+
\frac{1}{n^2  } \prod\limits_{i=1}^{n-1} \frac{\rho (n-i+1)}{i    + (n-i) \rho} 
\right].
\end{align}
\noindent 
%where, %$w_2=\frac{1}{(n-1)\lambda+\mu}$, $2 \le j \le n-1$, and
%\begin{align*}
%w_j &= \frac{1}{n \lambda} \prod\limits_{i=1}^{j-1} \frac{\lambda (n-i+1)}{i \mu  + (n-i) \lambda}\\
%&= \frac{1}{n \lambda} \prod\limits_{i=1}^{j-1} \frac{\rho (n-i+1)}{i    + (n-i) \rho}
%\quad 2 \le j \le n.
%\end{align*}
\end{theorem}

\begin{proof}
%\bl{?? Question. The equation can be rewritten as 
%\begin{align*}
 %    {\mathbf{v}}B \triangleq {\mathbf{v}}\left[ -(n\lambda+n\mu)I + \lambda(A_0+\dots+A_{n-1}) + \mu(A_n + \dots + A_{2n-1})  \right] = -[1,1,1,1,1,1,1,...,1].
%\end{align*}
%Now is it possible the directly write out the matrix $B$ and solve for $\mathbf{v}$? Will the proof be cleaner?
%}

Recall that $\mathbf{v}$ denotes the vector  $\mathbf{v}_0$ for the single state $q=0$.
By Lemma \ref{lem:yates} and the fact that there is only one state,
we need to calculate the vector $\mathbf{v}$ as a solution to \eqref{lem:yates}, and the $0$-th coordinate $v_{0}$ is the average AoI at the monitor. As we mentioned $\mathbf{v}A_{l}$ is in correspondence with $\mathbf{x}A_{l}$, so we have:

\vspace{-.4cm}
\begin{small}
\begin{align}
 (n\lambda+n\mu){\mathbf{v}}= & \quad \quad (1,1,1,1,1,1,1,...,1) \nonumber \\
   &+\lambda(v_{0},0,v_{2},v_{3},v_{4},...,v_{n})  \nonumber \\ 
   &+\lambda(v_{0},0,v_{1},v_{3},v_{4},...,v_{n})  \nonumber \\
   &+\lambda(v_{0},0,v_{1},v_{2},v_{4},...,v_{n})  \nonumber \\
   & \quad \quad \quad \quad \vdots \quad \quad \quad \quad    \quad \vdots  \nonumber \\
   &+\lambda(v_{0},0,v_{1},v_{2},v_{3},...,v_{n-1})  \nonumber \\
   &+\mu(v_{1},v_{1},v_{1},v_{1},v_{1},...,v_{1})   \nonumber \\
   &+\mu(v_{2},v_{1},v_{2},v_{2},v_{2},...,v_{2})   \nonumber \\
   &+\mu(v_{3},v_{1},v_{2},v_{3},v_{3},...,v_{3})   \nonumber \\
   & \quad \quad \quad \quad \vdots \quad \quad \quad \quad    \quad \vdots  \nonumber \\
   &+\mu(v_{n},v_{1},v_{2},v_{3},...,v_{n-1},v_{n})   \label{eq:main}.
\end{align}  
\end{small}

\vspace{-.3cm}
\noindent From the $0$-th coordinate of \eqref{eq:main}, we have $(n\lambda+n\mu)v_{0}= 1+n\lambda v_{0} + \mu \sum_{j=1}^{n} v_{j}$, implying
\begin{align}
%(n\lambda+n\mu)v_{0}= 1+n\lambda v_{0} + \mu \sum_{j=1}^{n} v_{j},   \implies 
  v_{0}= \frac{1}{n\mu} + \frac{\sum_{j=1}^{n} v_{j}}{n} \label{eq3}.
\end{align}
From the $1$-st coordinate of \eqref{eq:main}, it follows  that 
%\begin{align}
$v_{1}~=~\frac{1}{n\lambda}$.
%\end{align}
Then, to calculate $v_{0}$, we have to calculate $v_{i}$ for $i \in \{2,...,n\}$. From the $i$-th coordinate of \eqref{eq:main},

\vspace{-0.35cm}
\begin{small}
\begin{align}
%& ((n-i+1)\lambda+(i-1)\mu) v_{i} \nonumber\\
%= & 1+ \mu \sum_{j=1}^{i-1} v_{j} +\lambda (n-i+1) v_{i-1}.
((n-i+1)\lambda+(i-1)\mu) v_{i} =   1+ \mu \sum_{j=1}^{i-1} v_{j} +\lambda (n-i+1) v_{i-1}.
\label{w2}
\end{align}
\end{small}

\vspace{-0.4cm}
\noindent	For $i \in \{2,3,...,n-1\} $, from \eqref{w2}, we obtain
\begin{align*}
(i\mu +(n-i)\lambda)(v_{i+1}-v_{i}) = \lambda (n-i+1) (v_{i}-v_{i-1}) .
\end{align*} 
Hence,
%\begin{align*}
$w_{i+1} \triangleq v_{i+1}-v_{i} 
%&=\frac{\lambda (n-i+1)}{(i\mu +(n-i)\lambda)} (v_{i}-v_{i-1}) 
%&\nonumber\\
=\frac{\lambda (n-i+1)}{(i\mu +(n-i)\lambda)} w_i$.
%\end{align*}
\noindent Setting $i=2$ in \eqref{w2}, we have
\begin{align}
    ((n-1)\lambda+\mu)v_{2}=1+ \mu v_{1} +\lambda (n-1)v_{1}.
    \label{w_2_source}
\end{align}
Simplifying \eqref{w_2_source}, we obtain 
%$w_{2}=v_{2}-v_{1}=\frac{2-1/n+\frac{\mu}{n\lambda}}{\mu+(n-1)\lambda} - \frac{1}{n\lambda}=\frac{1}{(n-1)\lambda+\mu}$.
 $w_{2}=v_{2}-v_{1}=\frac{1}{(n-1)\lambda+\mu}$. Therefore, we write
 \begin{align}
w_j &= \frac{1}{n \lambda} \prod\limits_{i=1}^{j-1} \frac{\lambda (n-i+1)}{i \mu  + (n-i) \lambda}, 2 \le j \le n.
\label{w_j}
 \end{align}
%As a result by replacing $v_{1}$ with $\frac{1}{n\lambda}$, we get $v_{2}= \frac{2-1/n+\frac{\mu}{n\lambda}}{\mu+(n-1)\lambda}$.
%Therefore $w_{2}=v_{2}-v_{1}=\frac{2-1/n+\frac{\mu}{n\lambda}}{\mu+(n-1)\lambda} - \frac{1}{n\lambda}=\frac{1}{(n-1)\lambda+\mu}.$  
Finally, setting $i=n$ in \eqref{w2}, 
\begin{align}
   (\lambda+ (n-1)\mu)v_{n}=  1+ \mu \sum_{j=1}^{n-1} v_{j}+ \lambda v_{n-1} ,
\end{align}
implying
 $   \mu  \sum_{i=1}^{n} v_{i}
 =   \mu \sum_{j=1}^{n-1} v_{j}+\mu v_{n}  
 =   (\lambda+ (n-1)\mu)v_{n} + \mu v_{n} -1 - \lambda v_{n-1}.
$
Hence,
\begin{align}\label{eq2}
 \frac{1}{n} \sum_{i=1}^{n} v_{i} = \frac{\lambda}{n\mu} w_{n} +v_{n}-\frac{1}{n\mu}.
\end{align}
Combining \eqref{eq3} and \eqref{eq2}, we obtain the average AoI as
\begin{align*}
AoI = v_0=v_{n} +\frac{\lambda}{n\mu} w_{n} =\sum_{j=2}^{n} w_{j} +\frac{1}{n\lambda} +\frac{\lambda}{n\mu} w_{n},
\end{align*}
which is simplified to \eqref{AoI_single_source} using \eqref{w_j}.
%Thus the proof is completed.
\end{proof}
Figure \ref{fig4} shows the average AoI when the total arrival rate $n\lambda$ is fixed and the number of servers $n$ varies among $1,2,3,4,10$. We observe that for up to $4$ servers,  a significant decrease in the average AoI occurs with the increase of $n$. However, increasing the number of servers beyond $4$ provides only a negligible decrease in AoI. %and it is therefore inefficient from a network-resource point of view. 
\begin{figure}
\centering
\includegraphics[width = 0.43\textwidth]{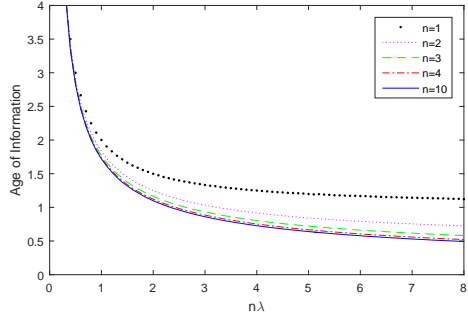} 
\caption{The average AoI versus the number of servers, for fixed total arrival rate. For each server, the service rate $\mu=1$ and the total arrival rate $n\lambda$ is shown in the x-axis.}
\label{fig4}
\end{figure}

\begin{figure}\label{lam_con}
\centering
\includegraphics[width = 0.43\textwidth]{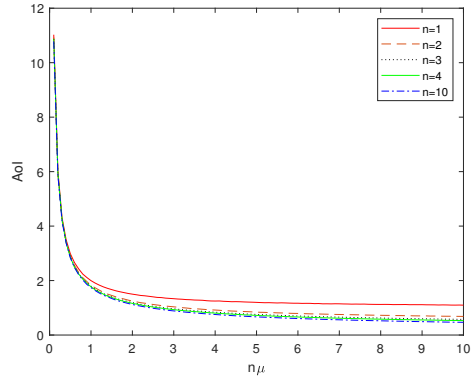} 
\caption{AoI versus the number of servers. For each server, the arrival rate is $\lambda=1$ and the total service rate $n\mu$ is shown in the x-axis. }
\end{figure}

In Figure \ref{fig6}, LCFS (with preemption in service), LCFS with preemption in waiting, and FCFS queue models are compared numerically. Preemption in waiting means that when a new update arrives, we drop any old updates that have not been served. As can be seen from the figure, LCFS outperforms the other two queue  models, which coincides with the intuition that exponential service time is memoryless and older updates in service should be preempted. Moreover, we observe that the optimal arrival rate for FCFS queue is approximately $0.5$ for all $n \le 50$, shown in Table \ref{table:table 2}. 
\begin{table}[]
    \centering
    \begin{tabular}{|c|c|c|c|c|c|c|}
        \hline 
        n & $1$ & $2$&  $4$ & $10$ & $50$ \\ \hline 
        $\lambda^{*}$ & $0.5$ & $0.5$& $0.525$ & $0.53$ & $0.529$  \\ \hline
    \end{tabular}
    \caption{Optimal individual arrival rate for FCFS queue, $\mu=1$.}

  %  \bl{Do you mean by $\lambda^*$ the individual arrival rate $\lambda$ or the total arrival rate?? Please use individual arrival rate.}

    \label{table:table 2}
\end{table}
\begin{figure}
\centering
\includegraphics[width=0.45\textwidth]{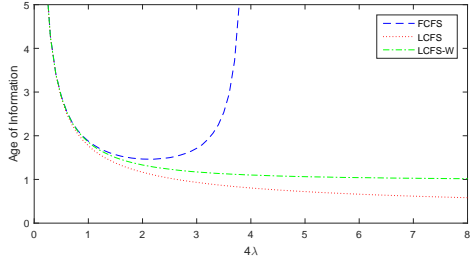}
\caption{Comparison of the average AoI under LCFS, FCFS, and LCFS with preemption in waiting (LCFS-W). The number of servers is $n=4$ and the service rate is $\mu=1$ for each server.}
\label{fig6}
\end{figure}
%\begin{remark}
%When $n=2,3$ assuming $\rho =\frac{\lambda}{\mu}$, AoI reduces to $\frac{1}{2\mu}(1+\frac{1}{\rho}+\frac{1}{1+\rho})$ and $\frac{1}{3\mu}(1+\frac{1}{\rho}+\frac{4(1+\rho)}{(2+\rho)(2\rho+1)})$, respectively.
%\end{remark}
\subsection{Multiple Sources Multiple Servers}
\label{multiple}
In this subsection, we present the average AoI with the LCFS queue for the $m$-source $n$-server homogeneous network. The arrival rate of Source $i$ at any server is $ \lambda_{j}^{(i)}= \lambda^{(i)} $, for all $i~\in~[m], j~\in ~[n]$. The arrival rate of the sources other than Source $i$ is $\overline{\lambda^{(i)}} ~\triangleq~\sum_{i' \neq i} \lambda^{(i')},  i~  \in~[m]$. The service rate at any server is $\mu$. Our goal is to compute $\Delta_{i}$, the average AoI at the monitor for Source $i$, $i \in [m]$. Without loss of generality, we calculate $\Delta_1$ for Source $1$. In the queue model, upon arrival of a new update from any source, each server immediately drops any previous update in service regardless of its source and starts to serve the new update.  

The continuous state $\mathbf{x}$ represents the age for Source $1$, and similar to the single-source case, it is defined as follows:
$x_0$ is AoI of Source $1$ at the monitor, $x_1$ is the age of the freshest update among all updates of Source $1$ in the servers, $x_2$ corresponds to the second freshest update in the servers, etc. Therefore $x_{1}   \leq x_{2} \leq .... \leq x_{n}$, for any time $t$. 
Using fake updates and fake preemption as explained in Section~\ref{LCFS}, we obtain an SHS with a single discrete state and $3n$ transitions described below: 
\begin{comment}
\begin{itemize}
\item	$0\le l \le n-1$: A fresh update arrives at virtual Server $1 \le l+1 \le n$ from Source $1$: this fresh update becomes the freshest update: so $x_1' = 0$. Now, the old freshest update becomes the second freshest update, that is $x_2' = x_1$, and so on. Then $\mathbf{x}' = [x_0,0, x_1, \ldots,x_{l}, x_{l+2}, \ldots,x_{n } ]$. The transition rate is~$\lambda^{(1)}$.

\item	$n\le l \le 2n-1$: A fresh update arrives at virtual server $1 \le l+1-n \triangleq l' \le n$ from source $i \neq 1$: The age at the source does not change $x_0' = x_0$. The $l'$-th freshest update is lost. Moreover, if the virtual Server $l'$ does complete service, it does not reduce the age of the process of interest. Thus, the $l'$-th virtual server becomes the $n$-th virtual server. Therefore, we have
$\mathbf{x}' = [x_0,x_1,   \ldots, x_{l' -1}, x_{l' +1} \ldots , x_{n}, x_0]$.  The transition rate is $\overline{\lambda^{(1)}}$.

\item $2 n\le l \le 3n-1$: the update in the $l+1-2n \triangleq h $ is delivered. The age $x_0$ is reset to $x_{h}$ and the virtual Server $h$ becomes idle. using fake update and fake preemption, we reset $x_l' = x_h, h \le j \le  n$. The transition rate is $\mu$.

\end{itemize}
\end{comment}

{\bf Case I.}	$l \in [0: n-1]$: A fresh update arrives at virtual Server $l+1$ from Source $1$. This update is the freshest update, so $x_1' = 0$. Now, the previous freshest update becomes the second freshest update, that is $x_2' = x_1$, and so on. Then $\mathbf{x}' = (x_0,0, x_1, \ldots,x_{l}, x_{l+2}, \ldots,x_{n } )$. The transition rate is~$\lambda^{(1)}$.

{\bf Case II.}	$l \in [n: 2n-1]$: A fresh update arrives at virtual Server $l' \triangleq l+1-n$ from Source $i \neq 1$. The age at the monitor does not change, namely, $x_0' = x_0$. The $l'$-th freshest update is preempted. Moreover, since the virtual Server $l'$ drops the update for the source of interest (Source $1$), with fake update, the $l'$-th virtual server becomes the $n$-th virtual server with age $x_0$. Therefore, we have
$\mathbf{x}' = (x_0,x_1,   \ldots, x_{l' -1}, x_{l' +1} \ldots , x_{n}, x_0)$. The transition rate is $\overline{\lambda^{(1)}}$.

{\bf Case III.} $l \in [2n: 3n-1]$: the update of Source $1$ in virtual Server $h \triangleq l+1-2n $ is delivered. The age $x_0$ is reset to $x_{h}$ and the virtual Server $h$ becomes idle. Using fake update and fake preemption, we reset $x_l' = x_h, h \le j \le  n$. The transition rate is $\mu$.

\begin{algorithm}
\begin{algorithmic}
\State \textbf{Input:} $n, \lambda^{(1)},\overline{\lambda^{(1)}}, \mu$
\State \textbf{Output:} $\Delta_1$ 

\State {\bf Part 1. Base case for $v_1.$}
\State
$c_2=\frac{n \lambda^{(1)}}{\overline{\lambda^{(1)}}}, d_2 = \frac{-1}{\overline{\lambda^{(1)}}}$ and $c_3 = \frac{n \lambda^{(1)} ((n-1)\lambda^{(1)}+\mu)}{2 \overline{\lambda^{(1)}}^2}, d_3 = -\frac{1}{2\overline{\lambda^{(1)}}}-\frac{((n-1)\lambda^{(1)}+\mu)}{2 \overline{\lambda^{(1)}}^2}$.
       
        \For{$j=4:n+1$}
            \State $c_j=  \frac{(n-j+2)\lambda^{(1)}+(j-2)\overline{\lambda^{(1)}}+(j-2)\mu}{(j-1)\overline{\lambda^{(1)}}} c_{j-1}  -\frac{\lambda^{(1)}(n-j+3)}{(j-1)\overline{\lambda^{(1)}}}c_{j-2}$
            \State $d_j=\frac{(n-j+2)\lambda^{(1)}+(j-2)\overline{\lambda^{(1)}}+(j-2)\mu}{(j-1)\overline{\lambda^{(1)}}} d_{j-1}-\frac{\lambda^{(1)}(n-j+3)}{(j-1)\overline{\lambda^{(1)}}}d_{j-2}$
        \EndFor
\State
$  v_1 = \frac{\frac{1}{n\mu} - \sum_{j=2}^{n+1} d_j (\frac{j-1}{n})}{\sum_{j=2}^{n+1} c_j (\frac{j-1}{n})}$  % where $c_j$ and $d_j$ are based on Equation \eqref{init}.

\State {\bf Part 2. Recursion for $v_j, j \ge 2$:}
        \For{$j=2:n$} 
            \State $a_j=1$,
            $b_j= -v_{j-1}$
            \For{$k=j:n$}
                \State $a_{k+1}= \frac{n-k+1}{k \overline{\lambda^{(1)}}} a_k+
             \frac{\mu}{k \overline{\lambda^{(1)}}} \sum_{l=2}^{k} (l-1)a_l$,
                \State $b_{k+1}=\frac{-1}{k \overline{\lambda^{(1)}}}+ \frac{n-k+1}{k \overline{\lambda^{(1)}}} b_k+
            \frac{\mu}{k \overline{\lambda^{(1)}}} \sum_{l=2}^{k} (l-1)b_l$
            \EndFor
        \State $v_j=\frac{ \frac{1}{n \mu} +\frac{\sum_{i=1}^{j-1} v_i}{n}-\sum_{i=j+1}^{n+1} b_i (\frac{i-1}{n})}{ \frac{j-1}{n} + \sum_{i=j+1}^{n+1} a_i (\frac{i-1}{n})}$\\
        \EndFor
    %\State $ v_{0} = \frac{1}{n \mu} + \frac{\sum_{i=1}^{n} v_{i}}{n}$ %based on equation \eqref{v0}
    \State \textbf{return} $\Delta_1=v_0 = \frac{1}{n \mu} + \frac{\sum_{i=1}^{n} v_{i}}{n}$ %\Comment{The average AoI is $v_0$}
    \end{algorithmic}
    \caption{The average AoI $\Delta_1$ of Source $1$ for the multiple-source $n$-server homogeneous network.}
    \label{alg:homo}
\end{algorithm}

\begin{theorem}\label{thm:homo_general}
Consider the $m$-source $n$-server homogeneous network, for $n \ge 3$.
The average AoI for Source $1$ can be computed in a recursive manner as in Algorithm \ref{alg:homo}. 
%The unique positive answers for $v_i$ where $i \in \{0,1,2,...,n\}$ in system of equations \eqref{multi-source} exist and we can find the average AoI in a recursive manner as explained in Algorithm $1$. %The fact that each $v_i$ is positive and more details are provided in the proof.
\end{theorem}
\begin{proof}
By applying Lemma \ref{lem:yates} and dropping the index $q=0$, the system of equations for $\mathbf{v}_0=\mathbf{v}=(v_0,v_1,\dots,v_n)$ becomes: 
\begin{align} \label{multi-source}
(n \lambda_1 + n \overline{\lambda^{(1)}} + n \mu) (v_0, v_1,\ldots, v_n)&=  
(1,1,1,\ldots,1,1,1,1)  \nonumber \\
&+ \lambda_1 (v_0, 0, v_2, v_3, \ldots, v_n) \nonumber\\
&+ \lambda_1 (v_0, 0, v_1, v_3, \ldots, v_n) \nonumber\\
&+ \lambda_1 (v_0, 0, v_1, v_2, \ldots, v_n) \nonumber\\
  &  \qquad \vdots \nonumber \\
&+ \lambda_1 (v_0, 0, v_1, v_2, \ldots, v_{n-1}) \nonumber\\
&+ \overline{\lambda^{(1)}} (v_0, v_2, v_3,   \ldots, v_{n}, v_0)\nonumber\\
&+ \overline{\lambda^{(1)}} (v_0, v_1, v_3,   \ldots, v_{n}, v_0)\nonumber\\
&+ \overline{\lambda^{(1)}} (v_0, v_1, v_2,   \ldots, v_{n}, v_0)\nonumber\\
& \qquad \vdots \nonumber \\
&+ \overline{\lambda^{(1)}}  (v_0, v_1, v_2,   \ldots, v_{n-1}, v_0)\nonumber\\
&+ \mu (v_1, v_1, v_1, v_1, \ldots, v_{1}) \nonumber\\
&+ \mu (v_2, v_1, v_2, v_2, \ldots, v_{2})\nonumber\\
&+ \mu (v_3, v_1, v_2, v_3, \ldots, v_{3})\nonumber\\
& \qquad \vdots  \nonumber\\
&+ \mu [v_n, v_1, v_2, v_3, \ldots, v_{n}].
\end{align}
To find the average AoI ($\Delta_1 = v_0$) we need to solve the system of equations in \eqref{multi-source}, and prove that the solution to $v_i$,  $0 \le i \le n,$ is positive. 
Equations in \eqref{multi-source} are equivalent to 
\begin{align}
n \mu v_{0} &= 1 + \mu \sum_{i=1}^{n} v_{i} , \label{v0}\\
v_{1}(\overline{\lambda^{(1)}}+n\lambda^{(1)})&=1+ \overline{\lambda^{(1)}}v_{2}. \label{eq:v1}
\end{align}
And for $2 \le i \le n$,
\begin{align}
n (\lambda + \mu ) v_i &= 1 + (i-1) \lambda^{(1)} v_i + (n-i+1) \lambda^{(1)} v_{i-1} + i \overline{\lambda^{(1)}} v_{i+1} +
 (n-i) \overline{\lambda^{(1)}} v_i + \mu \sum\limits_{j=1}^{i-1} v_j + (n-i+1) \mu v_i ,
\label{equiv_system}
\end{align}
where $v_{n+1} \triangleq v_0$ and $\lambda=\overline{\lambda^{(1)}}+\lambda^{(1)}=\sum_{i=1}^{n} \lambda_i$.

Let us rewrite the equations using the difference of adjacent $v_j$'s.
From \eqref{equiv_system}, we have for $2 \le i \le n$,
\begin{align}\label{p}
 ((n-i+1)\lambda^{(1)}+i\overline{\lambda^{(1)}}+(i-1)\mu) v_{i} =  1+\lambda^{(1)}(n-i+1)v_{i-1}+
 i\overline{\lambda^{(1)}}v_{i+1} + \mu \sum_{j=1}^{i-1} v_{j} . 
 \end{align}
We plug in $i+1$ in \eqref{p} and subtract the resulting equation from Equation \eqref{p}. Therefore,
\begin{align} \label{main}
&((n-i)\lambda^{(1)}+i\overline{\lambda^{(1)}}+i\mu) (v_{i+1}-v_{i})  \nonumber \\
 =& \lambda^{(1)}(n-i+1) (v_{i}-v_{i-1}) + (i+1)\overline{\lambda^{(1)}} (v_{i+2}-v_{i+1})
\end{align}
%Equation \eqref{main} holds for $i \in \{2,...,n-1\}$ with $v_{n+1}\triangleq v_{0}$. 
Let us define
$w_{i}=v_{i}-v_{i-1}$ ($2 \leq i \leq n-1$), then we have:
\begin{align*}
& ((n-i)\lambda^{(1)}+i\overline{\lambda^{(1)}}+i\mu) w_{i+1} = \lambda^{(1)}(n-i+1) w_{i} + (i+1)\overline{\lambda^{(1)}} w_{i+2}  
\end{align*}
Define for each $i \in \{2,...,n-1\}$, coefficients $r_{i+2}= \frac{(n-i)\lambda^{(1)}+i\overline{\lambda^{(1)}}+i\mu}{(i+1)\overline{\lambda^{(1)}}}$ and $t_{i+2}= -\frac{\lambda^{(1)}(n-i+1)}{(i+1)\overline{\lambda^{(1)}}}$, then
\begin{align}
     w_{i+2} = r_{i+2} w_{i+1}+ t_{i+2} w_{i}, \text{ for } 2 \le i \le n-1. \label{eq:w}
\end{align}
% \begin{align} \label{w}
% w_{4}=& a_{4} w_{3} + b_{4} w_{2} \\
% w_{5}=& a_{5} w_{4} + b_{5} w_{3} \\
% \vdots  \\
% w_{n+1}=&  a_{n+1} w_{n} + b_{n+1} w_{n-1}
% \end{align}

To show that each $v_j$ is positive and also to determine its value, our proof is inductive. For the base case, we find the value of $v_1$ and show it is positive. Then, using induction, assuming $v_1,...v_j$ are positive and we know their values, we find $v_{j+1}$ and prove that it is positive.

{\bf Base case.} We will find $v_1$ and show that $v_1>0$.
As we can see from the recursive equations in \eqref{eq:w}, we can express each $w_{j}$ for $j \in \{4,..,n+1\}$ in terms of $w_{2}$ and $w_{3}$. Write such expressions as $w_{j}= x_{j} w_{3} +y_{j} w_{2}$ where
 $x_{4}=r_{4}$, $y_{4}=t_{4}$, $x_{5}=r_{5}r_{4}+t_{5}$, $y_{5}=r_{5}t_{4}$ and $x_{j+1}= r_{j+1} x_{j}+ t_{j+1} x_{j-1}$, $y_{j+1}= r_{j+1} y_{j}+t_{j+1}y_{j-1}$
for $5 \leq j \leq n$.
So far we can write $w_j$ for $ 4 \le j \le n+1$  as a linear function of $w_2$ and $w_{3}$ which are in fact 
linear functions of $v_{1}, v_{2}$ and $v_{3}$ because $w_{2}=v_{2}-v_{1}$ and $w_{3}=v_{3}-v_{2}$.
We also know from \eqref{eq:v1} and \eqref{p} for $i=2$:
%$v_{1}(\overline{\lambda^{(1)}}+n\lambda^{(1)})=1+ \overline{\lambda^{(1)}}v_{2}$. And also:
\begin{align}\label{eq:v3}
((n-1)\lambda^{(1)}+2\overline{\lambda^{(1)}}+\mu) v_{2}= 1+\lambda^{(1)}(n-1)v_{1}+2\overline{\lambda^{(1)}}v_{3}+\mu v_{1}.
\end{align}
Combining \eqref{eq:v1} and \eqref{eq:v3} together we reach the conclusion that we can write $v_2, v_3,$ and all the
$w_{i}$, $2 \le i \le n+1$, based on $v_{1}$. Hence for some coefficients $c_i,d_i$, we write $$w_{i}=c_i v_1 +d_i.$$ 

Next, using (another) induction we will show that for $i \in \{2,3,...,n+1\}$,
\begin{align}
  c_i>0 \text{ and } d_i <0 \label{eq:induc}  .
\end{align}
For $i=2,3,$ from equation \eqref{equiv_system} we have
\begin{align}
  &  w_2  =\frac{v_1 n \lambda^{(1)} -1}{\overline{\lambda^{(1)}}} .\\
  &   w_3 =\frac{((n-1)\lambda^{(1)}+\mu) w_2 -1}{2\overline{\lambda^{(1)}}} .
\end{align}
Therefore, $c_2=\frac{n \lambda^{(1)}}{\overline{\lambda^{(1)}}}, d_2 = \frac{-1}{\overline{\lambda^{(1)}}}$ and $c_3 = \frac{n \lambda^{(1)} ((n-1)\lambda^{(1)}+\mu)}{2 \overline{\lambda^{(1)}}^2}, d_3 = -\frac{1}{2\overline{\lambda^{(1)}}}-\frac{((n-1)\lambda^{(1)}+\mu)}{2 \overline{\lambda^{(1)}}^2}$. Hence the claim in \eqref{eq:induc} holds. 

Assume that \eqref{eq:induc} holds for $2,3,\dots,i$, where $3 \le i\leq n$. We will prove that it also holds for $i+1$.
We can rewrite Equation \eqref{p} as %\bl{??? I added brackets for $v_i-v_k$ in \eqref{eq:25} and exchanged the variables $j$ and $k$ so that later you can use this equation. Please double check the correctness??}
\begin{align} 
    1+ i \overline{\lambda^{(1)}} w_{i+1}
    =& (n-i+1) \lambda^{(1)} w_i
    +\mu \sum_{k=1}^{i-1} (v_i -v_k) \label{eq:25}\\
    =& (n-i+1) \lambda^{(1)} w_i
    +\mu \sum_{k=1}^{i-1} \sum_{j=k+1}^{i} w_j \label{fin} \\
    =& c v_1 + d,
\end{align}
for some constants $c>0,d<0$. The last equality follows from the induction hypothesis \eqref{eq:induc} and the fact that \eqref{fin} consists of $w_j$'s where $j\leq i$. The above equation implies $c_{i+1}>0, d_{i+1}<0$. Therefore by induction the condition in \eqref{eq:induc} holds.

From \eqref{v0}, %and also knowing that $w_{n+1}=v_{0}-v_{n}$, by substituting $w_{i}$ with $c_i v_{1}+d_i$ in these $2$ equations and solving it for $v_{1},v_{0}$ we can finally find the value of $v_{0}$ which is the AoI in the monitor.
\begin{align}  
v_{0} &= \frac{1}{n \mu} + \frac{\sum_{i=1}^{n} v_{i}}{n}
 = \frac{1}{n \mu} + \frac{\sum_{j=2}^{n} \sum_{k=2}^{j} w_k}{n} +v_1 \nonumber\\ 
 &= \frac{1}{n \mu} +v_1 + \sum_{j=2}^{n}  \frac{n-j+1}{n} w_j. \label{eq:28}
 \end{align}
 Moreover, 
 \begin{align}
v_{0} = v_{n+1} &= w_{n+1} + v_n = w_{n+1} + \sum_{j=2}^{n} w_j  + v_1. \label{eq:29}
\end{align} 
Comparing \eqref{eq:28}, \eqref{eq:29} and using $w_j = c_i v_1 + d_i$, we have
\begin{align}
    \frac{1}{n \mu } = \sum_{j=2}^{n+1}  \frac{j-1}{n} w_j =  v_1 \sum_{j=2}^{n+1}  \frac{j-1}{n} c_j + \sum_{j=2}^{n+1}  \frac{j-1}{n} d_j.
\end{align} 
We can obtain $v_1$ by 
\begin{align}\label{init}
    v_1 = \frac{\frac{1}{n\mu} - \sum_{j=2}^{n+1}  \frac{j-1}{n} d_j}{\sum_{j=2}^{n+1}  \frac{j-1}{n} c_j}.
\end{align}
It can be seen that by the condition of \eqref{eq:induc}, $v_1$ is positive and we found its value in \eqref{init}.

{\bf Induction step.}
We assume that we obtained the values of $v_{1},...,v_{j-1}$ and they are positive. We need to show that $v_{j}$ is positive and find its value. From now on, $v_1,\dots,v_{j-1}$ are considered positive constants.

From \eqref{eq:w} and considering that $v_1,\dots,v_{j-1}$ are positive constants, it is obvious that we can write for $j \le i \le n$ and some constants $a_i,b_i$,
$$w_i=a_i v_j + b_i.$$ 
Next, We prove by (another) induction that for $j \le i \le n+1,$ 
\begin{align} \label{eq:induc2}
 a_i>0 \text{ and } b_i <0.   
\end{align}
Since $w_j= v_j - v_{j-1}$ and also $v_{j-1}$ is assumed to be a positive constant, the condition in \eqref{eq:induc2} is true for $j$. 

We assume \eqref{eq:induc2} holds for $j,j+1, \dots,i$, and prove it for $i+1$. We make use of \eqref{eq:25} again:
\begin{align}
     1 + i \overline{\lambda^{(1)}} w_{i+1} 
    &=  (n-i+1) \lambda^{(1)} w_i
    +\mu \sum_{k=1}^{i-1} (v_i -v_k) \\
    &= (n-i+1) \lambda^{(1)} w_i
    +\mu \sum_{k=1}^{j-1}  (v_i -v_k) + \mu \sum_{k=j}^{i-1}  (v_i -v_k)  \\
     &= (n-i+1) \lambda^{(1)} w_i
    +\mu \sum_{k=1}^{j-1} \left( \sum_{r=j+1}^{i} w_r + v_j - v_k  \right)
     + \mu \sum_{k=j}^{i-1} \sum_{r=k+1}^i w_r  \label{eq:35}\\
     &=a_{i+1} v_j+b_{i+1},
\end{align}
where $a_{i+1}>0, b_{i+1}<0$ are some constants. The last step holds because the right hand side of \eqref{eq:35} consists of $v_k$ ($1 \le k \le j-1$) and $w_r$ ($j+1 \le r \le i$), and we assumed the condition \eqref{eq:induc2} holds and $v_k$ is a positive constant for $k < j$. From the above equation, the condition in \eqref{eq:induc2} holds for $w_{i+1}$. Thus, we have proved the condition in \eqref{eq:induc2} by induction.
%We will show $v_i - v_k = x_k v_j + y_k$ for some constants $x_r >0 , y_r<0$, when $k \in \{1,...,i-1\}$. If $k \geq j$ 
% \begin{align*}
%     v_i - v_k = w_i + w_{i-1}+...+w_k
% \end{align*}
% which based on the assumption of the induction each of $w_l$ in the right hand side are in a form of a positive number multiplied by $v_j$ plus a negative number. If $k<j$ then 
% \begin{align}
%     v_i - v_k = w_i + w_{i-1}+...+w_{j+1} + v_j-v_{k}
% \end{align}
%Again each of $w_l$ in the right hand side are in a form of a positive number multiplied by $v_j$ plus a negative number and also since $k<j$ based on the assumption of the first induction it is positive and therefore $v_j - v_k$ is a positive number multiplied by $v_j$ plus a negative number. Consequently, we showed each $w_k$ in a form of $a_k v_j + b_k$ where $a_k>0, b_k <0$ and $k \in \{j,...,n+1\}$ assuming $v_1,...,v_{j-1}$ are positive constant.

%\bl{??I removed the part $v_i - v_k = x_k v_j + y_k$. However, in the algorithm seems that you used this relation. Please try to use $a_i,b_i$ instead of $x_i,y_i$. Or please add the part $v_i - v_k = x_k v_j + y_k$ in a clearer way.??}

Now, we intend to prove $v_j$ is positive and find its value recursively based on $v_1,...,v_{j-1}$ and some constants. Similar to the way we found $v_1$ as in \eqref{eq:28}, \eqref{eq:29}, we write \eqref{v0} as
\begin{align} \label{end}
    v_{0} &= \frac{1}{n \mu} + \frac{\sum_{i=1}^{n} v_{i}}{n}
 = \frac{1}{n \mu} + \frac{v_j}{n} + \sum_{r=j+1}^{n}  \frac{\sum_{k=j+1}^{r} w_{k}+v_j}{n}.% + \frac{w_{j+2}+w_{j+1}+v_j}{n}+...+ \frac{w_{n}+...+w_{j+1}+v_j}{n}.
\end{align}
Moreover,
\begin{align} 
v_{0} &= w_{n+1} + v_n = w_{n+1} + \sum_{k=j+1}^{n} w_k  + v_j. \label{eq:38}
\end{align}
Combining \eqref{end}, \eqref{eq:38}, we have
\begin{align}\label{vj}
    \frac{1}{n \mu} +\frac{\sum_{i=1}^{j-1} v_i}{n} 
    %= v_j (1- \frac{n-j+1}{n}) +w_{n+1}+ \sum_{i=j+1}^{n} w_i (1-\frac{n-i+1}{n}).
    =   \frac{j-1}{n} v_j +w_{n+1}+ \sum_{i=j+1}^{n}  \frac{i-1}{n} w_i.
\end{align}
From \eqref{vj} we can write 
\begin{align}
    v_j = \frac{ \frac{1}{n \mu} +\frac{\sum_{i=1}^{j-1} v_i}{n}- \sum_{i=j+1}^{n} \frac{i-1}{n}b_{i}}{\frac{j-1}{n}+\sum_{i=j+1}^{n} \frac{i-1}{n}a_{i}},
\end{align}
where the denominator and the numerator are both positive by condition \eqref{eq:induc2}. Therefore, we proved that $v_j$ is positive and also found its value recursively. The solution to $v_1,\dots,v_n$ using the recursive calculation and the age of information $\Delta_1=v_0$ is summarized in Algorithm \ref{alg:homo}.
\end{proof}

Let $\Delta_{i}$ denote the average AoI at the monitor for Source $i$.
In the corollaries below we state the average AoI for $n=2,3$ servers, which can be directly derived from Theorem \ref{thm:homo_general}. Here we define $\lambda \triangleq \lambda^{(1)}+\lambda^{(2)}+\dots+\lambda^{(m)}$, and recall each server has service rate $\mu$.
\begin{corollary}
\label{prop:all_2_servers}
For $m$ information sources and $n=2$ servers, we have
\begin{align}
\Delta_{i} = \frac{1}{2 (\lambda + \mu)} + \frac{\lambda + \mu}{2 \mu \lambda^{(i)}}, \quad 1 \leq i \leq m.
\label{all_2_servers}
\end{align}
\end{corollary}
In Figure \ref{homo2}, we observe that as we increase $\rho^{(1)}$, average AoI for Source $1$ decreases. Also, average AoI for Source $2$ increases since $\rho$ is constant which matches our intuition.
\begin{figure}
\centering
\includegraphics[width=0.45\textwidth]{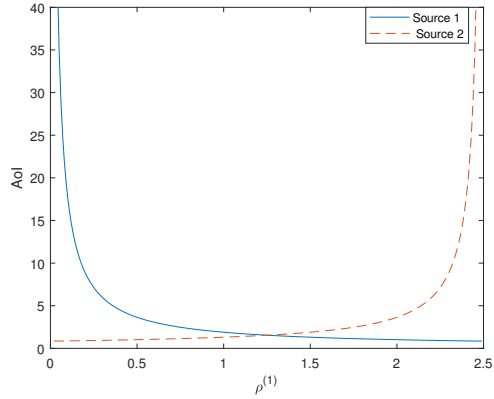}
\caption{Average AoI for Source $1$ and Source $2$ when $\rho=2.5$, $\mu=1$, and $n=2$. Here, $\rho^{(1)}$ is changing from $0$ to $2.5$.}
\label{homo2}
\end{figure}
Next,  we determine the optimal arrival rate given the sum arrival rate when $n=2$.
\begin{theorem} \label{convex}
Consider $m$ information sources and $n = 2$ servers. The optimal arrival rate ${\lambda^{(i)}}^*$ minimizing the weighted sum of AoIs in Corollary \ref{prop:all_2_servers}, i.e., $w_1 \Delta_1 + w_2 \Delta_2+...+w_m \Delta_m$ for $w_i \geq 0$, subject to the constraint $\lambda^{(1)}+ \lambda^{(2)}+...+ \lambda^{(m)} = \lambda$, is given by
\begin{align*}
{\lambda^{(i)}}^{*}=\frac{\lambda \sqrt{w_i} }{\sum_{j=1}^{m} \sqrt{w_j}},  i \in [m]. 
\end{align*}
\end{theorem}
\begin{proof}
The objective function that we are trying to minimize is convex and therefore we just have to set the partial derivative with respect to each $\lambda^{(i)}$ to be zero.
\begin{align} \label{deri}
 \frac{\partial}{\partial \lambda^{(i)}}   ( w_1 \Delta_1 + w_2 \Delta_2+...+w_m \Delta_m +a(\sum_{j=1}^{m} \lambda^{(j)}-\lambda)) =0,
\end{align}
for i $\in [m]$. Here $a$ is the Lagrange multiplier. Simplifying \eqref{deri} results in:
\begin{align}
    \frac{w_{1}}{(\lambda^{(1)})^{2}}= \frac{w_{2}}{(\lambda^{(2)})^{2}}=\dots=\frac{w_{n}}{(\lambda^{(n)})^{2}}=a.
\end{align}
Knowing the fact that $\lambda^{(1)}+ \lambda^{(2)}+...+ \lambda^{(m)} = \lambda$, we obtain the result in Theorem~\ref{convex}.
\end{proof}

From Theorem \ref{convex}, when there are 2 servers and the weights are all identical, i.e., $w_1=w_2=\dots=w_m$, the optimal arrival rate should be equal for all sources. In general, the optimal arrival rate is inversely proportional to the square root of the weight.

\section{AoI in Single-Source Heterogeneous Networks} \label{hetro-sec}
\subsection{Overview}
In this section, we consider a single source and assume that the arrival rates and service rates of servers are arbitrary. 
We denote by $\lambda_{j}^{(1)} \triangleq \lambda_j$ the arrival rate of a single source at Server $j$, and $\mu_j$ the service rate of Server $j \in [n]$.
For this setting, we can no longer use the technique of virtual servers used in the homogeneous case to reduce the state space and derive AoI. In particular, we need to keep track of the age of updates at the physical servers as well as their ordering, resulting in $n!$ number of states. However, we can still use fake update and fake preemption such that the server is always busy even after its update is delivered or is outdated.
If we consider $n$ servers, we will have $n!$ states, $(2n)n!$ transitions and $(n+1)n!$ equations and unknowns. Writing down the $(n+1)n!$ equations from Lemma \ref{lem:yates} in a matrix form, we obtain $T \mathbf{v} = \mathbf{ \pi }$, for the coefficient matrix $T$ and the steady state probabilities $\bm{\pi}$. First, we find steady state probabilities. Then, we prove that we can break down matrix $T$ into sub-matrices which have the same general form as some matrix $T_0$. Afterwards, by doing some column and row operations on matrix $T_0$, we show that we are able to solve all the equations and eventually find the average AoI.

In the following subsections, we present the notations, the main theorem, and examples with 2 and 3 servers. %and finally the proof for general numbers of servers.

\subsection{Notations and definitions.}\label{subsec:notation} 
A permutation of the set $\{1,2,...,n\}$ is denoted by a lower-case letter or a tuple of length $n$, e.g., $q=(q_1,q_2,\dots,q_n)$. Additionally, set by default $q_0=0$.
A permutation $q$ is said to be \emph{$j$-increasing} if the last $j$ positions are increasing:
$$q_{n-j+1} < q_{n-j+2}< \dots < q_{n}.$$
If a permutation is $2$-increasing, it is said to be \emph{odd}. Otherwise, it is \emph{even}.
Define a permutation $h_i(\cdot)$ that takes the $i$-th element of $q$ and place it at the first position: 
\begin{align*}
h_i(q)=(q_{i},q_{1},...,q_{i-1},q_{i+1},...,q_{n}),  
\end{align*}
for $i \in [n]$. Let its inverse be $h^{-1}(\cdot)$. Define the set
\begin{align}\label{eq:H_q_-1}
    H^{-1}_q = \{h^{-1}_i(q): i=1,2,\dots,n\}.
\end{align}
Define a function $g_{j,k}(\cdot)$ that takes the $k$-th element of $q$ and place it at the $(n-j)$-th position:
\begin{align}
g_{j,k}(q) =& (q_1,...,q_{n-j-1}, q_k, q_{n-j},...,q_{k-1},q_{k+1},...,q_n ), \quad k = n-j,...,n. 
\end{align}
Denote by $g_{j,k}^{-1}(\cdot)$ the inverse permutation.

Given set of linear equations:
$$A\mathbf{v}  = \mathbf{b},$$
the matrix $A$ is called the coefficient matrix, $\mathbf{v}$ the variable vector, and $\mathbf{b}$ the constant vector. 
The matrices and vectors will be indexed by permutations and/or integers. 
Let $m,n$ be the row index and the column index for a matrix $A$, then $A(m,n)$ is the $(m,n)$-th entry. 
Let $M,N$ be sets of rows and columns indices for a matrix $A$, then $A(M,N)$ is the corresponding submatrix of $A$ with rows $M$ and columns $N$. 
Moreover, $A(:,N)$ is the submatrix of $A$ with columns N, and $A(M,:)$ is the submatarix with rows $M$. 
For a vector $\mathbf{v}$, its $n$-th entry is denoted by $v_n$, and its sub-vector indexed by $N$ is denoted by $\mathbf{v}(N)$.

\subsection{Main result} \label{ez}
In this subsection, we derive the algorithm to compute the AoI of the heterogeneous network.
To simplify the presentation, the proofs for the results are provided in the appendix. 

First, let us describe the SHS. 
The continuous state $\mathbf{x}=(x_0,x_1,\dots,x_n)$ represents the ages of the monitor, Server 1, ..., and Server $n$. 
 % which $(n-1)!$ of them start with $1$, $(n-1)!$ start with $2$,..., and $(n-1)!$ start with $n$. %Let $i$ be $\in (1,2,...,n)$ which here is indicator of a server number.
 The set of discrete states $\mathcal{Q}$ is the set of all permutations of the set $\{1,2,...,n\}$. There are $n!$ states in total. 
 State $q=(q_{1},q_{2},...,q_{n})\in \mathcal{Q}$ represents the ordering of the age among all the servers, meaning $x_{q_{1}} \le x_{q_{2}} \le ...  \le x_{q_{n}}$. %Also, we know that $|\sigma_{n}|=n!$ and %$|T(i,\sigma_{n}^j)|=(n-1)!$ 		 $\forall i \in (1,2,...,n)$ and 
%and $\sigma_{n}= T(1,\sigma_{n}) \cup T(2,\sigma_{n}) \cup ... \cup T(n, \sigma_{n})$.
%There are $n$ of such $i$ and we name the set of all these $h_i(q)$ as $H_q$ including state $q$ as well. 

The incoming transitions of state $q=(q_1, q_2,...,q_n)$ are listed in Figure \ref{fig:hete_transitions}. Here for an incoming state $p$, $\mathbf{v}_p A_l$ corresponds to the last term in Equation \eqref{eq:yateslemma}.
For ease of exposition, the entries in vector $\mathbf{x}$ are reordered as $(x_0, x_{q_1},\dots,x_{q_n})$. By abuse of notation, in Figure \ref{fig:hete_transitions}, the reordered vector is still called $\mathbf{x}$. Similarly, $\mathbf{x}',\mathbf{v}_p$ are also reordered.

For transition $l, 1 \le l \le n,$  state $p= (q_2,...,q_{l-1},q_1,q_{l},...,q_n)$ is an incoming state of state $q$,  corresponding to an update arrival at server $q_1$ with rate $\lambda_{q_1}$. 
The $q_1$-th entry in $\mathbf{x}'$ becomes 0. Accordingly, the $q_1$-th entry in $\mathbf{v}_p$ becomes 0. The set of incoming states of $q$ for such transitions can be represented as $H^{-1}_q$.

For transition $l, n+1 \le l \le 2n$, set $i=l-n$. An update is delivered to the monitor from Server $q_i$ with rate $\mu_{q_i}$, and $q$ is an incoming state to itself. 
In this case, we preempt any update in the servers that has larger information age and put a fake update in them which is the update from Server $q_i$. In other words, we preempt updates in servers $(q_{i+1},...,q_{n})$ and replace them with the update from Server $q_i$. Therefore, the new vector $\mathbf{x}^{\prime}$ becomes $(x_{q_i},x_{q_1},\dots,x_{q_{i-1}},$ $x_{q_i},\dots,x_{q_i})$. Similarly the corresponding  vector $\mathbf{v}_p$ changes as $(v_{p,q_i},v_{p,q_1},\dots,v_{p,q_{i-1}},$ $v_{p,q_i},\dots,v_{p,q_i})$.

\begin{figure}
\begin{tabular}{ cccccccccc }
\centering
 $l$ & $\lambda^{(l)}$ & Transition &$\mathbf{x}^\prime$ =$\mathbf{x}A_{l}$ & $\mathbf{v}_{p} A_{l}$&\\          \hline
 $1$ & $\lambda_{q_1}$ & $q \leftarrow p=q $ & $(x_{0},0,x_{q_2},...,x_{q_n})$ &$(v_{p,0},0,v_{p,q_2},...,v_{p,q_n})$&\\ \hline
 $2$ & $\lambda_{q_1}$ & $q \leftarrow p=(q_2,q_1,q_3,...,q_n) $ & $(x_{0},0,x_{q_2},...,x_{q_n})$ &$(v_{p,0},0,v_{p,q_2},...,v_{p,q_n})$&\\ \hline
 $3$ & $\lambda_{q_1}$ & $q \leftarrow p=(q_2,q_3,q_1,...,q_n) $ & $(x_{0},0,x_{q_2},...,x_{q_n})$ &$(v_{p,0},0,v_{p,q_2},...,v_{p,q_n})$&\\ \hline
 & $\vdots$ & $\vdots$ & $\vdots$\\  \hline
  $n$ & $\lambda_{q_1}$ & $q \leftarrow p=(q_2,q_3,...,q_n,q_1) $ & $(x_{0},0,x_{q_2},...,x_{q_n})$ &$(v_{p,0},0,v_{p,q_2},...,v_{p,q_n})$&\\ \hline
  $n+1$ & $\mu_{q_1}$ & $q \leftarrow p=q $ & $(x_{q_1}, x_{q_{1}},x_{q_1},...,x_{q_1})$ &$(v_{q,q_1}, v_{q,q_{1}},v_{q,q_1},...,v_{q,q_1})$&\\ \hline
 $n+2$ & $\mu_{q_2}$ & $q \leftarrow p=q $ & $(x_{q_2}, x_{q_1},x_{q_2},...,x_{q_2})$ &$(v_{q,q_2}, v_{q,q_1},v_{q,q_2},...,v_{q,q_2})$&\\ \hline
 & $\vdots$ & $\vdots$ & $\vdots$\\  \hline
  $2n$ & $\mu_{q_n}$ & $q \leftarrow p=q $ & $(x_{q_n}, x_{q_1},x_{q_2},...,x_{q_n})$ &$(v_{q,q_n}, v_{q,q_1},v_{q,q_2},...,v_{q,q_n})$&\\ \hline

\end{tabular}
\caption{Incoming transitions of any given state $q$ caused by update arrival or update delivery. The incoming state is denoted as $p$.}
\label{fig:hete_transitions}
\end{figure}

Now, we write down Equation \eqref{eq:yateslemma} as in Lemma \ref{lem:yates} for each state $q \in \mathcal{Q}$. 
Notice that each update arrival or update delivery results in an outgoing state for state $q$.
Hence on the left-hand side of Equation \eqref{eq:yateslemma}, 
$\mathbf{v_{q}}$ is multiplied by sum of rates of outgoing transitions which for every $q \in \mathcal{Q}$ is equal to $\sum_{j=1}^{n} \lambda_{q_{j}} +\sum_{j=1}^{n} \mu_{q_{j}}$. Also, $\mathbf{b}_{{q}}=[1,...,1]$ due to the fake update, and $\pi_{q}$ is  the stationary distribution to be computed by Lemma \ref{ss}. The last term on the right-hand side of \eqref{eq:yateslemma} can be expressed according to Figure \ref{fig:hete_transitions}. Therefore, for $q \in \mathcal{Q}$,
\begin{comment}
\begin{align} 
\mathbf{v}_{q} (\sum_{j=1}^{n} \lambda_{q_{j}} +\sum_{j=1}^{n} \mu_{q_{j}})
&=\pi_{q} +\lambda_{q_{1}} (\sum_{p \in H^{-1}_{q}}^{} \mathbf{v}_{p}) \\ \nonumber &+ \sum_{i=1}^{n} {\mu_{q_{i}} F(\mathbf{v}_q) } \quad j \in (1,2,...,n!).
\end{align}
\end{comment}
\begin{align}\label{general_eq}
(v_{q,0}, v_{q,q_1}, v_{q,q_2},...,v_{q,q_n}) (\sum_{j=1}^{n} \lambda_{q_{j}} +\sum_{j=1}^{n} \mu_{q_{j}})
&=\pi_{q} +\lambda_{q_{1}} (\sum_{p \in H^{-1}_{q}}^{} (v_{p,0}, 0, v_{p,q_2},...,v_{p,q_n})) \\ \nonumber &+ \sum_{i=1}^{n} {\mu_{q_{i}} (v_{q,q_i}, v_{q,q_1},..., v_{q,q_{i-1}},v_{q,q_i},...,v_{q,q_i}) }. %\quad j \in (1,2,...,n!).
\end{align}

The following lemma gives the steady-state probability, which only depends on arrival rates $\lambda_{i}$ and the order of the update's age in a state.
\begin{lemma}
\label{ss}
For a given state ${q}= (q_1, q_2,..., q_n)$ in which ${q} \in \mathcal{Q}$, the steady state probability ($\pi_{q}$) is
\begin{equation}
\pi_{q}= \frac{\lambda_{q_{1}}}{\sum_{j=1}^{n} \lambda_{q_{j}}} \frac{\lambda_{q_{2}}}{\sum_{j=2}^{n} \lambda_{q_{j}}} \frac{\lambda_{q_{3}}}{\sum_{j=3}^{n} \lambda_{q_{1}}} ... \frac{\lambda_{q_{n-1}}}{\sum_{j=n-1}^{n} \lambda_{q_{j}}}
\end{equation}
\end{lemma}
% \begin{proof}
% The proof is shown in Appendix \ref{apen}.
% \end{proof}

In the next theorem we represent the equations of \eqref{general_eq} in matrix form as $T\mathbf{v} = \bm{\pi}$ for some coefficient matrix $T$ and some constant vector $\bm{\pi}$. In total, there are $(n+1)!$ equations since there $n!$ states and each  $\mathbf{v}_q=(v_{q,q_0}, v_{q,q_1},...,v_{q,q_n})$ has $n+1$ entries. 
%We represent the general coefficient matrix $T$ based on the following notation. There are $n!$ vector variables where each vector is of size $n+1$ resulting in $(n+1)!$ equation and variables.
%For state $q=(q_1,q_2,...,q_n)$, its corresponding vector variable is $\mathbf{v}_q= (v_{q,q_0}, v_{q,q_1},...,v_{q,q_n})$. 
We represent the row and column indices of matrix $T$ using $2$ tuples of $(q,i)$ and $(p,k)$ where $p,q$ are any $2$ arbitrary permutations and and $i,k$ are numbers in $\{0,1,...,n\}$. 
In particular, variable $v_{p,p_k}$ corresponds to column index $(p,k)$ in the coefficient matrix, and the $i$-th equation (out of $n+1$) in equation \eqref{general_eq} corresponds to row $(q,i)$. 
Accordingly, vectors $\mathbf{v}, \bm{\pi}$ are indexed by $(p,k)$.
%\bl{??? The $i$-th equation corresponds to $v_{q,q_i}$ on the left hand side of \eqref{general_eq}??}

%With some algebra we reach to the following formula for calculating AoI:
%\begin{align}
%AoI= \frac{1}{\sum_{i=1}^{n} \mu_{i}} + \frac{\sum_{i=1}^{n} \mu_{i}(\sum_{j=1}^{n!} v_{\sigma_{n}^{(j)}i})}{\sum_{i=1}^{n} \mu_{i}}.
%\end{align}

\begin{lemma} \label{gen_mat}
The $(n+1)!$ transition equations in \eqref{general_eq} can be written as
\begin{align}
    T \mathbf{v} = \bm{\pi}.
\end{align}
Here the constant vector $\bm{\pi}$ has entry $\pi_p$ in row $(p,k)$, for all $0 \le k \le n$, and any permutation $p$.
And the coefficient matrix $T$ is as follows:
\begin{align} \label{co_t}
T((q,i),(p,k))=
\begin{cases}
\sum_{l=2}^{n} \lambda_{q_{l}}+ \sum_{l=1}^{n} \mu_{q_{l}}
,& \text{if } i=0 \quad \text{and} \quad (q,i)=(p,k),\\
-\lambda_{q_{1}} ,& \text{if } i=k=0 \quad \text{and} \quad q = h_j(p), \quad \text{for} \quad j=2,\dots,n,\\
-\mu_{q_{k}} ,& \text{if } i=0 \quad \text{and} \quad q =p , \quad \text{for} \quad k=1,\dots,n,\\
\sum_{j=1}^{n} \lambda_{q_{j}}
,& \text{if } i=1 \quad \text{and} \quad (q,i)=(p,k),\\
\sum_{l=2}^{n} \lambda_{q_{l}}  +\sum_{l=1}^{i-1} \mu_{q_{l}}, & \text{if } i>1 \quad  \text{and} \quad (q,i)=(p,k)  \\
-\lambda_{q_{1}}, & \text{if } i>1 \quad \text{and} \quad q = h_j(p), k = \langle i \rangle_j, \quad \text{for} \quad j=2,\dots,n,\\
- \mu_{q_{k}}, & \text{if } i>1 \quad  \text{and } \quad  q=p \quad  \text{for} \quad k=1,2,...,i-1, \\
0,&  \quad \text{o.w.} \\
\end{cases}
\end{align}
Here $k = \langle i \rangle_j$ means that $k=i-1$ if $i \le j$, and $k=i$ if $i > j$.
\end{lemma}
% \begin{proof}
% The proof is shown in Appendix \ref{apen}.
% \end{proof}

Next we show that solving the original set of equations simplifies to solving smaller sets of equations separately. 
In Algorithm \ref{alg:heter_overall}, we break down the $(n+1)!$ equations into smaller sets to solve all variables $v_{q,q_i}$ with fixed $i$ and fixed $q_i,q_{i+1},\dots,q_{n}$ (Line \ref{line:breakdown}). Namely, we solve $(i-1)!$ variables at a time, for $1 \le i \le n$. 
After solving these variables, we remove them from the equations and update the constant vector $\mathbf{c}$ as in Line \ref{line:update_c}. Finally, the AoI equals the average of $v_{q,0}$'s, which requires solving $n!$ equations as in Line \ref{line:AoI}.

\begin{algorithm}
\begin{algorithmic}[1]
\For{$i=1,2,\dots,n$,}
    \For{distinct $c_{i},\dots,c_n \in [n]$}
        \State $N \gets \{(q,i):(q_{i},\dots,q_n)=(c_{i},\dots,c_n)\}$ \label{line:N_set}
        \State $\overline{N} \gets \{(q,i): (q_{i},\dots,q_n)\neq (c_{i},\dots,c_n) \}$
        \State Solve $T(N,N)\mathbf{v}(N) = \mathbf{c}(N)$ (will use Algorithm \ref{alg:heter}) \label{line:breakdown}
        \State $\mathbf{c}(\overline{N}) \gets \mathbf{c}(\overline{N}) + T(\overline{N},N)\mathbf{v}(N)$ \label{line:constant2} \label{line:update_c}
    \EndFor
\EndFor
\State $N \gets \{(q,0): \text{ all permutations } q\}$
\State $\triangleright$ $\mathbf{v}(N) =\{v_{q,0}: \text{ all permutations } q\}$
\State $AoI \gets \sum_{q} v_{q,0}$ where $T(N,N) \mathbf{v}(N) = \mathbf{c}(N)$ (will use Algorithm \ref{alg:heter}) \label{line:AoI}
\end{algorithmic}
\caption{AoI calculation of $n$-server heterogeneous network.}
\label{alg:heter_overall}
\end{algorithm}

The breakdown is justified in Lemma \ref{break_gen}. We show that the $(i-1)!$ equations in Line \ref{line:breakdown} and the $n!$ equation in Line \ref{line:AoI} have coefficient matrices in the same form, denoted as $T_0$. The equations defined by $T_0$ will be solved by Algorithm 3 explained later.

\begin{lemma} \label{break_gen}
Define $T_0$ parameterized by $i$ to be the following $i! \times i!$ matrix,
\begin{align}\label{eq:T_0}
T_{0}(q,p)=
\begin{cases}
\sum_{j=2}^{n} \lambda_{q_j}  +\sum_{j=1}^{n} \mu_{q_j} , & \text{if } q=p, \\
-\lambda_{q_{1}}, & \text{if } q = h_j(p), j=2,\dots,i,\\
0 , & \text{o.w.}
\end{cases}
\end{align}
Moreover, we define $T_0$ parameterized by $i=0,1$ to be the scalar 
\begin{align}\label{eq:T_0_n1}
    T_0 = \mu_1.   
\end{align}
% and $T_0$ parameterized by $i=0$ to be the scalar 
% \begin{align}\label{eq:T_0_n0}
%     T_0 = \lambda_1+ \mu_1.    
% \end{align}
Solving the set of equations in Lemma \ref{gen_mat} is equivalent to solving the equations corresponding to $T_0$ parameterized by $0,1,2,\dots,n$, shown in Lines \ref{line:breakdown} and \ref{line:AoI} of Algoirthm \ref{alg:heter_overall}.
\end{lemma}
% \begin{proof}
% You can find the proof in Appendix \ref{apen}.
% \end{proof}

\begin{algorithm}
\begin{algorithmic}[1]

\Function{HeteroSolver}{$n, T_{0}, \mathbf{c}^{(0)}$}
    \State $\triangleright$ Solve the equation $T_0 \mathbf{v}^{(0)} = \mathbf{c}^{(0)}$. 
    \State $\triangleright$ Output:
    $\{v_q^{(0)}: \text{all } q\}$, and $v_{(1,2,\dots,n)}^{(n-1)} =\sum\limits_{\text{all permutations } q}v_q^{(0)}$.
    \\
    
    \State $\triangleright$ Base cases:
    \If{$n=0$ or $1$}
        \State $v_1 \gets \frac{\mathbf{c}^{(0)}}{T_0}$
    \EndIf
    \\
    
    \State $\triangleright$ Forward path:
    \For{j=1,2,\dots,n-1}
      \State $\triangleright$ Column operation:
        \State $T_{j}^{\prime} \gets T_{j-1}$ \label{line:T_j'}
       \For{each $(j+1)$-increasing $p$} 
         \State $T_{j}^{\prime}(:,g_{j,k}(p)) \gets T_{j-1}(:,g_{j,k}(p))-T_{j-1}(:,p)$, for $k=n-j+1, \dots,n$ \label{line:column}
        \EndFor 
        \State $\triangleright$ Row operation:
        \State $T_{j}^{\prime\prime} \gets T_j^{\prime}$
        \For{each $(j+1)$-increasing $q$}
            \State $T_{j}^{\prime\prime}(q,:) \gets \sum_{k=n-j}^{n}(T_{j}^{\prime}(g_{j,k}(q),:))$ \label{line:row}
            \State $c_q^{(j)} \gets \sum_{k=n-j}^{n}c^{(j-1)}_{g_{j,k}(q)}$ \label{line:constant_vector}
        \EndFor 
    \State $\triangleright$ Pick specific rows and columns:
    \State $\triangleright$ Variables $\mathbf{v}^{(j)},\mathbf{v}^{(j-1)}$ satisfy  $T_{j}\mathbf{v}^{(j)}=\mathbf{c}^{(j)}$,
    $R_j\mathbf{v}^{(j-1)}(\overline{Q})=\mathbf{c}^{(j-1)}(\overline{Q}) - S_j \mathbf{v}^{(j)}$
    \State $Q \gets \{q: q \text{ is $(j+1)$-increasing}\}$
    \State $\overline{Q} \gets \{q: q \text{ is $j$-increasing but not $(j+1)$-increasing}\}$ \label{line:Q_bar}
    \State $T_{j} \gets T_{j}^{\prime\prime}(Q,Q)$ \label{line:pick_T_j} 
    %\State $\mathbf{c}^{(j+1)} \gets \mathbf{c}^{(j)}(Q)$ 
    \State $R_{j} \gets T_{j}^{\prime\prime}(\overline{Q},\overline{Q})$ \label{line:pick_R_j}
    \State $S_j \gets T_j^{\prime\prime} (\overline{Q},Q)$
    \EndFor
    \State $\triangleright$ Now $T_{n-1}, \mathbf{c}^{(n-1)}$ are both scalars 
    \State $v_{(1,2,\dots,n)}^{(n-1)} \gets  \frac{\mathbf{c}^{(n-1)}}{T_{n-1}}$ \label{line:sum}
    
    \\
    \State $\triangleright$ Backward path:
    \For{$j=n-1,n-2,\dots,1$}
        \For{distinct $c_{n-j},\dots,c_n \in [n]$ such that $c_{n-j+1} < \dots < c_n $ but $c_{n-j} > c_{n-j+1}$}
        \State $N \gets \{q: (q_{n-j},\dots,q_n)=(c_{n-j},\dots,c_n)\}$ \label{line:N}
        \State $\mathbf{v}^{(j-1)} (N) \gets$ HeteroSolver$(n-j-1, R_j(N,N),\mathbf{c}^{(j-1)}(N)-S_j(N,:)\mathbf{v}^{(j)})$ \label{line:solve_R_j}
        \EndFor
        \State $v^{(j-1)}_{p} \gets v^{(j)}_p - \sum_{k=n-j+1}^{n} v^{(j-1)}_{g_{j,k}(p)}$, for $j+1$-increasing $p$ \label{line:subtract_v_j}
    \EndFor
    
    \EndFunction
    \end{algorithmic}
    \caption{Single-source $n$-server heterogeneous network.}
    \label{alg:heter}
\end{algorithm}
%\bl{Alg 3 and form of $T_0$ should move before the theorem. Proof of the theorem should be move after Lemma \ref{lem:alg2_output}.}

So far, the entire system of equations can be solved once we solve equations defined by $T_0$. In Algorithm \ref{alg:heter}, we provide a recursive method for solving equations defined by $T_0$, which breaks down $T_0$ into matrices in the same form as $T_0$ but with smaller parameters. Thus, the AoI can be expressed by $\sum_{q} v_{q,0}$ (Algorithm \ref{alg:heter_overall} Line \ref{line:AoI}) and computed from Algorithm \ref{alg:heter} Line \ref{line:sum} according to Lemma \ref{lem:alg2_output}. Moreover, Lemma \ref{lem:heter_nonneg} shows the correctness of Algorithm \ref{alg:heter} and non-negativity of the solution.

\begin{lemma} \label{lem:alg2_output}
The result of Algorithm \ref{alg:heter} Line \ref{line:sum} is
\begin{align}
\sum_{\text{all permutations } q} v_q^{(0)} = \frac{\sum_{\text{all permutations } q} c_q^{(0)}}{\sum_{i=1}^n \mu_i}.
\end{align}
\end{lemma}
% \begin{IEEEproof}
% You can find the proof in Appendix \ref{apen}.
% \end{IEEEproof}

\begin{lemma} \label{lem:heter_nonneg}
Consider the following linear equation:
$$T_0 \mathbf{v}^{(0)} = \mathbf{c}^{(0)},$$
where $T_0$ as defined in \eqref{eq:T_0} is parameterized by $n \ge 0$. 
\\
{\bf $\bullet$ Correctness.} Algorithm \ref{alg:heter} finds its solution.
\\
{\bf $\bullet$ Non-negativity.} The solution is non-negative if the entries of $\mathbf{c}^{(0)}$ are non-negative and $\mu_1,\dots,\mu_n,$ $\lambda_1,\dots,\lambda_n > 0$.
\end{lemma}
% \begin{proof}
% You can find the proof in the Appendix \ref{apen}.
% \end{proof}

In summary, we can calculate the AoI by Algorithm \ref{alg:heter_overall}, of which the equations in Lines \ref{line:breakdown} and \ref{line:AoI} are solved by Algorithm \ref{alg:heter}. 
\begin{theorem}\label{thm:heter_AoI}
The AoI of heterogeneous network with one source and $n$ servers is 
\begin{align}\label{eq:heter_AoI}
AoI = \sum_{q} v_{q,0},     
\end{align} 
such that $\mathbf{v}=\{v_{q,i}: q \text{ is a permutation of length $n$}, 0 \le i \le n \}$ satisfy
\begin{align}
    T \mathbf{v} = \bm{\pi},
\end{align}
which is solved by Algorithms \ref{alg:heter_overall} and \ref{alg:heter}.
\end{theorem}
% \begin{IEEEproof}
% The proof is provided in Section \ref{subsec:general_proof_heter} and Appendix \ref{apen}.
% \end{IEEEproof}

\subsection{Cases with $2$ and $3$ Servers}
Before proving that Algorithms \ref{alg:heter_overall} and \ref{alg:heter} solve the $(n+1)!$ equations, we show how they execute when $n=2$ and $3$. From these two examples, we demonstrate the intuition of finding the average AoI, and our proof for the general case follows similar steps. In particular, the lemmas in Section \ref{ez} can be generalized from these two examples.
\begin{example} \label{ex_2} 
In the case of $n=2$, we have only $2$ states: $(1,2)$ and $(2,1)$.
State $(1,2)$ is defined as the state that Server $1$ contains a fresher update compared to Server $2$ and State $(2,1)$ as the state that Server $2$ has the fresher update. Upon arrival of an  update at each server or receipt of an update at the monitor, we observe some self-transitions and intra-state transitions. Transition rates and mappings are illustrated in Table~\ref{table2.1}.
\begin{table}
\centering
\begin{tabular}{ cccccccccc }
%\caption{Transition Table when $n=2$ heterogeneous servers}
 $l$ & $\lambda{(l)}$ & Transition&$\mathbf{x}^\prime$ =$\mathbf{x}A_{l}$ & $\mathbf{v}_{p} A_{l}$&\\          \hline
 $1$ & $\lambda_{1}$ & $(1,2) \leftarrow (1,2) $ & $(x_{0},0,x_{2})$ &$(v_{(1,2),0},0,v_{(1,2),2})$&\\ \hline
 $2$ & $\lambda_{1}$ & $(1,2) \leftarrow (2,1) $ & $(x_{0},0,x_{2})$ &$(v_{(2,1),0},0,v_{(2,1),2})$&\\ \hline
 $3$ & $\mu_{1}$ & $(1,2) \leftarrow (1,2) $&$(x_{1},x_{1},x_{1})$ &$(v_{(1,2),1},v_{(1,2),1},v_{(1,2),1})$&\\  \hline
 $4$ & $\mu_{2}$ & $(1,2) \leftarrow (1,2) $&$(x_{2},x_{1},x_{2})$ &$(v_{(1,2),2},v_{(1,2),1},v_{(1,2),2})$& \\  
\end{tabular}
\vspace{1em}

\begin{tabular}{ cccccccccc }
%\caption{Transition Table when $n=2$ heterogeneous servers}
 $l$ & $\lambda{(l)}$ & Transition&$\mathbf{x}^\prime$ =$\mathbf{x}A_{l}$ & $\mathbf{v}_{p} A_{l}$&\\          \hline
 $1$ & $\lambda_{2}$ & $(2,1) \leftarrow (2,1) $ & $(x_{0},0,x_{1})$ &$(v_{(2,1),0},0,v_{(2,1),1})$&\\ \hline
 $2$ & $\lambda_{2}$ & $(2,1) \leftarrow (1,2) $ & $(x_{0},0,x_{1})$ &$(v_{(1,2),0},0,v_{(1,2),1})$&\\ \hline
 $3$ & $\mu_{2}$ & $(2,1) \leftarrow (2,1) $&$(x_{2},x_{2},x_{2})$ &$(v_{(2,1),2},v_{(2,1),2},v_{(2,1),2})$&\\  \hline
 $4$ & $\mu_{1}$ & $(2,1) \leftarrow (2,1) $&$(x_{1},x_{2},x_{1})$ &$(v_{(2,1),1},v_{(2,1),2},v_{(2,1),1})$& \\  
\end{tabular}

\caption{Incoming transitions of state $(1,2)$ (top) and $(2,1)$ (bottom) caused by update arrival or update delivery. In the top table, $\mathbf{x}=(x_0,x_1,x_2)$. In the bottom table, $\mathbf{x}$ is reordered and we write $\mathbf{x}=(x_0,x_2,x_1)$. Similarly, $\mathbf{x}'$ and $\mathbf{v}_p$ are reordered.}
\label{table2.1}
\end{table}

% \begin{table}
% \centering
% \begin{tabular}{ cccccccccc }
% %\caption{Transition Table when $n=2$ heterogeneous servers}
%  $l$ & $\lambda{(l)}$ & Transition&$\mathbf{x}^\prime$ =$\mathbf{x}A_{l}$ & $v_{p} A_{l}$&\\          \hline
%  $1$ & $\lambda_{2}$ & $(2,1) \leftarrow (2,1) $ & $(x_{0},0,x_{1})$ &$(v_{(2,1),0},0,v_{(2,1),1})$&\\ \hline
%  $2$ & $\lambda_{2}$ & $(2,1) \leftarrow (1,2) $ & $(x_{0},0,x_{1})$ &$(v_{(1,2),0},0,v_{(1,2),1})$&\\ \hline
%  $3$ & $\mu_{2}$ & $(2,1) \leftarrow (2,1) $&$(x_{2},x_{2},x_{2})$ &$(v_{(2,1),2},v_{(2,1),2},v_{(2,1),2})$&\\  \hline
%  $4$ & $\mu_{1}$ & $(2,1) \leftarrow (2,1) $&$(x_{1},x_{1},x_{2})$ &$(v_{(2,1),1},v_{(2,1),1},v_{(2,1),2})$& \\  
% \end{tabular}
% \caption{Incoming transitions of state $(2,1)$ caused by update arrival or update delivery.}
% \label{table2.2}
% \end{table}
Steady states probabilities are found knowing that $\pi_{(1,2)}+\pi_{(2,1)}=1$ and $\pi_{(1,2)} \lambda_{2}= \pi_{(2,1)} \lambda_{1}$. Therefore, we have $[{\pi}_{(1,2)},\pi_{(2,1)}]=[\frac{\lambda_{1}}{\lambda_{1}+\lambda_{2}}, \frac{\lambda_{2}}{\lambda_{1}+\lambda_{2}}]$. The equations in \eqref{general_eq} are:
\begin{multline} \label{2_1}
(\lambda_{1}+\lambda_{2}+\mu_{1}+\mu_{2})\mathbf{v_{(1,2)}}= \mathbf{b}_{1} \pi_{(1,2)}+
\lambda_{1}(v_{10},0,v_{12})+\lambda_{1}(v_{20},0,v_{22}) +
\mu_{1}(v_{11},v_{11},v_{11})+\mu_{2}(v_{12},v_{11},v_{12}),
\end{multline}
\begin{multline} \label{2_2}
(\lambda_{1}+\lambda_{2}+\mu_{1}+\mu_{2})\mathbf{v_{(2,1)}}= \mathbf{b}_{2} \pi_{(2,1)}+
\lambda_{2}(v_{10},0,v_{11})+\lambda_{2}(v_{20},0,v_{21}) +
\mu_{1}(v_{21},v_{22},v_{21})+\mu_{2}(v_{22},v_{22},v_{22}),
\end{multline}
where $\mathbf{v_{(1,2)}}=(v_{(1,2),0},v_{(1,2),1},v_{(1,2),2})$, $\mathbf{v_{(2,1)}}=(v_{(2,1),0},v_{(2,1),2},v_{(2,1),1})$, and $\mathbf{b}_1=\mathbf{b}_2=(1,1,1)$.
Therefore, we have six equations and six unknowns here.
By writing down the equations from equations \eqref{2_1} and \eqref{2_2} in matrix form, $T \mathbf{v}=\bm{\pi}$ will be as follows:
% \begin{align*}
%     (\pi_{(1,2)},\pi_{(1,2)},\pi_{(1,2)},\pi_{(2,1)},\pi_{(2,1)},\pi_{(2,1)})^\top =
% \end{align*}

\begin{footnotesize}
\begin{align}
    \bordermatrix{
~	&	(1,2),0	&	(1,2),1	&	(1,2),2	&	(2,1),0	&	(2,1),2	&	(2,1),1	\cr
(1,2),0	&	\lambda_2 + \sum_{i=1}^{2} \mu_i 	&	-\mu_1	&	- \mu_2 	&	- \lambda_1	&	0	&		0		\cr
(1,2),1	&	0	&	\lambda_1 +\lambda_2	&	0	&	0	&	0 	&	0 			\cr
(1,2),2	&	0 	&	- \mu_1  	&	\lambda_2+\mu_1	&	0	&	-\lambda_1	&	0		\cr
(2,1),0	&	-\lambda_2	&	0	&	0	&	\lambda_1 +\sum_{i=1}^{2} \mu_i	&	- \mu_2	&		-\mu_1	&	\cr
(2,1),2	&	0	&	0	&	0		&	0		&	\lambda_1 +\lambda_2	&	0			\cr
(2,1),1	&	0	&	- \lambda_2	&	0	& 0		&	-\mu_2	&		\lambda_1 + \mu_2		\cr
   }
        \bordermatrix{
&		\cr
&		v_{(1,2),0}		\cr
&		v_{(1,2),1}		\cr
&		v_{(1,2),2}		\cr
&		v_{(2,1),0}	    \cr
&		v_{(2,1),2}	    \cr
&		v_{(2,1),1}	    \cr
    }
= \bordermatrix{
& \cr
& \pi_{(1,2)}\cr
& \pi_{(1,2)}\cr
& \pi_{(1,2)}\cr
& \pi_{(2,1)}\cr
& \pi_{(2,1)}\cr
& \pi_{(2,1)}\cr
}.
    \label{exm_2}
\end{align}
\end{footnotesize}

We can see that matrix $T$ here matches the general form in Lemma \ref{gen_mat}. 
Now we show these equations have non-negative solutions and use Lemma \ref{lem:yates} to find the AoI.
First, we look at the rows/columns $((1,2),1)$ and $((2,1),2)$ notice that they form a diagonal matrix of size $2$ by $2$. Therefore we can solve and remove the variables $v_{(1,2),1}, v_{(2,1),2}$. They correspond to variables $v_{q, q_1}$ in Lemma \ref{break_gen}. They are also non-negative since the the $2$ diagonal entries and the entries of vector $\bm{\pi}$ are positive. Second, consider rows/columns $((1,2,),2)$ and $((2,1),1)$
corresponding to variables $v_{q,q_2}$ in Lemma \ref{break_gen}, again we obtain a $2$ by $2$ diagonal matrix.
%which in this example is equivalent to variables $v_{(1,2),2}$ and $v_{(2,1),1}$, we notice that after removing $v_{q,q_1}$ variables, 
% and the coefficient matrix of size $2$ by $2$ is in the following form:
% \begin{footnotesize}
% \begin{align}
%     \bordermatrix{
% ~	&	(1,2),2	&	(2,1),1		\cr
% (1,2),2	&	\lambda_2 +  \mu_1 	&	0		\cr
% (2,1),1	&	0	&	\lambda_1+ \mu_2		\cr
%     }.
% \end{align}
% \end{footnotesize}
%We notice that this matrix is a special case of the following matrix \eqref{ex}. Therefore, by solving the matrix $T_0$ in \ref{ex}, 
Hence we are able to find the variables $v_{(1,2),2}, v_{(2,1),1}$. %as explained in lemma \ref{break_gen}. 
After removing these $4$ variables we are left with matrix $T_0$ which is in the same form as Equation \eqref{eq:T_0}:

\begin{footnotesize}
\begin{align} \label{ex}
    \bordermatrix{
~	&	(1,2),0	&	(2,1),0		\cr
(1,2),0	&	\lambda_2 + \sum_{i=1}^{2} \mu_i 	&	-\lambda_1		\cr
(2,1),0	&	-\lambda_2	&	\lambda_1+\sum_{i=1}^{2} \mu_i		\cr
    }.
\end{align}
\end{footnotesize}

%Here here we called them $v_1 = v_{10}$ and $v_2 = v_{20}$ for the simplicity of notation.
We can solve the matrix $T_0$ only after one iteration of Algorithm \ref{alg:heter}. The corresponding variables are denoted as $\mathbf{v}^{(0)}=({v^{(0)}_{(1,2)}},{v^{(0)}_{(2,1)}})^\top=({v_{(1,2),0}},{v_{(2,1),0}})^\top$.  By definition in Section \ref{subsec:notation}, the permutation $(1,2)$ is odd (2-increasing), and $(2,1)$ is even. In the forward path of Algorithm \ref{alg:heter} we  do column operation in Line \ref{line:column}, meaning subtracting the odd column from the even one, and then the row operation in Line \ref{line:row}, meaning adding the even row to the odd row. After these $2$ operations the matrix $T_0$ becomes $T_1^{\prime\prime}$:
\begin{align} \label{e_2}
    \bordermatrix{
		\cr
	&	 \sum_{i=1}^{2} \mu_i 	&	0		\cr
	&	-\lambda_2	&	\sum_{i=1}^{2} \mu_i  +  \sum_{i=1}^{2} \lambda_i	\cr
    }.
\end{align}
After the column operation, the second variable $v_{(2,1),0}$ remains unchanged, and the first variable becomes ${v^{(1)}_{(1,2)}}= {v^{(0)}_{(1,2)}}+{v^{(0)}_{(2,1)}}={v_{(1,2),0}}+{v_{(2,1),0}}$. From \eqref{e_2} we can solve the first equation with the first variable $v^{(1)}_{(1,2)}$, whose coefficient matrix (Line \ref{line:pick_T_j}) is $T_1= \sum_{i=1}^{3} \mu_i$. The remaining coefficient matrix (Line \ref{line:pick_R_j}) for the second variable $v_{(2,1),0}$ is $R_1=\sum_{i=1}^{2} \mu_i  +  \sum_{i=1}^{2} \lambda_i.$  %which is of size $\frac{2!}{2!} \times \frac{2!}{2!}$. 

In the backward path of Algorithm \ref{alg:heter}, we find the first variable $v^{(1)}_{(1,2)}$,
%the sum of $2$ variables $v_{(1,2),0}$ and $v_{(2,1),0}$. In the backward path
 %$R_1$ is of size $1$ with value $\sum_{i=1}^{2} \mu_i  +  \sum_{i=1}^{2} \lambda_i$. 
and then the second variable $v_{(2,1),0}$. Now we can find the variable $v_{(1,2),0} = v^{(1)}_{(1,2)} - v_{(2,1),0}.$ One can see from \eqref{e_2} that $v_{(2,1),0}$ is non-negative, and we will show in Lemma \ref{lem:heter_nonneg} that $v_{(1,2),0}$ is also non-negative. So the solution to \eqref{exm_2} is non-negative, and by Lemma \ref{lem:yates} the AoI equals $v^{(1)}_{(1,2)}$.
\end{example} 

%Definition:
%$Incr_{j}(p)$ takes permutation $p=(p_1, p_2,p_3,...,p_n)$ and returns another permutation. First permutes the last 2 elements of $p$ and makes it $2$-increasing meaning the last $2$ elements are in increasing order. Then it permutes the last 3 elements of the new permutation and make it $3$-increasing (so that the last $3$ elements of it is increasing). At every step, one position is changed. We continue to do this procedure until the final permutation is $j$ increasing.
%We also define $Incr_{j}(Q)$, where $Q$ is a set of permutation, as applying $Incr_{j}$ function to each member of the set as define above. 
\begin{example}
Consider the case with $n=3$ servers.
By writing down the equations in \eqref{general_eq}, we have $T\mathbf{v} = \bm{\pi}$, where 
the constant vector $\bm{\pi}$ is
\begin{align*}
    (\pi_{(1,2,3)},\pi_{(1,2,3)},\pi_{(1,2,3)},\pi_{(1,2,3)},\pi_{(1,3,2)},\pi_{(1,3,2)},\pi_{(1,3,2)},\pi_{(1,3,2)},...,\pi_{(3,2,1)},\pi_{(3,2,1)},\pi_{(3,2,1)},\pi_{(3,2,1)})^\top,
\end{align*}
the variable vector $\mathbf{v}$ is
\begin{align*}
(
	&	v_{(1,2,3),0}		,
		v_{(1,2,3),1}		,
		v_{(1,2,3),2}		,
		v_{(1,2,3),3}		,
		v_{(1,3,2),0}		,
		v_{(1,3,2),1}		,
		v_{(1,3,2),3}		,
		v_{(1,3,2),2}		,\\
	&	v_{(2,1,3),0}		,
		v_{(2,1,3),2}		,
		v_{(2,1,3),1}		,
		v_{(2,1,3),3}		,
		v_{(2,3,1),0}		,
		v_{(2,3,1),2}		,
		v_{(2,3,1),3}		,
		v_{(2,3,1),1}		,\\
	&	v_{(3,1,2),0}		,
		v_{(3,1,2),3}		,
		v_{(3,1,2),1}		,
		v_{(3,1,2),2}		,
		v_{(3,2,1),0}		,
		v_{(3,2,1),3}		,
		v_{(3,2,1),2}		,
		v_{(3,2,1),1}		
)^\top.
\end{align*}
and the coefficient matrix $T$ is

\begin{tiny} \nonumber
\begin{align}
\setlength\arraycolsep{2pt}
    \left[ \begin{array}{cccc|cccc|cccc|cccc|cccc|cccc}
 * &	 -\mu_1	&	 -\mu_2 &	-\mu_3	&	0	&	0	&	0	&	0	&	  -\lambda_1	&	0	&	0	&	0	&	   -\lambda_1	&	0	&	0	&	0	&	0	&	0	&	0	&	0	&	0	&	0	&	0	&	0
\\
 0	&	 \$	&	0	& 0	&	0	&	0	&	0	&	0	&	0	&	0	&	0	&	0	&	0	&	0	&	0	&	0	&	0	&	0	&	0	&	0	&	0	&	0	&	0	&	0
\\ 
  0	&	 -\mu_1		&	 @	&	0	&	0	&	0 &	0	&	0	&	0	&	 -\lambda_1	&	0	&	0	&	0	&	-\lambda_1	&	 0	&	0	&	0	&	0	&	0	&	0	&	0	&	0	&	0	&	0 
\\
 0	&	 -\mu_1	&	 -\mu_2	&	 \#	&	0	&	0	&	0	&	0	&	0	&	0	&	0	&	-\lambda_1	&	0	&	0	&	-\lambda_1	&	 0	&	0	&	0	&	0	&	0	&	0	&	0	&	0	&	0	\\
 \hline
 0	&	0	&	0	&	0	&	 *	&	 -\mu_1		&	-\mu_3		&	 -\mu_2	&	0	&	0	&	0	&	0	&	0	&	0	&	0	&	0	&	   -\lambda_1	&	0	&	0	&	0	&	   -\lambda_1	&	0	&	0	&	0	\\ 
 0	&	0	&	0	&	0	&	0	&	 \$	&	0	&	0 &	0	&	0	&	0	&	0	&	0	&	0	&	0	&	0	&	0	&	0	&	0	&	0	&	0	&	0	&	0	&	0	\\
  0	&	0	&	0	&	0	&	0	&	 -\mu_1	&	@	&	 0	&	0	&	0	&	0	&	0	&	0	&	0	&	0	&	0	&	0	&	-\lambda_1	&	0	&	  0	&	0	&	-\lambda_1	&	0	&	  0	\\
 0	&	0	&	0	&	0	&	0	&	 -\mu_1	&	-\mu_3	&	  \#	&	0	&	0	&	0	&	0	&	0	&	0	&	0	&	0	&	0	&	0	&	0	&	-\lambda_1	&	0	&	0	&	 -\lambda_1	&	0	\\
 \hline
  -\lambda_2	&	0	&	0	&	0	&	   -\lambda_2	&	0	&	0	&	0	&	 *	&	 -\mu_2	&	 -\mu_1	&	-\mu_3		&	0	&	0	&	0	&	0	&	0	&	0	&	0	&	0	&	0	&	0	&	0	&	0	\\
   0	&	0	&	0	&	0	&	0	&	0	&	0	&	0	&	0	&	\$	&	 0	&	0	&	0	&	0	&	0	&	0	&	0	&	0	&	0	&	0	&	0	&	0	&	0	&	0	\\
 0	&	 -\lambda_2	&	0	&	0	&	0	&	 -\lambda_2	&	0	&	0	&	0	&	 -\mu_2	&	 @	&	0	&	0	&	0	&	0	&	0	&	0	&	0	&	0	&	0	&	0	&	0	&	0	&	0	\\
 0	&	0	&	0	&	-\lambda_2	&	0	&	0	&	-\lambda_2	&	 0	&	0	&	 -\mu_2	&	 -\mu_1	&	 \#	&	0	&	0	&	0	&	0	&	0	&	0	&	0	&	0	&	0	&	0	&	0	&	0	\\
 \hline
 0	&	0	&	0	&	0	&	0	&	0	&	0	&	0	&	0	&	0	&	0	&	0	&	 *	&	-\mu_2	&	 -\mu_3	&	 -\mu_1	&	  -\lambda_2	&	0	&	0	&	0	&	  -\lambda_2	&	0	&	0	&	0	\\
  0	&	0	&	0	&	0	&	0	&	0	&	0	&	0	&	0	&	0	&	0	&	0	&	0	&	\$	&	 0	&	0	&	0	&	0	&	0	&	0	&	0	&	0	&	0	&	0	\\
   0	&	0	&	0	&	0	&	0	&	0	&	0	&	0	&	0	&	0	&	0	&	0	&	0	&	-\mu_2	&	 @	&	 0	&	0	&	-\lambda_2		&	0	&	 0 &	0	&	-\lambda_2	&	0	&	 0	\\
 0	&	0	&	0	&	0	&	0	&	0	&	0	&	0	&	0	&	0	&	0	&	0	&	0	&	 -\mu_2	&	 -\mu_3	&	 \#	&	0	&	0	&	 -\lambda_2	&	0	&	0	& 0	&	0	&		-\lambda_2	\\
 \hline
  -\lambda_3	&	0	&	0	&	0	&	  -\lambda_3	&	0	&	0	&	0	&	0	&	0	&	0	&	0	&	0	&	0	&	0	&	0	&	 *	&	 -\mu_3	&	-\mu_1	&	 -\mu_2	&	0	&	0	&	0	&	0	\\
   0	&	0	&	0	&	0	&	0	&	0	&	0	&	0	&	0	&	0	&	0	&	0	&	0	&	0	&	0	&	0	&	0	&	\$	&	0	&	 0	&	0	&	0	&	0	&	0	\\
 0	&	 -\lambda_3	&	0	&	0	&	0	&	 -\lambda_3	&	0	&	0	&	0	&	0	&	0	&	0	&	0	&	0	&	0	&	0	&	0	&	-\mu_3 	&	@	&	 0	&	0	&	0	&	0	&	0	\\
 0	&	0	&	 -\lambda_3	&	0	&	0	&	0	&	0	&	-\lambda_3	&	0	&	0	&	0	&	0	&	0	&	0	&	0	&	0	&	0	&	 -\mu_3	&	-\mu_1 	&	\# 	&	0	&	0	&	0	&	0	\\
 \hline
 0	&	0	&	0	&	0	&	0	&	0	&	0	&	0	&	  -\lambda_3	&	0	&	0	&	0	&	  -\lambda_3	&	0	&	0	&	0	&	0	&	0	&	0	&	0	&	 *	&	-\mu_3	&	 -\mu_2	&	-\mu_1 	\\
  0	&	0	&	0	&	0	&	0	&	0	&	0	&	0	&	0	&	0	&	0	&	0	&	0	&	0	&	0	&	0	&	0	&	0	&	0	&	0	&	0	&	\$	&	0	&	 0	\\
 0	&	0	&	0	&	0	&	0	&	0	&	0	&	0	&	0	&	-\lambda_3	&	 0	&	0	&	0	&	-\lambda_3	&	 0	&	0	&	0	&	0	&	0	&	0	&	0	&	-\mu_3	&	 @	&	 0	\\
  0	&	0	&	0	&	0	&	0 &	0	&	0	&	0	&	0	&	 0	&	-\lambda_3	&	0 &	0	&	0	&	0	&	-\lambda_3	&	0	&	0	&	0	&	0	&	0	&-\mu_3	&	 -\mu_2	&	 \# 	
\end{array}\right].
\end{align}
\end{tiny}
Here * in row $(q,i)$ is $\sum_{j=2}^{3} \lambda_{q_j} + \sum_{j=1}^{3} \mu_{q_j} $, \$ is $\sum_{j=1}^{3} \lambda_{q_j}$, @ is $\sum_{j=2}^{3} \lambda_{q_j} +  \mu_{q_1} $, and \# is $\sum_{j=2}^{3} \lambda_{q_j} + \sum_{j=1}^{2} \mu_{q_j} $. We can see that matches with our result in Lemma \ref{gen_mat} as expected.
Non-zero elements of the first $4$ rows indexed by the permutation $((1,2,3))$ are in columns indexed by permutations $(1,2,3)$, $(2,1,3)$, and $(2,3,1)$, which are the incoming states of $(1,2,3)$. Non-zero elements of the first $4$ columns are in rows indexed by $(1,2,3)$, $(2,1,3)$, and $(3,1,2)$, which are the outgoing states of state $(1,2,3)$. 

We illustrate here how we use Lemma \ref{break_gen} in order to solve matrix $T$. Variables $v_{q,q_1}$ correspond a diagonal submatrix of $T$ with size $n! \times n!$ (the \$ entries), and we can find their values. %They correspond to rows/columns $2,6,11,15,20,24$ in this example.
After removing these variables, for finding $v_{q,q_2}$, we solve the ones that the last $2$ entries of their permutation are the same. For instance if $(q_2,q_3)= (2,3)$ we see that we only need to solve the single variable $v_{(1,2,3),2}$, corresponding to the $3$rd row/column.
Therefore, we can solve all these variables individually. 
For solving $v_{q,q_3}$, we solve the ones that their last  entry of their permutation is the same. For instance if $q_3= 3$, %after removing entries corresponding to variables $v_{q,q_1}$ and $ v_{q,q_2}$ from the matrix, 
we need to solve variables $v_{(1,2,3),3}$ and  $v_{(2,1,3),3}$ together. The resulting coefficients for these variables are as follows:
\begin{align} \label{e_2_again}
    \bordermatrix{
		\cr
	&	 \lambda_2+ \lambda_3 + \mu_1 + \mu_2 	&	-\lambda_1		\cr
	&	-\lambda_2	&	\lambda_1+\lambda_3 + \mu_1 +\mu_2	\cr
    },
\end{align}
and we can see that we solved this in Equation \eqref{ex} of Example \ref{ex_2}  with a change of variable. At the end, we need to solve variables $v_{q,q_0}$ or in another word $T_0$ which is as follows:

\begin{scriptsize}
\begin{align}
    \bordermatrix{
~	&	(1,2,3)	&	(1,3,2)	&	(2,1,3)	&	(2,3,1)	&	(3,1,2)	&	(3,2,1)		\cr
(1,2,3)	&	\lambda_2 +\lambda_3+ \sum_{i=1}^{3} \mu_i 	&	0	&	- \lambda_1 	&	- \lambda_1	&	0	&		0		\cr
(1,3,2)	&	0	&	\lambda_2 +\lambda_3+\sum_{i=1}^{3} \mu_i	&	0	&	0	&	- \lambda_1 	&	- \lambda_1 			\cr
(2,1,3)	&	- \lambda_2 	&	- \lambda_2 	&	\lambda_1 +\lambda_3+\sum_{i=1}^{3} \mu_i	&	0	&	0	&	0		\cr
(2,3,1)	&	0	&	0	&	0	&	\lambda_1 +\lambda_3+\sum_{i=1}^{3} \mu_i	&	- \lambda_2	&			- \lambda_2	\cr
(3,1,2)	&	- \lambda_3	&	- \lambda_3	&	0	& 0		&	\lambda_1 +\lambda_2+\sum_{i=1}^{3} \mu_i	&	0			\cr
(3,2,1)	&	0	&	0	&	- \lambda_3		&	- \lambda_3		&	0	&	\lambda_1 +\lambda_2+\sum_{i=1}^{3} \mu_i			\cr
    }.
    \label{exm_4}
\end{align}
\end{scriptsize}
In the first run of the forward path in Algorithm \ref{alg:heter} we perform row and column operations on $T_0$, and obtain $T_{1}^{\prime\prime}$ as

\begin{footnotesize}
\begin{align}
    \bordermatrix{
    	\cr
&	\lambda_2 +\lambda_3+\sum_{i=1}^{3} \mu_i	&	0	&	- \lambda_1	&	0	&	- \lambda_1 	&	0
	\cr
	&	0 	&\lambda_2 +\lambda_3+ \sum_{i=1}^{3} \mu_i&	0&	0	&		- \lambda_1	& 0			\cr
	&	- \lambda_2 	&	0	&	\lambda_1 +\lambda_3+\sum_{i=1}^{3} \mu_i	&	0	&	- \lambda_2	&	0		\cr
&	0	&	0	&	0	&	\lambda_1 +\lambda_3+\sum_{i=1}^{3} \mu_i	&	- \lambda_2 &		 0	\cr
	&	- \lambda_3	&	0&	- \lambda_3	& 0		&	\lambda_1 +\lambda_2+\sum_{i=1}^{3} \mu_i	&	0			\cr
	&	0	&	0&	- \lambda_3		&	0	&	0	&	\lambda_1 +\lambda_2+\sum_{i=1}^{3} \mu_i
   	  }.
\end{align}
\end{footnotesize}
Therefore the submatrices $T_{1}$ and $R_{1}$ are as follows, respectively.
\begin{footnotesize}
\begin{align}
    \bordermatrix{
    	\cr
&	\lambda_2 +\lambda_3+\sum_{i=1}^{3} \mu_i	&	- \lambda_1	&	- \lambda_1 
		\cr
	&	- \lambda_2 	&	\lambda_1 +\lambda_3+\sum_{i=1}^{3} \mu_i	&	- \lambda_2 	&		\cr
	&	- \lambda_3	&	- \lambda_3	&	\lambda_1 +\lambda_2+\sum_{i=1}^{3} \mu_i	&			\cr
    },
\end{align}
\end{footnotesize}
\begin{footnotesize}
\begin{align}
    \bordermatrix{
    	\cr
&	\lambda_2 +\lambda_3+\sum_{i=1}^{3} \mu_i	&	0	&	0
		\cr
	&	0 	&	\lambda_1 +\lambda_3+\sum_{i=1}^{3} \mu_i	&	0 	&		\cr
	&0	&	0	&	\lambda_1 +\lambda_2+\sum_{i=1}^{3} \mu_i	&			\cr
    }.
\end{align}
\end{footnotesize}
Since we performed column operations on each iteration of Algorithm \ref{alg:heter}, the variables change accordingly. After the first run of the algorithm the new variables corresponding to $T_1$ are as follows:
\begin{align}
 v_{(1,2,3)}^{(1)} &= v_{(1,2,3),0} + v_{(1,3,2),0}\quad ,\nonumber\\  
 v_{(2,1,3)}^{(1)} &= v_{(2,1,3),0} + v_{(2,3,1),0}\quad, \nonumber\\ 
 v_{(3,1,2)}^{(1)} &= v_{(3,1,2),0} + v_{(3,2,1),0}\quad. \label{eq:v_1_new}
\end{align}
The remaining variables corresponding to $R_1$ are unchanged:
\begin{align}
  v_{(1,3,2)}^{(1)} &= v_{(1,3,2),0} \quad, \nonumber\\ 
  v_{(2,3,1)}^{(1)} &= v_{(2,3,1),0} \quad, \nonumber\\ 
  v_{(3,2,1)}^{(1)} &= v_{(3,2,1),0}\quad. \label{eq:v_1_old}
\end{align}
%and can be solved because $R_1$ is diagonal.
After the second run of the forward path, we perform row and column operations on $T_1$ and obtain $T_2^{\prime\prime}$:
\begin{footnotesize}
\begin{align} \label{e_3}
    \bordermatrix{
    	\cr
&	\sum_{i=1}^{3} \mu_i	&	0	&	0
		\cr
	&	-\lambda_2 	&	\sum_{i=1}^{3} \mu_i	&	0 	&		\cr
	&-\lambda_3	&	0	&	\sum_{i=1}^{3} \mu_i	&			\cr
    }.
\end{align}
\end{footnotesize}
Hence, $T_2=\sum_{i=1}^{3} \mu_i$ and $R_2$ is the $2 \times 2$ diagonal matrix with diagonal entries $\sum_{i=1}^{3} \mu_i$.
Correspondingly, the new variable corresponding to $T_2$ is
\begin{align}
    v_{(1,2,3)}^{(2)} &= v_{(1,2,3),0} + v_{(1,3,2),0}+v_{(2,1,3),0} + v_{(2,3,1),0}+v_{(3,1,2),0} + v_{(3,2,1),0}\quad, \label{eq:v_2_new}
\end{align}
and the other two variables corresponding to $R_2$ are not changed:
\begin{align}
    v_{(2,1,3)}^{(2)} &= v_{(2,1,3)}^{(1)} = v_{(2,1,3),0} + v_{(2,3,1),0}\quad, \nonumber\\ 
    v_{(3,1,2)}^{(2)} &= v_{(3,1,2)}^{(1)} = v_{(3,1,2),0} + v_{(3,2,1),0}\quad. \label{eq:v_2_old}
 % v_{(1,3,2)}^{(2)} &= v_{(1,3,2),0}\quad , \\ 
 % v_{(2,3,1)}^{(2)} &= v_{(2,3,1),0} \quad, \\ 
 % v_{(3,2,1)}^{(2)} &= v_{(3,2,1),0}\quad. 
\end{align}

In the backward path, we solve the variables.
From $T_2$, we can find the variable $v_{(1,2,3)}^{(2)}$, and then after removing it from the matrix $T_2^{\prime \prime}$, we can find $v_{(2,1,3)}^{(2)}$ and $v_{(3,1,2)}^{(2)}$. Hence, we can solve the variables in \eqref{eq:v_1_new}. After removing these variables, the variables in \eqref{eq:v_1_old} can be solved since $R_1$ is diagonal. Finally, $v_{q,0}$ can be solved for any $q$ using \eqref{eq:v_1_new} and \eqref{eq:v_1_old}. The non-negativity of the solution is shown in Lemma \ref{lem:heter_nonneg}, and by Lemma \ref{lem:yates} the AoI can be computed as $v_{(1,2,3)}^{(2)}$.
\end{example}

In the following, we derive average AoI explicitly in the case of $n=2$ in Example \ref{ex_2}. %and find the optimal arrival rates given that the sum of arrival rates are fixed.

\begin{theorem}
Consider one source and $n=2$ heterogeneous servers. The AoI is given by:
\begin{small}
\begin{align}
%\label{single_soure_2_hetereg_servers}
\Delta=  
\frac{1}{\mu_{1}+\mu_{2}} + \frac{1}{\lambda_{1}+\lambda_{2}}+\frac{1}{\mu_{1}+\mu_{2}} \frac{1}{\lambda_{1}+\lambda_{2}} (\frac{\mu_{1}\lambda_{2}}{\lambda_{1}+\mu_{2}} + \frac{\mu_{2}\lambda_{1}}{\lambda_{2}+\mu_{1}}).\nonumber
\end{align}
\end{small}
\end{theorem}
\begin{proof}
Following the solution in Example~\ref{ex_2}, we find the $6$ variables corresponding to $v_{1,2}=(v_{(1,2),0},v_{(1,2),1},v_{(1,2),2})$ and $v_{2,1}=(v_{(2,1),0},v_{(2,1),2},v_{(2,1),1})$ as:
 $v_{(1,2),1}=\frac{\pi_{1}}{\lambda_{1}+\lambda_{2}}$ and $v_{(2,1),2}=\frac{\pi_{2}}{\lambda_{1}+\lambda_{2}}$.
\begin{align*}
v_{(1,2),2}=  \pi_{1} (\frac{1}{\lambda_{1}+\lambda_{2}} + \frac{1}{\lambda_{2}+\mu_{1}}), \quad 
v_{(2,1),1}=  \pi_{2} (\frac{1}{\lambda_{1}+\lambda_{2}} + \frac{1}{\lambda_{1}+\mu_{2}}).
\end{align*}
\begin{comment}
\begin{equation*}
(\lambda_{1}+\lambda_{2}+\mu_{1}+\mu_{2})v_{(1,2),0}=\pi_{1}+\lambda_{1}v_{(1,2),0}+\lambda_{1}v_{(2,1),0}+\mu_{1}v_{(1,2),1}+\mu_{2}v_{(1,2),2},
\end{equation*}
\begin{equation*}
(\lambda_{1}+\lambda_{2}+\mu_{1}+\mu_{2})v_{(2,1),0}=\pi_{2}+\lambda_{2}v_{(1,2),0}+\lambda_{2}v_{(2,1),0}+\mu_{1}v_{(2,1),1}+\mu_{2}v_{(2,1),2}.
\end{equation*}
\end{comment}
Also from Lemma \ref{ss} we know that, $[{\pi}_{(1,2)},\pi_{(2,1)}]=[\frac{\lambda_{1}}{\lambda_{1}+\lambda_{2}}, \frac{\lambda_{2}}{\lambda_{1}+\lambda_{2}}]$.
Following the steps in Example \ref{ss}, the average AoI is  $v_{(1,2),0}+v_{(2,1),0}$ which simplifies to:
\begin{align*}
 AoI &= \frac{1}{\mu_{1}+\mu_{2}}+ \frac{\mu_{1}(v_{(1,2),1}+v_{(2,1),1})+\mu_{2}(v_{(1,2),2}+v_{(2,1),2})}{\mu_{1}+\mu_{2}}\\ & = \frac{1}{\mu_{1}+\mu_{2}} + \frac{1}{\lambda_{1}+\lambda_{2}}+\frac{1}{\mu_{1}+\mu_{2}} \frac{1}{\lambda_{1}+\lambda_{2}} (\frac{\mu_{1}\lambda_{2}}{\lambda_{1}+\mu_{2}} + \frac{\mu_{2}\lambda_{1}}{\lambda_{2}+\mu_{1}}).
\end{align*}

\end{proof}
Next, for $n=2$ servers, we find the optimal arrival rates of servers, ${\lambda_1}^{*},{\lambda_2}^{*}$,  given fixed service rates $\mu_1,\mu_2$ and sum arrival rate $\lambda \triangleq \lambda_1+\lambda_2$. The optimal ${\lambda_1}^{*}$ is illustrated in Figure~\ref{optimal_lambda}.
\begin{figure}
\centering
\includegraphics[width = 0.4\textwidth]{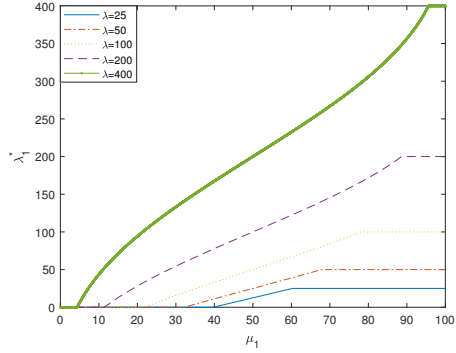} 
\caption{Optimal value of $\lambda_1$ as a function of $\mu_1$. $\lambda_1+\lambda_2=\lambda, \mu_1+\mu_2=100$.}
\label{optimal_lambda}
\end{figure}

\begin{theorem}
\label{thm:optimal_lambda_heto_n_2}
For one source and $n=2$ heterogeneous servers,  given $\mu_1,\mu_2$ and fixed $\lambda_1+\lambda_2=\lambda$, the optimal ${\lambda_{1}}^{*}$ satisfies

\noindent $\bullet$ if  $\mu_{1} < \mu_{2}$ and $\mu_{2}^2 - \frac{\mu_{1}(\lambda+\mu_{1})(\lambda+\mu_{2})}{\mu_{2}} < 0$,
\begin{align*}
 {\lambda_{1}}^{*}=\frac{-(\mu_{2}+c(\lambda+\mu_{1}))+\sqrt{\mu_{1}(\lambda+\mu_{2})(2+\frac{\mu_{2}}{\lambda+\mu_{1}}+\frac{\lambda+\mu_{1}}{\mu_{2}})}}{1- \frac{\mu_{1}(\lambda+\mu_{2})}{\mu_{2}(\lambda+\mu_{1})}},  
\end{align*}
 \noindent $\bullet$  if  $\mu_{1} < \mu_{2}$ and $\mu_{2}^2 - \frac{\mu_{1}(\lambda+\mu_{1})(\lambda+\mu_{2})}{\mu_{2}} \geq 0: $ %$  {\lambda_1}^{*}=0, {\lambda_2}^{*}=\lambda,$
     \begin{align*}
        \lambda_1^{*}=0, \lambda_2^{*}=\lambda,
   \end{align*}
\noindent $\bullet$ if $\mu_{1}>\mu_{2}$ and $\mu_{1}^2 \geq \frac{\mu_{2}(\lambda+\mu_{1})(\lambda+\mu_{2})}{\mu_{1}}:$ %${\lambda_{1}}^{*}=\lambda, {\lambda_{2}}^{*}=0,$
   \begin{align*}
         \lambda_{1}^{*}=\lambda, \lambda_{2}^{*}=0,
   \end{align*}
 \noindent $\bullet$    if $\mu_{1}>\mu_{2}$ and $\mu_{1}^2 < \frac{\mu_{2}(\lambda+\mu_{1})(\lambda+\mu_{2})}{\mu_{1}}$.
     \begin{align*}
{\lambda_{1}}^{*}= \lambda-\frac{-(\mu_{1}+\frac{(\lambda+\mu_{2})}{c})+\sqrt{\mu_{2}(\lambda+\mu_{1})(2+\frac{\mu_{1}}{\lambda+\mu_{2}}+\frac{\lambda+\mu_{2}}{\mu_{1}})}}{1- \frac{\mu_{2}(\lambda+\mu_{1})}{\mu_{1}(\lambda+\mu_{2})}},
\end{align*}  
%\begin{small}
%\begin{align*}
 %   \lambda_1^{*}=\frac{-(\mu_{2}+c(\lambda+\mu_{1}))+\sqrt{\mu_{1}(\lambda+\mu_{2})(2+\frac{\mu_{2}}{\lambda+\mu_{1}}+\frac{\lambda+\mu_{1}}{\mu_{2}})}}{1- \frac{\mu_{1}(\lambda+\mu_{2})}{\mu_{2}(\lambda+\mu_{1})}}, \\
 %   \end{align*}
  %    \vspace{-0.2cm}
    %if $\mu_{1} < \mu_{2}$ and $\mu_{2}^2 - \frac{\mu_{1}(\lambda+\mu_{1})(\lambda+\mu_{2})}{\mu_{2}} \geq 0$.
%    \begin{align*}
  %      \lambda_1^{*}=0, \lambda_2^{*}=\lambda,
 %   \end{align*}
%    \vspace{-0.2cm}
%   if  $\mu_{1} < \mu_{2}$ and $\mu_{2}^2 - \frac{\mu_{1}(\lambda+\mu_{1})(\lambda+\mu_{2})}{\mu_{2}} \geq 0.$
 %  \begin{align*}
 %        \lambda_{1}^{*}=\lambda, \lambda_{2}^{*}=0,
%   \end{align*}
   % \vspace{-0.2cm}
%   if $\mu_{1}>\mu_{2}$ and $\mu_{1}^2 \geq \frac{\mu_{2}(\lambda+\mu_{1})(\lambda+\mu_{2})}{\mu_{1}}$.
 %  \begin{align*}
%\lambda_{1}^{*}= \lambda-\frac{-(\mu_{1}+\frac{(\lambda+\mu_{2})}{c})+\sqrt{\mu_{2}(\lambda+\mu_{1})(2+\frac{\mu_{1}}{\lambda+\mu_{2}}+\frac{\lambda+\mu_{2}}{\mu_{1}})}}{1- \frac{\mu_{2}(\lambda+\mu_{1})}{\mu_{1}(\lambda+\mu_{2})}},
%\end{align*}
  %    \vspace{-0.2cm}
   % if $\mu_{1}>\mu_{2}$ and $\mu_{1}^2 < \frac{\mu_{2}(\lambda+\mu_{1})(\lambda+\mu_{2})}{\mu_{1}}$.
%\end{small}
\end{theorem}

\noindent where $c= \frac{\mu_{1}(\lambda+\mu_{2})}{\mu_{2}(\lambda+\mu_{1})}$.
\begin{proof}
In order to find the optimal values of $\lambda_{1}$ and $\lambda_{2}$ for a given values of $\mu_1,\mu_{2}, \lambda$ where $\lambda_{1}+\lambda_{2}=\lambda$, we set the derivative of the following equation with respect to $\lambda_{1}$, $\lambda_{2}$ and $a$ to zero.
\begin{align*}
 AoI &=\frac{1}{\mu_{1}+\mu_{2}}+ \frac{\mu_{1}(v_{(1,2),1}+v_{(2,1),1})+\mu_{2}(v_{(1,2),2}+v_{(2,1),2})}{\mu_{1}+\mu_{2}} - a(\lambda_{1}+\lambda_{2}-\lambda),\\
 \frac{\partial AoI}{\partial \lambda_{1}} &= \frac{-1}{(\lambda_{1}+\lambda_{2})^2} 
-\frac{\mu_{1}\lambda_{2}(2\lambda_{1}+\lambda_{2}+\mu_{2})}{(\lambda_{1}+\lambda_{2})^2 (\lambda_{1}+\mu_{2})^2}+\frac{(\lambda_{2}+\mu_{1})(\mu_{2}\lambda_{2})}{(\lambda_{1}+\lambda_{2})^2 (\lambda_{2}+\mu_{1})^2}-a =0,\\
 \frac{\partial AoI}{\partial \lambda_{2}} & = \frac{-1}{(\lambda_{1}+\lambda_{2})^2} 
-\frac{\mu_{2}\lambda_{1}(2\lambda_{2}+\lambda_{1}+\mu_{1})}{(\lambda_{1}+\lambda_{2})^2 (\lambda_{2}+\mu_{1})^2}+\frac{(\lambda_{1}+\mu_{2})(\mu_{1}\lambda_{1})}{(\lambda_{1}+\lambda_{2})^2 (\lambda_{1}+\mu_{2})^2}-a =0.
\end{align*} 
Also, we know that $\lambda_{1}+\lambda_{2}=\lambda$. With some algebraic simplification we reach to this $2nd$ order polynomial equation for finding the optimal value of $\lambda_{1}$ and consequently $\lambda_{2}$.
\begin{align} \label{opti}
\lambda_{1}^2 (1-c)+ 2\lambda_{1} (\mu_{2}+c(\lambda+\mu_{1})) + \mu_{2}^2 - c(\lambda+\mu_{1})^2,
\end{align}

\noindent where $c= \frac{\mu_{1}(\lambda+\mu_{2})}{\mu_{2}(\lambda+\mu_{1})}$.

When $c=1$ it is equivalent to $\mu_{1}=\mu_{2}$ and the equation \eqref{opti} becomes a first order polynomial which results in $\lambda_{1}=\lambda_{2}= \frac{\lambda}{2}$.
This polynomial has $2$ real roots because of its positive discriminant and therefore solving the equation \eqref{opti} gives us $2$ possible candidate for our optimization problem. When $\mu_{1} < \mu_{2}$ then $c<1$. Knowing the fact that for $2$ roots of \eqref{opti} we have,
\begin{align*}
& r_{1}+r_{2} = \frac{\mu_2 + \frac{\mu_{1}(\lambda+\mu_{2})}{\mu_{2}}}{c-1}, \\
& r_{1}r_{2} = \frac{\mu_{2}^2 - \frac{\mu_{1}(\lambda+\mu_{1})(\lambda+\mu_{2})}{\mu_{2}}}{1-c},
\end{align*}

we conclude, when $\mu_{1} < \mu_{2}$ and $\mu_{2}^2 - \frac{\mu_{1}(\lambda+\mu_{1})(\lambda+\mu_{2})}{\mu_{2}} \geq 0$, the $2$ roots are negative and therefore in this regime our optimal values become $\lambda_{1}=0, \lambda_{2}=\lambda$. When $\mu_{1} < \mu_{2}$ and $\mu_{2}^2 - \frac{\mu_{1}(\lambda+\mu_{1})(\lambda+\mu_{2})}{\mu_{2}} \geq 0$, the positive root is the optimal rate which is equal to:
\begin{align*}
\lambda_{1}=\frac{-(\mu_{2}+c(\lambda+\mu_{1}))+\sqrt{\mu_{1}(\lambda+\mu_{2})(2+\frac{\mu_{2}}{\lambda+\mu_{1}}+\frac{\lambda+\mu_{1}}{\mu_{2}})}}{1- \frac{\mu_{1}(\lambda+\mu_{2})}{\mu_{2}(\lambda+\mu_{1})}}.
\end{align*}

Similarly by writing the $2$-nd order polynomial for $\lambda_{2}$, we reach to the conclusion that when $\mu_{1}>\mu_{2}$ , if $\mu_{1}^2 \geq \frac{\mu_{2}(\lambda+\mu_{1})(\lambda+\mu_{2})}{\mu_{1}}$ the optimal rates are $\lambda_{1}=\lambda, \lambda_{2}=0$.
In the regime that $\mu_{1}>\mu_{2}$ and $\mu_{1}^2 < \frac{\mu_{2}(\lambda+\mu_{1})(\lambda+\mu_{2})}{\mu_{1}}$, the positive root is the optimal rate.
\begin{align*}
\lambda_{2}= \frac{-(\mu_{1}+\frac{(\lambda+\mu_{2})}{c})+\sqrt{\mu_{2}(\lambda+\mu_{1})(2+\frac{\mu_{1}}{\lambda+\mu_{2}}+\frac{\lambda+\mu_{2}}{\mu_{1}})}}{1- \frac{\mu_{2}(\lambda+\mu_{1})}{\mu_{1}(\lambda+\mu_{2})}}.
\end{align*}
\end{proof}
\vspace{0.1cm}
When $\mu_1=\mu_2$ the optimal rates that minimize AoI are ${\lambda_{1}}^{*}={\lambda_{2}}^{*}=\frac{\lambda}{2}$. As Figure \ref{optimal_lambda} illustrates, for $\mu_1=\mu_2=50$, optimal rates are ${\lambda_{1}}^{*}=\frac{\lambda}{2}$ and in the regimes that one of the service rates is much greater than the other one, AoI minimizes when all the updates are sent to the server with greater service rate.  
\begin{figure}
\centering
\includegraphics[width = 0.55\textwidth]{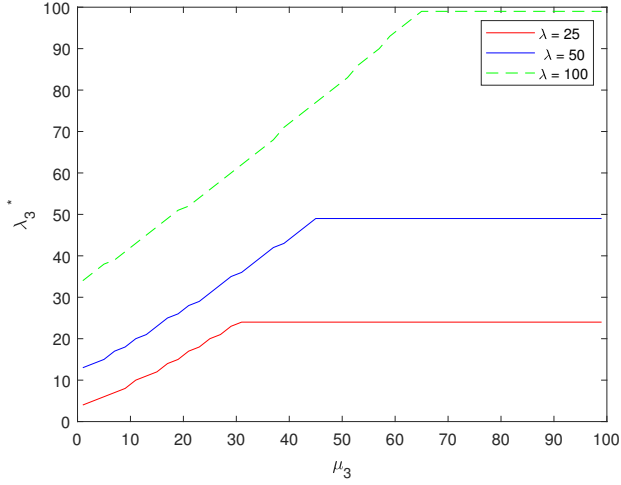} 
\caption{Optimal arrival rate (${\lambda_3}^* $) for server $3$ based on service rate ($\mu_3$) of server $3$, where $\lambda_1=\lambda_2$, $\mu_1=\mu_2$ ,  $\sum_{i=1}^{3} \mu_i=100$, $\sum_{i=1}^{3} \lambda_i = \lambda$, and $n=3$.}
\label{final_simul}
\end{figure}

Also, when we have $n=3$ servers and sum of service rates is $\mu_1+\mu_2+\mu_3=100$, we notice a saturation region in Figure \ref{final_simul} similar to what we observed in Figure \ref{optimal_lambda}. Here Server $1$ and Server $2$ have the same arrival rate ($\lambda_1=\lambda_2$) and service rate  ($\mu_1=\mu_2$). It shows that when the service rate for one of the servers is much greater compared to the other $2$ servers, it is optimal in terms of minimizing average AoI to allocate most of the arrival rate to that server which is intuitive.

\section{Conclusion}
\label{conc}
In this paper, we studied the age of information in the presence of multiple independent servers monitoring several information sources. We derived AoI for the LCFS queue model using SHS analysis when we had a homogeneous network and a single source. We also provided an algorithm for deriving AoI when we have $m$ sources and $n$ servers in a homogeneous network. For a heterogeneous network, 
we proved that using our algorithm we can solve the enormous numbers of equations and find the average AoI. We illustrated how these algorithms implement for the cases of $n=2,3$ servers and derive the optimal arrival rate given that the sum arrival rate is fixed for when $n=2$. From the simulation, it is observed that LCFS outperforms LCFS with preemption in waiting and FCFS for a homogeneous single information source network. Future directions include deriving explicit formula of AoI in heterogeneous sensing networks where the update arrival rate and/or the service rate are different among the servers for any number of sources and servers. Also, investigating arrival rates and service rate other than Poison distribution can further enrich the system model practicality.
\appendix \label{apen}
{\bf Notations.}
Below are notations and observations used throughout the proof of the AoI for heterogeneous networks.
\begin{itemize}
   
\item
Consider linear equations in matrix form, $A\mathbf{v}=\mathbf{c}$. For most cases, the variable vector $\mathbf{v}$ and the constant vector $\mathbf{c}$ will have non-negative entries, denoted as 
$$\mathbf{v} \succeq 0.$$
$$\mathbf{c} \succeq 0.$$
\item 
A row of $A$ corresponds to an equation, hence we use the terms row and equation interchangeably. A column of $A$ corresponds to a variable, hence we use the terms column and variable interchangeably. 
\item
Denote the $(q,p)$-th entry of $A$ by $A(q,p).$ The $p$-th variable is denoted by $v_p$. The $p$-th entry of $\mathbf{c}$ is $c_p$.
\item
For row/equation $q$, column/variable $p$, if $A(q,p)<0, \mathbf{c} \succeq 0$  and we have solved $v_p \ge 0$, then we can change the $p$-th entry of $\mathbf{c}$ to $c_p-A(q,p)v_p$. The new constant vector is still non-negative. Thus we have one less variable in the equation. We say we \emph{exclude $v_p$ from the equation}.
\item If we do column operations on the coefficient matrix $A$ (in order to simplify the equations), then it is equivalent to changing the variables. For example, suppose $A$ has 4 columns, $A = [A(:,1),A(:,2),A(:,3),A(:,4)]$, and we subtract Column 1 from Columns 2 and 3. Then the variables become $v_1+v_2+v_3,v_2,v_3,v_4$, because
\begin{align}\label{eq:column_op_eg}
    [A(:,1),A(:,2),A(:,3),A(:,4)]
    \left[ \begin{matrix}
    v_1 \\
    v_2 \\
    v_3 \\
    v_4 
    \end{matrix} \right]
    = 
    [A(:,1),A(:,2)-A(:,1),A(:,3)-A(:,1),A(:,4)] 
    \left[ \begin{matrix}
    v_1+v_2+v_3 \\
    v_2 \\
    v_3 \\
    v_4 
    \end{matrix} \right].
\end{align}
By the above observation, in the forward path of Algorithm \ref{alg:heter}, after the $j=1$st iteration's column operation,
the $p$-th variable, denoted as $v_p^{(j)}$, becomes 
\begin{align}\label{eq:1_iter_variable}
v_p^{(1)} =
\sum\limits_{k=n-1}^{n}v_{g_{1,k}(p)} = \sum\limits_{q: (q_1,\dots,q_{n-2})=(p_1,\dots,p_{n-2})} v_{q}^{(0)},  
 \text{ $p$ is $2$-increasing}.
\end{align}
In the $j=2$nd iteration, only $2$-increasing variables are considered. The $p$-th variable becomes
\begin{align}
    v_p^{(2)} = 
    \sum\limits_{k=n-2}^{n} v^{(1)}_{g_{2,k}(p)} = \sum\limits_{q: (q_1,\dots,q_{n-3})=(p_1,\dots,p_{n-3})} v_{q}^{(0)},
     \text{ $p$ is $3$-increasing}.
\end{align}
Continue in a similar manner, in the $j$-th iteration, $1 \le j \le n-1$, the variable is
\begin{align} \label{eq:j_iter_variable}
    v_p^{(j)} = 
    \sum\limits_{k=n-j}^{n} v^{(j-1)}_{g_{j,k}(p)} =
    \sum\limits_{q: (q_1,\dots,q_{n-j-1})=(p_1,\dots,p_{n-j-1})} v_{q}^{(0)},
    \text{ $p$ is $(j+1)$-increasing}.
\end{align}
After the $(n-1)$-th iteration, the $p=(1,2,\dots,n)$-th variable becomes
\begin{align}\label{eq:last_iter_variable}
    v_{(1,2,\dots,n)}^{(n-1)} = \sum_{\text{all permutations $q$}} v_q^{(0)}, \text{ $p=(1,2,\dots,n)$ is $n$-increasing.}
\end{align}
%\item
%A permutation $p$ is said to be \emph{$j$-increasing} if the last $j$ positions are increasing:
%$$p_{n-j+1} < p_{n-j+2}< \dots < p_{n}.$$
%\item
%If a permutation is $2$-increasing, it is said to be \emph{odd}. Otherwise, it is \emph{even}.
% \item
% Define a function $g_{j,k}(\cdot)$ that takes the $k$-th element of $p$ and place it at the $(n-j)$-th position:
% \begin{align}
% g_{j,k}(p) =& (p_1,...,p_{n-j-1}, p_k, p_{n-j},...,p_{k-1},p_{k+1},...,p_n ), \quad k = n-j,...,n. 
% \end{align}
% Denote by $g_{j,k}^{-1}(\cdot)$ the inverse permutation.
% \item
% Define a permutation $h_i(\cdot)$ that takes the $i$-th element of $p$ and place it at the first position:
% \begin{align}
% h_{i}(p)=&(p_i,p_1,...,p_{i-1},p_{i+1},...,p_n).
% \end{align}
% Denote by $h_i^{-1}(\cdot)$ the inverse permutation.
\item
Let $Q$ be the set of all $j$-increasing permutation, $|Q|=\frac{n!}{j!}$. Then it can be partitioned as below:
\begin{align}
    Q = \bigcup\limits_{p: (j+1)\text{-increasing}} \{ g_{j,k}(p), n-j \le k \le n\}.
\end{align}
In other words, any permutation that is $j$-increasing but not $(j+1)$-increasing can be written as $g_{j,k}(p)$, for some $(j+1)$-increasing $p$ and $k$, $n-j+1 \le k \le n$.
% \item
% Define a permutation $h_i(\cdot)$ that takes the $i$-th element of $p$ and place it at the first position:
% \begin{align}
% h_{i}(p)=&(p_i,p_1,...,p_{i-1},p_{i+1},...,p_n).
% \end{align}
% Denote by $h_i^{-1}(\cdot)$ the inverse permutation.
\item
% Define a permutation $Incr_2(\codt)$ whose result is 2-increasing:
% \begin{align}
%     Incr_2(p) = \begin{cases}
%     p, & \text{if } p_{n-1} < p_n,\\
%     (p_1,\dots,p_{n-2},p_n,p_{n-1}), & \text{if } p_{n-1}>p_n.
%     \end{cases}
% \end{align}
If $p$ is $j$-increasing, define a permutation $Incr_{j+1}(\cdot)$ on $p$ whose result is $(j+1)$-increasing, $1 \le j \le n-1$:
\begin{footnotesize}
\begin{align} \label{eq:Incr}
    Incr_{j+1}(p) = \begin{cases}
    p, & \text{if } p_{n-j} < p_{n-j+1},\\
    (p_1,\dots, p_{n-j-1}, p_{n-j+1},..., p_{k}, p_{n-j}, p_{k+1}, ..., p_n), & \text{if }  p_{k} < p_{n-j} < p_{k+1}, k=n-j+1, ..., n.
    \end{cases}
\end{align}
\end{footnotesize}
Here we define $p_{n+1} = \infty$. We also define $Incr_1(\cdot)$ as the identity permutation.
% From \eqref{eq:Incr} we see that $$Incr_{j+1}(p) = g^{-1}_{j,k}(p),$$ for some $k$, $n-j \le k \le n$.
% For an arbitrary permutation $p$, define $Incr_{j+1}$ recursively:
% \begin{align}
%     Incr_{j+1}(p) = Incr_j\Big( Incr_{j-1} \big(\dots Incr_2(p)  \big)\Big).
% \end{align}
For example,
$Incr_2(1,3,4,2)=(1,3,2,4)$, $Incr_3(1,3,2,4)= (1,2,3,4)$.
\item We see that if $p$ is $(j+1)$-increasing, 
\begin{align}\label{eq:Incr_gjk}
    Incr_{j+1}(g_{j,k}(p)) = p, k=n-j,\dots,n.
\end{align}
For example, $Incr_3(g_{2,4}(2,1,3,4))=Incr_3(2,4,1,3) = (2,1,3,4)$.
\item 
If $p, q$ are $(j+1)$-increasing, there exits at most one $i$, $1 \le i \le n$, such that $g_{j,k}(q)=Incr_j(h_i(p))$ for some $n-j \le k \le n$.  This is because $Incr_j(h_i(p))$ starts with $p_i$, and $g_{j,k}(q)$ starts with $q_1$, but it is not possible to have $q_1 = p_i = p_{i'}$ for distinct $i,i'$. Moreover, we get
\begin{align}
    q =& g_{j,k}^{-1}(Incr_j(h_i(p))) \\
    =& Incr_{j+1}(Incr_j(h_i(p))) \label{eq:74}\\
    =& Incr_{j+1}(h_i(p)), \label{eq:75}
\end{align}
where \eqref{eq:74} follows from \eqref{eq:Incr_gjk}, and \eqref{eq:75} follows from $h_i(p)$ being already $j$-increasing.
\end{itemize}
\begin{proof} (\textbf{Lemma \ref{ss}})
There are $n!$ states and therefore $n!$ steady states probabilities $\pi_{q}$ where $ \sum _{q \in \mathcal{Q}}\pi_{q}=1$.
 Considering the fact that our Markov chain $q(t)$ is ergodic, the state probabilities always converge to a unique solution satisfying the following equations: 
\begin{equation}\label{eq:pi}
{{\pi}}_{{q}} \sum_{l \in L_{{q}}}^{} \lambda^{(l)}=
 \sum_{l \in L^\prime_{{q}}}^{}
\lambda^{(l)}\pi_{q_{l}}, \quad {q} \in \mathcal{Q}.
\end{equation} 
Self-loops do not need to be considered in the above equation because they will be cancelled.
Without loss of generality  let us suppose our state is $q=(q_{1},q_{2},...,q_{n})$ which is a permutation of $\{1,2,...,n\}$. The incoming states of state $q$ (excluding self-loops) are $H_q^{-1}$ as defined in \eqref{eq:H_q_-1}, with incoming rates $\lambda_{q_{1}}$. The outgoing states of $q$ excluding self-loops have rates $\lambda_{q_i}$ for $i\in [n]$. Therefore, the incoming states are
\begin{align}
s_1=h_{1}^{-1}(q)&=(q_{1},q_{2},q_{3},...,q_{n})=q, \nonumber \\
s_2=h_{2}^{-1}(q)&=(q_{2},q_{1},q_{3},...,q_{n}), \nonumber\\
s_3=h_{3}^{-1}(q)&=(q_{2},q_{3},q_{1},...,q_{n}), \nonumber \\ 
	\quad	\quad \quad \vdots 	&	\quad	\quad \quad \quad \quad	\vdots   \label{eq:set}\\
s_n=h_{n}^{-1}(q)&=(q_{2},q_{3},...,q_{n},q_{1}),\nonumber
\end{align}
and Equation \eqref{eq:pi} is
\begin{align}
\pi_{s_1}\sum_{i=1}^{n} {\lambda_{q_{i}}} &= \lambda_{q_{1}} \sum_{i=1}^n  {\pi_{s_i}}, \quad  {q=s_1} \in \mathcal{Q}. \label{steady}
\end{align} 
Knowing the fact that the steady-state probability is unique, we only need to find a solution for each $\pi_{q}$ which satisfies the equations in \eqref{steady}. Next we prove that the following $\pi_{q}$ satisfies \eqref{steady} and it is a probability function:
\begin{equation} \label{steady2}
\pi_{q}= \frac{\lambda_{q_{1}}}{\sum_{j=1}^{n} \lambda_{q_{j}}} \frac{\lambda_{q_{2}}}{\sum_{j=2}^{n} \lambda_{q_{j}}} \frac{\lambda_{q_{3}}}{\sum_{j=3}^{n} \lambda_{q_{j}}} ... \frac{\lambda_{q_{n-1}}}{\sum_{j=n-1}^{n} \lambda_{q_{j}}}.
\end{equation}

First, let us verity that \eqref{steady2} is a probability function that sums to 1. %where $q= (q_1,q_2,...,q_n)$ is any  permutation of $(1,2,...,n)$.
Consider two permutations (states) that only differ in the last two elements. Therefore based on \eqref{steady2} their steady-state probabilities only differ in the last terms which are
$\frac{\lambda_{q_{n}}}{q_{n}+q_{1}}$ and $\frac{\lambda_{q_{1}}}{q_{1}+q_{n}}$, respectively. Consequently, if we add them together, as the other $n-2$ terms are the same and $\frac{\lambda_{q_{n}}}{q_{n}+q_{1}}+\frac{\lambda_{q_{1}}}{q_{1}+q_{n}}=1$ we can cancel the last term in \eqref{steady2}. Then, consider 6 permutations that only differ in the last $3$ elements. Using a similar argument, the last $2$ terms will be cancelled after we sum their probabilities. Continue in a similar argument, the total probability sum is 1.
% \begin{equation*}
% \pi_{s_n}+  \pi_{s_{n-1}}=\frac{\lambda_{q_{2}}}{\sum_{i=1}^{n} {\lambda_{q_{i}}}} \frac{\lambda_{q_{3}}}{\sum_{i=1}^{n} {\lambda_{q_{i}}} -\lambda_{q_{2}}}... \frac{\lambda_{q_{n-1}}}{\lambda_{q_{n-1}}+\lambda_{q_{n}}+\lambda_{q_{1}}}.
% \end{equation*}

Next, we prove that \eqref{steady2} satisfies \eqref{steady} by induction.

{\bf Base case.} For $n=2$, it is easy to check \eqref{steady2} satisfies \eqref{steady}.

{\bf Induction step.}
Let us suppose \eqref{steady2} satisfies \eqref{steady} for the case that we have $n-1$ servers and consequently $(n-1)!$ states. We prove that Equation \eqref{steady2} holds for $n$. Let us remove server $q_{2}$ consider the $n-1$ servers $q_{1},q_{3},...,q_{n}$ with rates $\lambda_{q_{1}},\lambda_{q_{3}},...,\lambda_{q_{n}}$, respectively. For state $q^{\prime}=(q_{1},q_{3}, q_{4},..., q_{n})$, the incoming states are
\begin{align}
&s_{1}^{\prime}= h_1^{-1}(q')=(q_{1},q_{3}, q_{4},..., q_{n}) = q', \nonumber\\
&s_{2}^{\prime}= h_2^{-1}(q')=(q_{3},q_{1}, q_{4},..., q_{n}), \nonumber\\
&s_{3}^{\prime}= h_3^{-1}(q')=(q_{3},q_{4}, q_{1},..., q_{n}), \nonumber\\
&\qquad \vdots \qquad \quad \quad \quad 	\quad	\vdots \label{eq:set2}\\
&s_{n-1}^{\prime}= h_{n-1}^{-1}(q')=(q_{3},q_{4}, q_{5},..., q_{1}),  \nonumber
\end{align}
and the corresponding Equation \eqref{eq:pi} is
\begin{align}
&\pi_{s'_1}(\lambda_{q_{1}}+\lambda_{q_{3}}+...+\lambda_{q_{n}})= \lambda_{q_{1}} \sum_{i=1}^{n-1}  {\pi_{s_i^{\prime}}}. \label{steady3}
\end{align}
Notice that the first term $\pi_{s'_1}\lambda_{q_1}$ on both sides in the above equation is identical and can be cancelled.
Comparing the set of states in \eqref{eq:set} and \eqref{eq:set2} and pluging them into \eqref{steady2}, we observe that
$\pi_{s_{i}}= \frac{\lambda_{q_{2}}}{\sum_{j=1}^{n} {\lambda_{q_{j}}}}\pi_{s_{i-1}^{\prime}}$ for $i \in \{3,4,...,n\}$. 
Therefore:
\begin{equation} \label{indi}
\sum_{i=3}^{n} {\pi_{s_{i}}} = \frac{\lambda_{q_{2}}}{\sum_{j=1}^{n} \lambda_{q_{j}}} \sum_{i=3}^{n}\pi_{s_{i-1}^{\prime}}.
\end{equation} 
Using the induction assumption that $\{\pi_{s'_i}\}$ satisfy \eqref{steady3}, we know that:
\begin{align*}
& \lambda_{q_{1}} \sum_{i=3}^{n}\pi_{s_{i-1}^{\prime}} = 
\lambda_{q_{1}} \sum_{i=2}^{n-1}\pi_{s_{i}^{\prime}}=
 (\sum_{i=3}^{n} {\lambda_{q_{i}}}) \pi_{s_{1}^{\prime}} \\ 
 =&(\sum_{i=3}^{n} {\lambda_{q_{i}}}) \frac{\lambda_{q_{1}}}{\sum_{i=1}^{n} {\lambda_{q_{i}}} -\lambda_{q_{2}}} \frac{\lambda_{q_{3}}}{\sum_{i=3}^{n}{\lambda_{q_{i}}}} \frac{\lambda_{q_{4}}}{\sum_{i=4}^{n} {\lambda_{q_{i}}}}...\frac{\lambda_{q_{n-1}}}{\sum_{i=n-1}^{n} {\lambda_{q_{i}}}}.
\end{align*} 
As a result we can now calculate $\sum_{i=3}^{n} \pi_{s_{i}}$ as follows:
\begin{align}
\sum_{i=3}^{n} {\pi_{s_{i}}}= \frac{\lambda_{q_{2}}}{\sum_{j=1}^{n} \lambda_{q_{j}}} \frac{\sum_{i=3}^{n} {\lambda_{q_{i}}}}{\lambda_{q_{1}}}    \frac{\lambda_{q_{1}}}{\sum_{i=1}^{n} {\lambda_{q_{i}}} -\lambda_{q_{2}}}  \frac{\lambda_{q_{3}}}{\sum_{i=3}^{n}{\lambda_{q_{i}}}} \frac{\lambda_{q_{4}}}{\sum_{i=4}^{n} {\lambda_{q_{i}}}}...\frac{\lambda_{q_{n-1}}}{\sum_{i=n-1}^{n} {\lambda_{q_{i}}}}. \label{eq:sum_pi}
\end{align}
Now we calculate the term $\pi_{s_{2}}+\sum_{i=3}^{n} \pi_{s_{i}}$ using \eqref{steady2} and \eqref{eq:sum_pi}:
\begin{align*}
 & \pi_{s_{2}}+\sum_{i=3}^{n} \pi_{s_{i}} \\
 =& \frac{\lambda_{q_{2}}}{\sum_{i=1}^{n} {\lambda_{q_{i}}}} 
 \frac{\lambda_{q_{1}}}{\sum_{i=1}^{n} {\lambda_{q_{i}} -\lambda_{q_{2}}}}
 \frac{\lambda_{q_{3}}}{\sum_{i=3}^{n} \lambda_{q_{i}}} \dots +\frac{\lambda_{q_{n-1}}}{\sum_{i=n-1}^{n} \lambda_{q_{i}}} \\
 & +
\frac{\lambda_{q_{2}}}{\sum_{i=1}^{n} \lambda_{q_{i}}} \frac{\sum_{i=3}^{n} {\lambda_{q_{i}}}}{\lambda_{q_{1}}}       \frac{\lambda_{q_{1}}}{\sum_{i=1}^{n} {\lambda_{q_{i}}} -\lambda_{q_{2}}} \frac{\lambda_{q_{3}}}{\sum_{i=3}^{n}{\lambda_{q_{i}}}}  \dots \frac{\lambda_{q_{n-1}}}{\sum_{i=n-1}^{n} {\lambda_{q_{i}}}} \\
=& 
\frac{\lambda_{q_{1}}}{\sum_{i=1}^{n} {\lambda_{q_{i}}} -\lambda_{q_{2}}} \frac{\lambda_{q_{3}}}{\sum_{i=3}^{n}{\lambda_{q_{i}}}} ...\frac{\lambda_{q_{n-1}}}{\sum_{i=n-1}^{n} {\lambda_{q_{i}}}} \frac{\lambda_{q_{2}}}{\sum_{i=1}^{n} {\lambda_{q_{i}}}} (1+\frac{\sum_{i=3}^{n} {\lambda_{q_{i}}}}{\lambda_{q_{1}}})\\ 
=& 
\frac{\lambda_{q_{1}}}{\sum_{i=1}^{n} {\lambda_{q_{i}}} -\lambda_{q_{2}}} \frac{\lambda_{q_{3}}}{\sum_{i=3}^{n}{\lambda_{q_{i}}}}... \frac{\lambda_{q_{n-1}}}{\sum_{i=n-1}^{n} {\lambda_{q_{i}}}} \frac{\lambda_{q_{2}}}{\sum_{i=1}^{n} {\lambda_{q_{i}}}} \frac{\sum_{i=1}^{n} {\lambda_{q_{i}}} -\lambda_{q_{2}}}{\lambda_{q_{1}}} \\
=& 
\frac{\lambda_{q_{3}}}{\sum_{i=3}^{n}{\lambda_{q_{i}}}}...\frac{\lambda_{q_{n-1}}}{\sum_{i=n-1}^{n} {\lambda_{q_{i}}}} \frac{\lambda_{q_{2}}}{\sum_{i=1}^{n} {\lambda_{q_{i}}}} .
\end{align*}
Now we verify that \eqref{steady} holds, which can be rewritten as
\begin{align} \label{s11}
 \pi_{s_{1}} =
\frac{\lambda_{q_{1}}}{\sum_{i=2}^{n} {\lambda_{q_{i}}}} \sum_{i =2}^{n}
{\pi_{s_i}}  = \frac{\lambda_{q_{1}}}{\sum_{i=2}^{n} {\lambda_{q_{i}}}} (\pi_{s_{2}}+\sum_{i=3}^{n} \pi_{s_{i}}).
\end{align}
By substituting $\pi_{s_{2}}+\sum_{i=3}^{n} \pi_{s_{i}}$
 and simplification, we find the right-hand side of \eqref{s11} as:
\begin{align*}
RHS = \frac{\lambda_{q_{1}}}{\sum_{i=2}^{n} {\lambda_{q_{i}}}} \frac{\lambda_{q_{3}}}{\sum_{i=3}^{n}{\lambda_{q_{i}}}}...\frac{\lambda_{q_{n-1}}}{\sum_{i=n-1}^{n} {\lambda_{q_{i}}}} \frac{\lambda_{q_{2}}}{\sum_{i=1}^{n} {\lambda_{q_{i}}}}= 
\frac{\lambda_{q_{1}}}{\sum_{i=1}^{n} {\lambda_{q_{i}}}} \frac{\lambda_{q_{2}}}{\sum_{i=2}^{n} {\lambda_{q_{i}}}} \frac{\lambda_{q_{3}}}{\sum_{i=3}^{n}{\lambda_{q_{i}}}}... \frac{\lambda_{q_{n-1}}}{\sum_{i=n-1}^{n} {\lambda_{q_{i}}}}.    
\end{align*}
Plugging $q=s_1$ in \eqref{steady2}, the left-hand side of \eqref{s11} is equal to the right-hand side. Thus, \eqref{steady2} is a solution to \eqref{steady}. The proof is completed.
\end{proof}

\begin{proof} (\textbf{Lemma \ref{gen_mat}})
Let us take all the terms other than $\pi_q$ to the left side of Equation \eqref{general_eq}. It is clear that the constant vector should be $\bm{\pi}$.

From Equation \eqref{general_eq}, we notice that 
\begin{align} \label{step1}
({\sum_{j=1}^{n} {\lambda_{q_j}}}) v_{q,q_{1}} = {\pi_{q}}.
\end{align}
Therefore, the case of $i=1$ holds in \eqref{co_t}.
%now we know the solution to $n!$ variables out of $(n+1)!$ when $i=1$. 
We notice that when $i > 1$, in the row $(q,i)$ and column $(p,i)$ of matrix $T$ there is a non-zero entry of $-\lambda_{q_{1}}$, if and only if $p$ is an incoming state of $q$ or in another word $q= h_{j} (p)$. %Therefore, the entry of $T((q,i),(p,k))$ becomes $-\lambda_{q_{1}}$ when  $q= h_{j} (p)$ and if $i>j$ then $k$ has to be qual to $i$ since    .
When $(q,i)=(p,k)$, the coefficient of variable $v_{q ,q_{i}}$ is equal to $\sum_{l=1}^{n} (\lambda_{q_{l}} +\mu_{q_{l}})$ minus 2 parts. One part is minus $\lambda_{q_{1}}$ because every state is also an incoming state for itself, and the second part is $\sum_{j=i}^{n} \mu_{q_{j}}$ because whenever we are having an update delivery $\mu_{q_{j}}$ where $j \geq i$, we will have $-\mu_{q_{j}}$ multiplied by the variable $v_{q,q_{i}}$. Therefore the coefficient for variable $v_{q,q_{i}}$ at the end becomes $\sum_{l=1}^{n} (\lambda_{q_{l}} +\mu_{q_{l}}) -\lambda_{q_{1}} - \sum_{j=i}^{n} \mu_{q_{j}}= \sum_{l=2}^{n} \lambda_{q_{l}}  +\sum_{l=1}^{i-1} \mu_{q_{l}}$. Also, when we have an update delivery from server $q_{k}$ with rate $\mu_{q_{k}}$ where $k<i$, since we preempt the older update, we will have the term $-\mu_{q_{k}} v_{q,q_{k}}$ in Equation \eqref{general_eq}, i.e., an entry of $-\mu_{q_{k}}$ in $T((q,i),(p,k))$ when $q=p$ and $k<i$. The cases of $i=0$ can be similarly verified. Obviously the rest of entries of matrix are equal to zero.
\end{proof}
\begin{proof} (\textbf{Lemma \ref{break_gen}})
We take an inductive approach and prove that for iteration $i=1,2,\dots,n$,
variables 
\begin{align}\label{eq:variables_to_solve}
    \{v_{q,q_i}: (q_i,\dots,q_n) \text{ is fixed} \}
\end{align}
can be solved iteratively (Line \ref{line:breakdown} of Algorithm \ref{alg:heter_overall}).
For iteration $i=n+1$, variables 
\begin{align}\label{eq:variables_q0}
    \{v_{q,0}: \text{ all permutations } q\}
\end{align}
can be solved (Line \ref{line:AoI}).
And the associated coefficient matrix for each iteration are in the same form as $T_0$.

{\bf Base case.} For $i=1$, we can solve the variable $ v_{q,q_{1}} = \frac{\pi_{q}}{\sum_{j=1}^{n} {\lambda_{q_j}}}$ according to the case of $i=1$ in \eqref{co_t} or Equation \eqref{step1}. Since the coefficient is a scalar, it is in the same form as $T_0$ parameterized by $0$.

{\bf Induction step.}
%If we look at the equations defined by $T$ from lemma \ref{gen_mat} consisting of the variable $v_{q,q_{1}}$, we easily can see that it forms a diagonal $n!\times n!$ matrix which can be solved and it is also demonstrated algebraically in Equation \eqref{step1}.
%In the following we focus on solving variables $v_{q,q_i}$ for $i>1$. 
The induction hypothesis is that we have solved the variables $v_{q,q_1},..., v_{q,q_{i-1}}$, for $2 \le i \le n$. 
Now we solve the $(i-1)!$ variables in \eqref{eq:variables_to_solve}.
%For solving variables $v_{q,q_i}$, we fix $q_i,q_{i+1},\dots,q_n$. There are $(i-1)!$ such $q$ that ends with $q_i,q_{i+1},\dots,q_n$. We show that these variables can be solved based on the induction hypothesis.
Consider the $(i-1)! \times (i-1)!$ sub-matrix of $T$ associated with rows/columns $(q,i)$ in which $(q_{i},q_{i+1},...,q_{n})$ are fixed, denoted by $T_0^i$. We will prove $T_0^i$ contains all the non-zero entries of $T((q,i),:)$ unless we have already solved the corresponding variables. Namely, the equations defined by $T_0^i$ can be used to solve the variables in \eqref{eq:variables_to_solve} given the induction hypothesis.

The first case of non-zero entries in $T((q,i),:)$ is for
for $i>1$, $p=q$, and $k=1,2,...,i-1$, but based on induction hypothesis we have already solved variables for values of $k<i$. Another case of non-zero entries of $T((q,i),:)$ is when $(q,i)=(p,k)$, which is contained in the sub-matrix $T_0^i$.
The third case of non-zero entries is when $q= h_j(p), k = \langle i \rangle_j $ for $j=2,...,n$. %When we only consider the rows and columns corresponding to $(q,i)$ where $q_{i},q_{i+1},...,q_{n}$ are fixed, we are considering the entries corresponding to the case $(q,i)=(p,k)$. If $q=(q_1,...,q_n)=h_j (p)$, then $p= (q_2,...,q_j,q_1,q_{j+1},...,q_n)$. Also we know $j=2,...,n$ and $k= \langle i \rangle_j$. 
Note that $p= (q_2,...,q_j,q_1,q_{j+1},...,q_n)$. 
If $i\geq j+1, k=\langle i \rangle_j = i$ and the last $n-i$ values of $p$ are in the form 
of $(q_i,q_{i+1},...,q_n)$, which are included in the sub-matrix $T_0^i$. If $i \leq j$, then $k=\langle i \rangle_j=i-1$ and we have solved these variables based on the induction hypothesis.
In summary, having excluded the solved variables based on the hypothesis, the variables \eqref{eq:variables_to_solve} can be solved according to the coefficient matrix $T_{0}^{i}$ where $q_{i},q_{i+1},...,q_{n}$ are fixed for $p,q$: 
\begin{align}
T_{0}^{i}(q,p)=
\begin{cases}
\sum_{j=2}^{n} \lambda_{q_j} + \sum_{j=1}^{i-1} \mu_{q_{j}} , & \text{if } q=p, \\
-\lambda_{q_{1}}, & \text{if } q = h_j(p), \quad j=2,\dots,i-1,\\
0 , & \text{o.w.}
\end{cases}
\end{align}
This is exact the same general format of $T_0$ parameterized by $i-1$ as in \eqref{eq:T_0} if we replace $\mu_{q_{1}}$ by $\mu_{q_1}+\sum_{j=i}^{n}\lambda_{q_j}$.

%In the following we solve a more general matrix form named $T_0$. After proving that coefficient matrix $T_0$ is solvable, here in lemma \ref{break_gen}, we use this result in an inductive fashion for solving $T_{0}^{i}$ and therefore solving $v_{q,q_{i}}$. In summary, for solving  $Tv=c$ we use the algorithm \ref{alg:heter_overall} by solving the following general matrix $T_0$ for any value of $n$. Once we prove that, the proof of lemma \ref{break_gen} is complete.
% \begin{align}%\label{eq:T_0}
% T_{0}(q,p)=
% \begin{cases}
% \sum_{i=2}^{n} \lambda_{q_i}  +\sum_{i=1}^{n} \mu_{q_i} , & \text{if } q=p, \\
% -\lambda_{q_{1}}, & \text{if } q = h_i(p), i=2,\dots,n,\\
% 0 , & \text{o.w.}
% \end{cases}
% \end{align}
% where $p=(p_{1},p_{2},...,p_{n})$, $q=(q_{1},q_{2},...,q_{n})$ are any permutations of numbers $(1,2,...,n)$.  We can see that $T_{0}^{i}$ is a special case of $T_{0}$ when $n=i-1$. Moreover, we define $T_0$ parameterized by $n=1$ to be the scalar 
% \begin{align}\label{eq:T_0_n1}
%     T_0 = \mu_1,   
% \end{align}
% and $T_0$ parameterized by $n=0$ to be the scalar 
% \begin{align}\label{eq:T_0_n0}
%     T_0 = \lambda_1+ \mu_1.    
% \end{align} 

%With some column and row operations on matrix $T_{0}$ and utilizing induction, we show that we can solve these $n!$ equations in algorithm \ref{alg:heter}. 
For the final iteration $i=n+1$, we solve the variables in \eqref{eq:variables_q0}.
Let $T_0^{n+1}$ be the $n! \times n!$ sub-matrix of $T$ whose 
rows/columns are $(q,0)$, for any $q$. 
Let us consider the 3 types of non-zero entries of the equation $T((q,0),:)$. 
%Based on \eqref{co_t}, the non-zero entries are of $3$ types. 
First there are non-zero entries %of $\sum_{l=2}^{n} \lambda_{q_{l}}+ \sum_{l=1}^{n} \mu_{q_{l}}$ 
when $i=0$ and $(q,i)=(p,k)$ already included in $T_0^{n+1}$. %The other entries are $-\lambda_{q_{1}} $ 
The second type is when $i=k=0$ and $q = h_j(p)$ for  $j=2,\dots,n,$, which are also included in $T_0^{n+1}$. The last type is %entries are $-\mu_{q_{k}}$ 
when $i=0$ and $q =p$  for $k=1,\dots,n$. These entries are coefficients of the variables that are already solved based on the hypothesis, hence can be excluded.
Therefore, for solving variables \eqref{eq:variables_q0}, it is sufficient to only consider $T_0^{n+1}$: 
\begin{align}
T_{0}^{n+1}(q,p)=
\begin{cases}
\sum_{i=2}^{n} \lambda_{q_i}  +\sum_{i=1}^{n} \mu_{q_i} , & \text{if } q=p, \\
-\lambda_{q_{1}}, & \text{if } q = h_i(p), i=2,\dots,n,\\
0 , & \text{o.w.}
\end{cases}
\end{align}
This is exact same general format of $T_0$ parameterized by $n$ as in \eqref{eq:T_0}.
\end{proof}

Below is a lemma useful to show the correctness of Algorithm \ref{alg:heter}.
\begin{lemma} \label{lem:column operation} 
Let $p$ be a $(j+1)$-increasing permutation. Fix $j,k$, for some $1 \le j \le n-1, n-j+1 \le k \le n$. Then $Incr_j(h_i(p))$, $Incr_j(h_{i}(g_{j,k}(p)))$, $1 \le i \le n$ are as follows.
\begin{align}
    & Incr_j(h_i(p))  =  (p_i,p_1,\dots,p_{i-1},p_{i+1},\dots,p_{n-j},\dots,p_n) = h_i(p), i=1,2,\dots,n-j-1.\\
    & Incr_j(h_{i}(g_{j,k}(p)))  = (p_i,p_1,\dots,p_{i-1},p_{i+1},\dots,p_{n-j-1},p_k,p_{n-j},\dots,p_{k-1},p_{k+1},\dots,p_n)  = g_{j,k}(h_{i}(p)),\\
    & \quad  i=1,2,\dots,n-j-1. \nonumber \\
    & Incr_j(h_i(p))  =  Incr_j(h_{i+1}(g_{j,k}(p))) = (p_i,p_1,\dots,p_{i-1},p_{i+1},\dots,p_n), i=n-j,\dots,k-1 \label{eq:index_column_op}\\
    & Incr_j(h_k(p))  =  Incr_j(h_{n-j}(g_{j,k}(p))) = (p_k,p_1,\dots,p_{k-1},p_{k+1},\dots,p_n).\\
    & Incr_j(h_i(p))  =  Incr_j(h_{i}(g_{j,k}(p))) = (p_i,p_1,\dots,p_{i-1},p_{i+1},\dots,p_n), i=k+1,\dots,n.
\end{align}
\end{lemma}
\begin{proof} (\textbf{Lemma \ref{lem:column operation} })
The proof follows immediately from the definitions of the permutations. We show \eqref{eq:index_column_op} below, and the remaining equations can be shown similarly. For $n-j \le i \le k-1$,
$$h_i(p) = (p_i,p_1,\dots,p_{n-j-1}, p_{n-j}, \dots, p_{i-1}, p_{i+1},\dots, p_n).$$
Noting that the last $j$ elements $p_{n-j}, \dots, p_{i-1}, p_{i+1},\dots, p_n$ are already increasing, we get
$$Incr_j(h_i(p)) = h_i(p).$$
Moreover, for $n-j+1 \le i+1 \le k$,
\begin{align}
     & h_{i+1}(g_{j,k}(p)) \\
    =& h_{i+1}(p_1,\dots,p_{n-j-1}, p_k, p_{n-j}, \dots, p_{k-1}, p_{k+1}, \dots, p_n) \\
    =& (p_i,p_1,\dots,p_{n-j-1}, p_k, p_{n-j}, \dots, p_{i-1}, p_{i+1},\dots, p_{k-1},p_{k+1},\dots, p_n).
\end{align}
If we apply $Incr_j$ to the above permutation, the last $j$ positions become $p_{n-j},\dots,p_{i-1},p_{i+1},\dots, p_n$. Thus
$$Incr_j(h_{i+1}(g_{j,k}(p))) = (p_i,p_1,\dots,p_{n-j-1},p_{n-j},\dots,p_{i-1},p_{i+1},\dots, p_n) = Incr_j(h_i(p)).$$
The equation is proved.
\end{proof}

\begin{proof}  (\textbf{Lemma \ref{lem:alg2_output}}) \\
% Note: $(1)$ Since $T_j$ is indexed by $(j+1)$-increasing permutations, its size is $\frac{n!}{(j+1)!} \times \frac{n!}{(j+1)!}$. $(2)$ From \eqref{eq:T_j}, the non-zero entries of $T_{j-1}$ in column $p$ are indexed by 
% \begin{align} \label{eq:index_T_j}
%  (Incr_j( h_i(p) ), p), 1 \le i \le n.    
% \end{align}
%
%In the following we explain and prove the algorithm \ref{alg:heter}.\\
%
\textbf{Claim C1.}
After the $j$-th iteration in the forward path of Algorithm \ref{alg:heter}, $0 \le j \le n-2$ we get for $(j+1)$-increasing permutations $q,p$,
\begin{align} \label{eq:T_j}
T_{j} (q,p) =
\begin{cases}
\sum_{i=1}^{n} (\lambda_i +\mu_i ) - \lambda_{q_{1}}, & \text{if } q=p \\
-\lambda_{q_{1}}, & \text{if } q =  Incr_{j+1}( h_i(p) ), i=2,3,\dots,n, \\
0 , & \text{o.w.}
\end{cases}
\end{align}
Note that the matrix is of size $\frac{n!}{(j+1)!} \times \frac{n!}{(j+1)!}$ since $q,p$ are $(j+1)$-increasing. We prove the claim by induction.

The base case is trivial for $j=0$.

Assume C1 holds after the $(j-1)$-th iteration. We will show C1 for the $j$-th iteration. From \eqref{eq:T_j}, the non-zero entries of $T_{j-1}$ in column $p$ are indexed by 
\begin{align} \label{eq:index_T_j}
 (Incr_j( h_i(p) ), p), 1 \le i \le n.    
\end{align}

After Line \ref{line:T_j'}, $T_j^{\prime}=T_{j-1}$. The rows and columns are $j$-increasing.

After the column operations  in Line \ref{line:column}, for $(j+1)$-increasing $p$, column $p$ does not change. Otherwise, for column $g_{j,k}(p)$, $n-j+1 \le k \le n$,
$$ T_j^{\prime}(:, g_{j,k}(p))
    = T_{j-1}(:, g_{j,k}(p)) - T_{j-1}(:, p). $$  
Lemma \ref{lem:column operation} compares the row indices of the non-zero entries in columns $ g_{j,k}(p)$ and  $p$, which are $\{Incr_j(h_i(g_{j,k}(p))),$ $1 \le i \le n\}$ and $\{Incr_j(h_{i}(g_{j,k}(p))), 1 \le i \le n\}$.
We get for $(j+1)$-increasing $p$ and $j$-increasing $q$,
\begin{align} \label{eq:90}
    & T_j^{\prime}(q, g_{j,k}(p)) 
    =& \begin{cases}
    T_{j-1}(q, g_{j,k}(p)) = \sum_{i=1}^{n} (\lambda_i +\mu_i)  - \lambda_{q_{1}}, & \text{if } q=g_{j,k}(p),\\
    -T_{j-1}(q, p) \quad = -\sum_{i=1}^{n} (\lambda_i +\mu_i)  + \lambda_{q_{1}}, & \text{if } q=p \\
    T_{j-1}(q, g_{j,k}(p)) = -\lambda_{q_1}, & \text{if } q = g_{j,k}(h_i(p)), 2 \le i \le n-j-1,\\
    -T_{j-1}(q, p) \quad = \lambda_{q_1}, & \text{if } q = h_i((p)), 2 \le i \le n-j-1,\\
    0, & \text{o.w.}
    \end{cases}
\end{align}

After the row operations in Line \ref{line:row}, for a row that is not $(j+1)$-increasing, it does not change. Otherwise,
\begin{align}\label{eq:row_op}
    T_j^{\prime\prime}(q,:) = \sum_{k=n-j}^{n}(T_{j}^{\prime}(g_{j,k}(q),:)).
\end{align}
First, consider column $p$ that is $(j+1)$-increasing, and the  non-zero entries of $T_j^{\prime}(:,p) = T_{j-1}(:,p)$ are indexed by \eqref{eq:index_T_j}. Using the observation from \eqref{eq:75}, there is at most one non-zero term in the sum of \eqref{eq:row_op}, and
\begin{align} \label{eq:91}
    T_j^{\prime\prime}(q,p) = \begin{cases}
    \sum_{i=1}^{n} (\lambda_i +\mu_i)  - \lambda_{q_{1}}, & \text{if } q=p,\\
    -\lambda_{q_1}, & \text{if } q=Incr_j(h_i(p)), \text{ or } q = Incr_{j+1}(h_i(p)), i=2,3,\dots,n \\
    0, & \text{o.w.}
    \end{cases}
\end{align}
Second, consider column $g_{j,k}(p)$ that is $j$-increasing but not $(j+1)$-increasing, whose entries are in \eqref{eq:90}. 
For row $q = p$ which is $(j+1)$-increasing, the first two cases of \eqref{eq:90} will be added and canceled according to \eqref{eq:row_op}. Similarly for row $q=h_i(p)$ which is $(j+1)$-increasing, Cases 3 and 4 of \eqref{eq:90} will be added  and canceled. Thus,
\begin{align} \label{eq:92}
     T_j^{\prime\prime}(q, g_{j,k}(p)) 
    = \begin{cases}
    \sum_{i=1}^{n} (\lambda_i +\mu_i)  - \lambda_{q_{1}}, & \text{if } q=g_{j,k}(p),\\
     -\lambda_{q_1}, & \text{if } q = g_{j,k}(h_i(p))=h_i(g_{j,k}(p)), 2 \le i \le n-j-1,\\
    0, & \text{o.w.}
    \end{cases}
\end{align}
From \eqref{eq:92} one can see that for row/equation $q$ that is $(j+1)$-increasing, $T_j^{\prime\prime}(q, g_{j,k}(p)) = 0$. Namely, $(j+1)$-increasing equations  only involve $(j+1)$-increasing variables. Thus we only need to find $T_j$ as in Line \ref{line:pick_T_j} in Algorithm \ref{alg:heter} and solve all variables $v^{(j)}_{p}$ for $(j+1)$-increasing $p$,
\begin{align}\label{eq:T_j_equation}
    T_j \mathbf{v}^{(j)} = \mathbf{c}^{(j)}.
\end{align}
Afterwards, we can exclude these solved variables, and consider only $R_j$ as in Line \ref{line:pick_R_j} to solve the remaining variables, 
$v^{(j-1)}_p$ where $p$ is $j$-increasing but not $(j+1)$-increasing. Notice that the solved variables $- S_j \mathbf{v}^{(j)}$ should be added to the constant vector, and define $\overline{Q}$ as in Line \ref{line:Q_bar}. The equation associated with $R_j$ is:
\begin{align}\label{eq:R_j_equation}
        R_j\mathbf{v}^{(j-1)}(\overline{Q})=\mathbf{c}^{(j-1)}(\overline{Q}) - S_j \mathbf{v}^{(j)}.
\end{align}

$R_j$ is the sub-matrix of \eqref{eq:92} whose rows and columns are not  $(j+1)$-increasing (Line \ref{line:pick_R_j}). Since the row index $q=g_{j,k}(p)$ or $q=h_i(g_{j,k}(p)), 2 \le i \le n-j-1$ is not $(j+1)$-increasing, $R_j$ is in the same form as \eqref{eq:92}. We rewrite it such that the column is indexed by $p$,
\begin{align}\label{eq:R_j}
    R_j(q,p) =  \begin{cases}
    \sum_{i=1}^{n} (\lambda_i +\mu_i)  - \lambda_{q_{1}}, & \text{if } q=p,\\
     -\lambda_{q_1}, & \text{if } q = h_i(p), 2 \le i \le n-j-1,\\
    0, & \text{o.w.}
    \end{cases}
\end{align}
Similarly, $T_j$ is the sub-matrix of \eqref{eq:91} with $(j+1)$-increasing columns and rows. Notice that in Case 2 of \eqref{eq:91}, $Incr_j(h_i(p))$ is either equal to  $Incr_{j+1}(h_i(p))$ or not $(j+1)$-increasing. Thus we can remove the case $q=Incr_j(h_i(p))$ and obtain $T_j$, which is identical to \eqref{eq:T_j}. Thus the induction to prove C1 is completed.

{\bf Proof of Line \ref{line:sum}.}
Based on $T_{n-2}$ in \eqref{eq:T_j}, after the last iteration, one can see that the coefficient matrix is
\begin{align}
 T_{n-1} = \sum_{i=1}^{n}\mu_i, 
\end{align}
which is a scalar. From Line \ref{line:constant_vector} of Algorithm \ref{alg:heter}, we see that the constant vector 
\begin{align}
    \mathbf{c}^{(n-1)} =  \sum_{\text{all permutations } q} c_q^{(0)} 
\end{align}
is a scalar, and the corresponding variable is $\sum_{\text{all permutations } q} v_q^{(0)} $. Hence the lemma holds. 
\end{proof}
\begin{IEEEproof} (\textbf{Lemma \ref{lem:heter_nonneg}})
The lemma will follow after proving the following claim by induction.

{\bf Claim C2.} The correctness and non-negativity statements of the lemma hold when $T_0$ is parameterized by $i$, $i=0,1,\dots,n$.

Base case. When $T_0$ is parameterized by $i=0$ or $1$, it is a positive scalar defined in \eqref{eq:T_0_n1}. Therefore, the claim holds trivially.

{\bf Correctness.}
We prove that the backward path correctly finds $\mathbf{v}^{(0)}$ as the solution to $T_0\mathbf{v}^{(0)}=\mathbf{c}^{(0)}$.
Recall the variables in each iteration are listed in \eqref{eq:j_iter_variable}, and the equations are in \eqref{eq:T_j_equation} and \eqref{eq:R_j_equation}.
In Line \ref{line:sum} of Algorithm \ref{alg:heter}, we solved $v_{(1,2,\dots,n)}^{(n-1)} = \sum_{\text{all permutations $q$}}v_q^{(0)}$ using $T_{n-1}$, where $p=(1,2,\dots,n)$. We can exclude this variable and solve $v_p^{(n-2)}$ for $p$ that is $(n-1)$-increasing but not $n$-increasing using $R_{n-1}$ (to be explained in more details later). Finally, for $n$-increasing $p=(1,2,\dots,n)$, we get from \eqref{eq:j_iter_variable} that $v_p^{(n-2)} = v_p^{(n-1)} - \sum\limits_{k=2}^{n} v^{(n-2)}_{g_{n-1,k}(p)}$.
Equivalently, the equation $T_{n-2}\mathbf{v}^{(n-2)} = \mathbf{c}^{(n-2)}$ is solved.

Next, exclude $v_{p}^{(n-2)}$ for all $(n-1)$-increasing $p$. We solve $v_{p}^{(n-3)}$ for $p$ that is $(n-2)$-increasing but not $(n-1)$-increasing using $R_{n-2}$. For $(n-1)$-increasing $p$, $v_{p}^{(n-3)} = v_p^{(n-2)} - \sum\limits_{k=3}^{n} v^{(n-3)}_{g_{n-2,k}(p)}$. Equivalently, $T_{n-3}\mathbf{v}^{(n-3)} = \mathbf{c}^{(n-3)}$ is solved.
Continuing in the same manner, for odd (2-increasing) $p$, we can finally solve the sum $v_p^{(1)} = v^{(0)}_p+v^{(0)}_{g_{1,n}(p)}$, and the even variable $v^{(0)}_{g_{1,n}(p)}$ associated with $T_1$ and $R_1$, respectively. Apparently, the odd variable $v_p^{(0)}$ is solved as well. 

Next, let us explain how to solve the equations defined by $R_j$ for some $1 \le j \le n-1$. 
Let $N$ be the set defined in Line \ref{line:N}.
By \eqref{eq:R_j}, for row $q \in N$, all non-zero entries of $R_j$ are in columns $N$. Thus we can consider  the submatrix of $R_j$ whose the row and column indices $(q,p)$ are restricted by $N$ and solve the associated equations.
Comparing $R_j$ in \eqref{eq:R_j} and $T_0$ in \eqref{eq:T_0}, we see that if we substitute $\mu_{n-j-1}+\dots+\mu_n+\lambda_{n-j}+\dots+\lambda_{n}$ by $\mu_{n-j-1}$, this submatrix of $R_j$ is the same as $T_0$ parameterized by $n-j-1$. When $(c_{n-j},\dots,c_{n})$ traverses over all possible tuples, all equations defined by $R_j$ are solved as in Line \ref{line:solve_R_j} assuming the correctness induction hypothesis.

%Hence the algorithm correctly solves the equation defined by $T_0$.

{\bf Non-negativity.}
First, we show the even variables $v_p^{(0)}$ (namely, $p_{n-1} > p_n$) are non-negative, which are defined by $R_1$.
Note that the equations defined by $R_1$ can be decomposed into equations defined by $T_0$ parameterized by $n-2$, and the corresponding constant vector is non-negative.
By the non-negativity induction hypothesis, the equations have a non-negative solution. Therefore, the even variables are non-negative.

Hence, all the even variables can be excluded from the odd equations.
We need to prove the odd (2-increasing) variables are non-negative.
%Let us consider an odd row $q$ in $T_0$. From \eqref{eq:T_0}, if the column $p$ is even, then the entry of $T_0$ is either $-\lambda_{q_{1}}$ if $q = h_i(p)$, or $0$ otherwise. Knowing the even variable $v_p \ge 0$, it contributes as a non-negative constant term for the $q$-th linear equation. Thus, the even variables can be excluded from the equations. 
Instead, we prove that variable $\{v_p^{(0)}: p_n=c\}$ is non-negative, for $c=n,n-1,\dots,2$ (note that some of these variables are even).  
%We consider the submatrix of $T_0$ whose rows and columns $(q, p)$ satisfy $q_n = p_n= c$, for constant $c=n,n-1,\dots,2$.

%Induction hypothesis:
%The $p$-th variable is non-negative, where $p_n \in \{j+1,j+2,...,n\}$.

%Base case: 
Consider the $q$-th row/equation of $T_0$, $q_n=n$. 
By \eqref{eq:T_0}, its nonzero column/variable indices are
$$p \in \{ h_1^{-1}(q),\dots,h_n^{-1}(q) \},$$
all of which has $p_n=n$ except $p=h_n^{-1}(q)=(q_2,\dots,q_{n-1},n,q_1)$. However, $p=h_n^{-1}(q)$ corresponds to an even variable and is excluded from the equation. 
As a result, we are left with only columns in the set $\{p:p_n=n\}$. 
Furthermore, these equations correspond to the submatrix $T_{sub}$ of $T_0$ whose rows and columns have $q_n=p_n=n$, 
\begin{align} \label{eq:S}
    T_{sub}(q,p) =  \begin{cases}
\sum_{i=1}^{n} (\lambda_i +\mu_i)  - \lambda_{q_{1}}, & \text{if } p=q \\
-\lambda_{q_{1}}, & \text{if } q = h_i(p),i=2,3,\dots,n-1\\
0 , & \text{o.w.}
\end{cases}
\end{align}
We see that if we substitute $\mu_{n-1}+\mu_n+\lambda_n$ by $\mu_{n-1}$, the submatrix is equal to $T_0$ parameterized by $n-1$. By the non-negativity hypothesis, the variables $\{v_p^{(0)}:p_n=n\}$ are all non-negative. From here on we exclude these variables.

Next, consider the $q$-th row/equation of $T_0$, $q_n=n-1$. 
By \eqref{eq:T_0}, its nonzero column/variable indices are 
$$p \in \{ h_1^{-1}(q),\dots,h_n^{-1}(q) \}.$$
all of which has $p_n=n-1$ except $p=h_n^{-1}(q)=(q_2,\dots,q_{n-1},c,q_1)$.\\
1. If $n-1 > q_1$, $p=h_n^{-1}(q)$ corresponds to an even variable and is excluded from the equation. \\
2. If $n-1 < q_1 = p_n$, variable $p$ is already solved and excluded.  
Hence we can again form a submatrix $T_{sub}$ that is in the same form as \ref{eq:S} restricted to rows and columns with $q_n=p_n=n-1$. The corresponding variables are non-negative and then excluded.

Continue in a similar manner, for equations $\{q: q_n=c\}$, $c = n, n-1, \dots, 2$, we can obtain non-negative solutions to $\{v_p^{(0)}: p_n=c\}$.
The proof is completed.
% Induction step:
% Assume the $p$-th variable is non-negative, where $p_n \in \{c+1,c+2,...,n\}$. We will show if $p_n=c$, variable $p$ is non-negative.
% Consider the $q$-th row/equation of $T_0$, $q_n=c$. 
% By \eqref{eq:T_0}, its nonzero column/variable indices are
% $$p \in \{ h_1^{-1}(q),\dots,h_n^{-1}(q) \}.$$
% all of which has $p_n=c$ except $p=h_n^{-1}(q)=(q_2,\dots,q_{n-1},c,q_1)$.\\
% 1. If $c > q_1$, $p=h_n^{-1}(q)$ corresponds to an even variable and is excluded from the equation. \\
% 2. If $c < q_1 = p_n$, variable $p$ is already solved and is non-negative. The corresponding coefficient is $-\lambda_{q_1} <0$, and it contributes a positive term to the constant vector. Thus we can exclude variable $p$. 
% Furthermore, these equations correspond to the submatrix $S$ of $T_0$ whose rows and columns have $q_n=p_n=c$, and $S$ is the same format as \eqref{eq:S}.
% Again, $S$ is equal to $T_0$ parameterized by $n-1$. By Fact 2, the variables $\{p:p_n=c\}$ are all non-negative.
\end{IEEEproof}
\begin{proof} (\textbf{Theorem \ref{thm:heter_AoI}})
In Lemma \ref{gen_mat} we found the format of $T$ and $\bm{\pi}$ for the transition equations. By Lemma \ref{lem:yates}, if these equations have a non-negative solution, then AoI is calculated as in  \eqref{eq:heter_AoI}. In Lemma \ref{break_gen}, the $(n+1)!$ equations are broken down into smaller sets of equations by Algorithm \ref{alg:heter_overall}, all of which have coefficient matrix in the form of $T_0$ parameterized by various numbers. Therefore, each set of equations (Line \ref{line:breakdown} of Algorithm \ref{alg:heter_overall}) can be solved by calling Algorithm \ref{alg:heter}. The solutions should be non-negative according to Lemma \ref{lem:heter_nonneg}, resulting in  non-negative constant vector for the remaining equations (Line \ref{line:constant2} of Algorithm \ref{alg:heter_overall}). Finally, the AoI (Line \ref{line:AoI} of Algorithm \ref{alg:heter_overall}) can be computed by Algorithm \ref{alg:heter} according to Lemma \ref{lem:alg2_output}.
\end{proof}

\bibliographystyle{IEEEtran}
\bibliography{biblio} 
%%%%%%
%% Appendix:
%% If needed a single appendix is created by
%%
%\appendix
%%
%% If several appendices are needed, then the command
%%
% \appendices
%%
%% in combination with further \section-commands can be used.
%%%%%%

%%%%%%
%% To balance the columns at the last page of the paper use this
%% command:
%%
%\enlargethispage{-1.2cm} 
%%
%% If the balancing should occur in the middle of the references, use
%% the following trigger:
%%
\IEEEtriggeratref{3}

\end{document}